\documentclass[a4paper,11pt]{article} 
  
\usepackage[english]{babel}
\usepackage[a4paper]{geometry}
\usepackage{amsmath}
\usepackage{amsthm}
\usepackage{amssymb}
\usepackage{bbm}
\usepackage{color}
\usepackage{enumerate}
\usepackage{tasks}

\usepackage{hyperref}         
\usepackage{tikz}
\usetikzlibrary{fit}
\usetikzlibrary{shapes.geometric}
\usetikzlibrary{decorations.pathmorphing}

\tikzset{
point/.style={circle,fill=black,inner sep=1pt},
vertex/.style={circle,fill=black,inner sep=1.5pt},   
bvertex/.style={circle,fill=black,inner sep=2.8pt},
Bvertex/.style={circle,fill=black,inner sep=4pt}, 
specialEP/.style={rectangle,fill=white,draw,inner sep=3pt},  
whitevex/.style={circle,fill=white,draw, inner sep=2pt},
linelabel/.style={sloped,above,very near start, inner sep=1pt,execute at begin node=$\scriptstyle,execute at end node=$},
baseline=(current  bounding  box.center),doubled/.style={double distance= 1pt,line width=1.5pt},
th/.style={line width=0.5 pt, gray},  
med/.style={line width=1 pt}  
}

\xdefinecolor{myred}{rgb}{0.7,0,0.2} 
\xdefinecolor{mypurple}{rgb}{0.6,0.2,0.4} 
\xdefinecolor{myblue}{rgb}{0.259,0.2,0.6}   
\xdefinecolor{myorange}{rgb}{1,0.6,0.2}   
\definecolor{orange}{rgb}{1,0.5,0}
\xdefinecolor{purpleish}{cmyk}{0.75,0.75,0,0}

\def\bR{\mathbb{R}}

\def\bN{\mathbb{N}}
\def\NN{\mathbb{N}}
\def\bZ{\mathbb{Z}}
\def\cC{\mathcal{C}}

\def\cQ{\mathcal{Q}}
\def\cD{\mathcal{D}}

\def\cM{\mathcal{M}}

\def\cV{\mathcal{V}}
\def\cO{\mathcal{O}}
\def\cF{\mathcal{F}}
\def\cG{\mathcal{G}}
\def\cL{\mathcal{L}}
\def\cJ{\mathcal{J}}
\def\cN{\mathcal{N}}
\def\cE{\mathcal{E}}
\def\cK{\mathcal{K}}
\def\cH{\mathcal{H}}

\def\eps{\varepsilon}
\def\ph{\varphi}

\def\wt{\widetilde}

\def\indic{\hbox{\raise-2pt \hbox{\indbf 1}}}

\let\==\equiv

\let\io=\infty
\let\0=\noindent

\def\*{{\hfill\break\null\hfill\break}}
\def\bra#1{{\langle#1|}}
\def\ket#1{{|#1\rangle}}

\def\expec#1#2{\langle{#1},#2{#1}\rangle}

\let\arr=\rightarrow

\def\norm#1{{\left|\hskip-.05em\left|#1\right|\hskip-.05em\right|}}

\def\tende#1{\,\vtop{\ialign{##\crcr\rightarrowfill\crcr
             \noalign{\kern-1pt\nointerlineskip}
             \hskip3.pt${\scriptstyle #1}$\hskip3.pt\crcr}}\,}
\def\otto{\,{\kern-1.truept\leftarrow\kern-5.truept\to\kern-1.truept}\,}
\def\fra#1#2{{#1\over#2}}

\newtheorem{theorem}{Theorem}[section]  

\newtheorem{prop}[theorem]{Proposition}
\newtheorem{lemma}[theorem]{Lemma}

\numberwithin{equation}{section}

\def\be{\begin{equation}}
\def\ee{\end{equation}}
\def\spl#1{\[ \begin{split}#1\end{split} \]}
\def\bes#1{\be \begin{split}#1\end{split} \ee}  
\newcommand{\hc}{\mathrm{h.c.}}

\let\a=\alpha \let\b=\beta    \let\g=\gamma     \let\d=\delta     \let\e=\varepsilon
  \let\h=\eta     \let\th=\vartheta \let\k=\kappa     \let\l=\lambda  \let\ka=\kappa
    \let\n=\nu      \let\x=\xi        \let\p=\pi        \let\r=\rho
\let\s=\sigma \let\t=\tau         \let\ph=\varphi   \let\c=\chi
        
\let\G=\Gamma \let\D=\Delta   \let\Th=\Theta    \let\L=\Lambda    \let\X=\Xi
\let\P=\Pi             
\let\O=\Omega 
\let\Y=\Upsilon 

\newcommand{\fa}{ \mathfrak a}
\def\sl{\fa_0}
\newcommand{\ret}{\Lambda^*}
\newcommand{\retp}{\ret_+}
\def\quadre#1{\left[#1\right]}

\def\wh#1{\widehat{#1}}
\def\excsp{\cF^{\le N}_+}

\def\gatn#1{\be \begin{gathered}#1\end{gathered} \ee}
\def\abs#1{\left|#1\right|}
\newcommand{\st}{\,:\,}
\def\cka{\check a}
\def\ckb{\check b}
\def\ckd{\check d}
\def\cke{\check \eta}

\def\ckg{\check \g}
\def\cks{\check \s}

\def\consB#1#2{e^{-#1B(\t)}#2e^{#1B(\t)}}
\def\conA#1{e^{-A}#1e^A}
\def\consA#1{e^{-sA}#1e^{sA}}
\def\tg{\wt\g}
\def\ts{\wt\s}
\def\convo#1{(\wh V(\cdot/N^{1-\ka})\star#1)}
\def\cNN{(\cN_++N^{3\ka/2})}
\def\cKK{(\cK+N^{\a+2\ka})}

\definecolor{lightblue}{rgb}{0, 0.33, 0.71}

\begin{document}

\title{Excitation Spectrum for Bose Gases beyond the Gross--Pitaevskii Regime} 

\author{Christian Brennecke$^1$, Marco Caporaletti$^2$, Benjamin Schlein$^3$ \\
\\
Department of Mathematics, Harvard University, \\
One Oxford Street, Cambridge MA 02138, USA$^{1}$ \\
\\
Institute of Mathematics, University of Zurich, \\
Winterthurerstrasse 190, 8057 Zurich, Switzerland$^{2,3}$}

\maketitle

\begin{abstract}
	We consider Bose gases of $N$ particles in a box of volume one, interacting through a repulsive potential with scattering length of order $N^{-1+\ka}$, for $\ka>0$. Such regimes interpolate between the Gross-Pitaevskii and thermodynamic limits. Assuming that $\ka$ is sufficiently small, we determine the ground state energy and the low-energy excitation spectrum, up to errors vanishing in the limit $N\to\io$. 
\end{abstract}


\section{Introduction and main result}
\label{sec:intro}

We consider systems of $N\in\NN$ bosons in the three-dimensional box $\Lambda=[-1/2,1/2]^3$, with periodic boundary conditions, interacting through a repulsive potential with effective range of the order $N^{-1+\kappa}$, for $\kappa \geq 0$ small enough. The Hamilton operator of the system has the form 
\be
\label{eq:hamiltoniano}
H_N=\sum_{i=1}^N-\Delta_{x_i}+\sum_{1\le i<j\le N}N^{2-2\ka}V(N^{1-\ka}(x_i-x_j))
\ee
and it is densely defined on $L^2_s(\Lambda^N)$, the subspace of $L^2 (\Lambda^N)$ consisting of permutation-invariant functions on $\Lambda^N$. Throughout the paper we will assume that the potential $V\in L^3(\bR)$ is non-negative, radial and compactly supported. We recall the zero-energy scattering equation 
\[
\left[-\D+\fra12V(x)\right]f(x)=0,
\]  with the boundary condition $f(x)\to1$ as $\abs{x}\to\infty$. It is easy to see that, outside of the range of the potential $V$, the solution must have the form $f=1-\sl/\abs x$ for some unique positive real number $\sl$ called the scattering length of $V$. By scaling, 
\[
\left[-\D+\fra12N^{2-2\ka}V(N^{1-\ka}x)\right]f(N^{1-\ka}x)=0,
\] so that the scattering length of the rescaled potential is given by $\fa_N=\fa_0/N^{1-\ka}$. For $\ka<2/3$ the Hamiltonian \eqref{eq:hamiltoniano} represents a ``dilute limit'', in the sense that the rescaled density $\rho \fa_N^3 = N^{-2+3\ka}$ tends to zero, as $N \to \infty$. 

For $\ka=0$, (\ref{eq:hamiltoniano}) describes a Bose gas in the so-called Gross-Pitaevski regime, for which several rigorous results have been obtained, including convergence of the many-body dynamics towards the time-dependent Gross-Pitaevskii equation \cite{ESY,P,BDS,BS}, complete Bose-Einstein condensation \cite{LS,LS2,NRS,BBCS1,BBCS4,NNRT,H,BSS}, bounds on the ground state energy to leading order \cite{LSY,NRS} and, recently, up to errors vanishing in the limit of large $N$ \cite{BBCS3}, as well as estimates on the low-energy excitation spectrum \cite{BBCS3}. For $\kappa = 2/3$, the Hamilton operator (\ref{eq:hamiltoniano}) describes instead (after appropriate rescaling) a gas in the thermodynamic limit, at fixed density. In this regime, only the ground state energy has been determined rigorously, at leading order \cite{LY} and, recently, including next-order corrections \cite{FS}. In this regime, the existence of Bose-Einstein condensation and the form of the excitation spectrum are still open and in fact they represent very challenging goals in mathematical physics (a review of a long-term project based on a renormalization group approach to these questions can be found in \cite{Balaban}). 

In this paper, we choose $\ka \in (0;2/3)$ small enough, interpolating between the Gross-Pitaevskii and the thermodynamic limit. Our main theorem extends precise estimates on the ground state energy and on low-lying excitations from the Gross-Pitaevskii limit to regimes with $\kappa > 0$. 
\begin{theorem}\label{thm:main}
	Let $V\in L^3(\bR)$ be pointwise nonnegative, spherically symmetric and compactly supported  and denote by $\fa_0$ its scattering length, and let $E_N$ be the ground state of the Hamiltonian $H_N$  defined in \eqref{eq:hamiltoniano}. For $\kappa \geq 0$ small enough, there exists $\eps > 0$ such that, in the limit $N \to \infty$, 	
	\bes{\label{eq:mainthm_GS}
E_N= \; &4\pi\sl N^{\ka} (N-1)+e_\L(\fa_0N^\ka)^2\\
	&+\fra{1}{2}\sum_{p\in\retp}\quadre{\sqrt{|p|^4+16\pi\sl p^2N^\ka}-p^2-8\pi\sl N^\ka+\fra{(8\pi\sl N^\ka)^2}{2p^2}}+ \cO (N^{-\eps}) \,,
	} with the notation $\retp=2\pi\bZ^3\setminus\{0\}$ and	
	\begin{equation}\label{eq:eLambda} e_\L=2-\lim_{M\to\io}\sum_{\substack{p\in\bZ^3_+\\ \abs{p_1},\abs{p_2},\abs{p_3}\le M}}\fra{4\cos(\abs p)}{\abs p^2}\,,\end{equation}  
where, in particular, the limit exists. Moreover, for if $\kappa , \mu > 0$ are small enough, there exists $\eps > 0$ such that the spectrum of $H_N-E_N$ below the threshold $N^{\ka/2 + \mu}$ consists of eigenvalues given, in the limit $N\to\io$, by
\be\label{eq:mainthm_eigenvalues}
\sum_{p\in\retp}n_p\sqrt{|p|^4+16\pi\fa_0  N^\ka p^2} + \cO (N^{-\eps}) \, ,
\ee where $n_p\in\bN$ for all $p\in\retp$ and $n_p\neq0$ only for a finite number of $p\in\retp$.
\end{theorem}

{\it Remark:} Approximating the sum on the r.h.s. of (\ref{eq:mainthm_GS}) with an integral, we find  
\[ \begin{split}  
E_N &= 4\pi a_0 N^{1+\kappa} + \frac{N^{5\kappa/2}}{2 (2\pi)^3} \int \left[ \sqrt{|p|^4 + 16 \pi \fa_0 p^2} - p^2 - 8\pi \fa_0 + \frac{(8\pi \fa_0)^2}{2p^2} \right] dp + \cO (N^{2\kappa}) \\ &= 4 \pi a_0 N^{1+\kappa} + 4 \pi \cdot \frac{128}{15 \sqrt{\pi}} \fa_0^{5/2} N^{5\kappa/2} + \cO (N^{2\kappa}) 
\end{split} \]
which is consistent with the Lee-Huang-Yang formula, derived in \cite{FS} for $\kappa = 2/3$ (after appropriate rescaling). The expansion (\ref{eq:mainthm_GS}) is more precise, since it identifies all contributions to the ground state energy that do not vanish as $N \to \infty$ (but, of course, it only 
holds for $\kappa > 0$ small enough). 

\medskip

{\it Remark:} Eq. (\ref{eq:mainthm_eigenvalues}) implies that low-lying excited eigenvalues of the Hamiltonian $H_N$ have energy of order $N^{\kappa/2}$ (for this reason, the interval $[0; N^{\kappa/2+\mu}]$ considered in (\ref{eq:mainthm_eigenvalues}) contains several eigenvalues of the operator $H_N$, for any $\mu > 0$, in the limit of large $N$). In fact, (\ref{eq:mainthm_eigenvalues}) also shows that, to leading order, the dispersion of the excitations is linear in $|p|$. By Landau's criterion, this gives a heuristic explanation of the emergence of superfluidity \cite{Lan}. 

\medskip

{\it Remark:} The bounds (\ref{eq:mainthm_GS}), (\ref{eq:mainthm_eigenvalues}) establish the validity of the predictions of Bogoliubov theory \cite{B} for $\kappa > 0$ small enough, extending the results obtained in \cite{BBCS3} for $\kappa = 0$ (and previous results obtained in \cite{S,GS,LNSS,DN,P3,BPS} for mean-field bosons and in \cite{BBCS2} for regimes interpolating between the mean-field and the Gross-Pitaevskii limit).

\medskip

The proof of Theorem \ref{thm:main} follows the strategy developed in \cite{BBCS4} for the Gross-Pitaevskii regime with $\kappa = 0$. It makes use, crucially, of recent estimates obtained in \cite{ABS}, for sufficiently small $\kappa > 0$, for the expectation of the number and the energy of excitations of the condensate (with a rate that becomes optimal as $\kappa \to 0$). 

First of all, we factor out the Bose-Einstein condensate, focusing on its orthogonal excitations. To this end, we introduce a unitary operator $U_N$, mapping the Hilbert space $L^2_s (\Lambda^N)$ into the truncated Fock space $\cF_+^{\leq N}$, constructed on the orthogonal complement of the condensate wave function. This leads us to an excitation Hamiltonian $\cL_N = U_N H_N U_N^*$, acting on $\cF_+^{\leq N}$. 

As a second step, we conjugate $\cL_N$ with a (generalized) Bogoliubov transformation, defining the renormalized excitation Hamiltonian $\cG_N = e^{-B} \cL_N e^{B}$. The antisymmetric operator $B$ is quadratic in (modified) creation and annihilation operators. Conjugation with $e^B$ acts on high momenta; it creates  short-scale correlations among particles. Even after conjugation with $e^B$, there are still some important contributions to the energy  hidden in parts of $\cG_N$ that are cubic and quartic in creation and annihilation operators. For this reason, we need to conjugate $\cG_N$ with another unitary operator of the form $e^A$, this time with $A$ cubic, rather than quadratic, in creation and annihilation operators. We obtain the twice renormalized Hamiltonian $\cJ_N = e^{-A} \cG_N e^A$. Up to error terms that can be estimated with the a-priori bounds from \cite{ABS}, $\cJ_N$ is the sum of a constant term, a quadratic and a quartic contribution. The quartic part is positive and can be neglected, when proving lower bounds on the ground state energy and on the excited eigenvalues. To show matching upper bounds, we only have to control the quartic potential on appropriate trial states. To conclude the proof of Theorem \ref{thm:main}, we still need to diagonalize the quadratic part of $\cJ_N$. To this end, we use a last (generalized) Bogoliubov transformation to define the final excitation Hamiltonian $\cM_N = e^{-T} \cJ_N e^T$, with a diagonal quadratic part. 

While the previous steps of the analysis are quite similar to what was done in \cite{BBCS3} for $\kappa = 0$, this last step differs substantially. The reason is that, for $\kappa > 0$, the (generalized) Bogoliubov transformation $e^T$ diagonalizing the quadratic component of $\cJ_N$ creates a large number of excitations (the number of excitations of the Bose-Einstein condensate is of the order $N^{3\kappa/2}$, and thus diverges, as $N \to \infty$, for any $\kappa > 0$). For this reason, it is more difficult to determine the action of $e^T$ (being $e^T$ a generalised Bogoliubov transformation, its action is not explicit; it has to be computed through an expansion, whose convergence depends on the size of $T$) and to control the resulting growth of error terms. This part of the analysis, which is carried out in Section \ref{sec:diagonalization} and leads to the proof of Theorem \ref{thm:main} in Section \ref{sec:main_thm}, is the main novelty of our work. Together with the a-priori bounds on the number and the energy of excitations of the Bose-Einstein condensate from \cite{ABS}, this is also the main reason why we have to restrict our analysis to very small values of the parameter $\kappa > 0$. While we did not try to optimize the choice of $\kappa$, it is clear that to make it substantially larger and to approach the thermodynamic limit at $\kappa =2/3$, genuinely new ideas are needed. In this direction, let us mention the recent result obtained in \cite{F}, where the techniques of \cite{FS} have been applied to show the existence of Bose-Einstein condensation up to $\kappa < 2/5$; compared with the bounds of \cite{ABS}, this approach only controls the expectation of the number of particles and therefore it cannot be directly applied to show Theorem \ref{thm:main}. It would in fact be possible to combine the estimate in \cite{F} with the bounds in \cite[Proposition 3.3]{ABS} to get a better control of energy and high powers of number of excitations. Since, however, we would still have to restrict to very small values of $\kappa$, we decided to use directly the bounds in \cite{ABS}, to keep our analysis as short as possible. 

\medskip

\textit{Acknowledgements.} B. S. gratefully acknowledges partial support from the NCCR SwissMAP, from the Swiss National Science Foundation through the Grant ``Dynamical and energetic properties of Bose-Einstein condensates'' and the grant ``Bogoliubov theory for bosonic systems'' and from the European Research Council through the ERC-AdG CLaQS. 



\section{Fock Space and the Excitation Hamiltonian}
\label{sec:exc_ham}

We describe excitations of the Bose-Einstein condensate on the truncated Fock space 
\[ \cF_+^{\leq N} = \bigoplus_{k=0}^N L^2_\perp (\Lambda)^{\otimes_s k} \]
built over the orthogonal complement $L^2_\perp (\Lambda)$ of the zero-momentum mode $\ph_0 \equiv 1$. We map the original $N$-particle Hilbert space $L^2_s (\Lambda^N)$ into $\cF_+^{\leq N}$ through the unitary operator $U_N : L^2_s (\L^N) \to \cF_+^{\leq N}$, defined by $U_N \psi_N = \{ \alpha_0, \alpha_1, \alpha_2, \dots , \alpha_N \}$, with $\alpha_j \in L^2_\perp (\Lambda)^{\otimes_s j}$, if 
\[ \psi_N = \alpha_0 \ph_0^{\otimes N} + \alpha_1 \otimes_s \ph_0^{\otimes_s (N-1)} + \dots + \alpha_N \]
where $\otimes_s$ denotes the symmetrized tensor product. The map $U_N$ factors out the Bose-Einstein condensate and allows us to focus on its orthogonal excitations. With $U_N$, we can define the excitation Hamiltonian $\cL_N = U_N H_N U_N^*$ as a self-adjoint operator on a dense subspace of $\cF_+^{\leq N}$. Proceeding as in \cite[Sect. 2]{ABS}, we find 
\[
\cL_N = \cL_N^{(0)}+\cL_N^{(2)}+\cL_N^{(3)}+\cL_N^{(4)}\,,
\] with 
\bes{\label{eq:excitation_hamiltonian}
	\cL_N^{(0)}=\,& \fra{N-1}{2N}N^\ka\wh{V}(0)(N-\cN_+)+\fra{N^\ka\wh{V}(0)}{2N}\cN_+(N-\cN_+)\\
	\cL_N^{(2)}=\,&\sum_{p\in\retp}p^2a_p^*a_p+\sum_{p\in\retp}N^\ka\wh V(p/N^{1-\ka})\left[b_p^*b_p-\fra{1}{N}a_p^*a_p\right]\\&+\fra{N^\ka}{2}\sum_{p\in\retp}\wh V(p/N^{1-\ka})[b_p^*b_{-p}^*+b_pb_{-p}]\\
	\cL_N^{(3)}=\,&\fra{N^\ka}{\sqrt N}\sum_{\substack{p,q\in\retp\\p\neq-q}}\wh V(p/N^{1-\ka})\left[b_{p+q}^*a_{-p}^*a_q+a_q^*a_{-p}b_{p+q}\right]\\
	\cL_N^{(4)} =\; &\fra{N^\ka}{2N}\sum_{\substack{p,q\in\retp, r\in\ret\\r\neq-p,-q}}\wh V(r/N^{1-\ka})a_{p+r}^*a_q^*a_{p}a_{q+r}\,,
} where $\L_+^* = 2\pi \bZ^3 \backslash \{ 0 \}$ is the set of possible momenta of excitations of the condensate and where, for any $p \in \L^*_+$, $a^*_p, a_p$ are the usual creation and annihilation operators. Moreover, for $p \in \L^*_+$, we introduced generalized creation and annihilation operators 
\[ \begin{split} b_p^* &= U_N a_p^* \frac{a_0}{\sqrt{N}} U_N^* = a_p^* \sqrt{\frac{N- \cN_+}{N}} \\
b_p &= U_N \frac{a_0^*}{\sqrt{N}} a_p U_N^* = \sqrt{\frac{N- \cN_+}{N}}  a_p \end{split}  \]
satisfying the approximate CCR
\be\label{eq:bpCCR}
[b_p,b_q]=[b_p^*,b_q^*]=0,\qquad [b_p,b_q^*]=\d_{p,q}\left(1-\fra{\cN_+}{N}\right)-\fra{1}{N}a_q^*a_p,
\ee together with the useful commutation relations
\[
[b_p, a^*_qa_r]=\d_{p,q}b_r,\qquad [b_p^*, a_q^*a_r]=-\d_{p,r}b_q^*.
\]
Observe that, while $a_p, a_p^*$ do not preserve the truncation on the number of particles, the operators $b_p, b_p^*$ (and also products of the form $a_p^* a_q$) are well-defined on $\cF_+^{\leq N}$. 

In the following, we will also use the notation \begin{equation}\label{eq:KVN} \cK = \sum_{p \in \L^*_+} p^2 a_p^* a_p, \qquad \cV_N = \cL_N^{(4)} =\; \fra{N^\ka}{2N}\sum_{\substack{p,q\in\retp, r\in\ret\\r\neq-p,-q}}\wh V(r/N^{1-\ka})a_{p+r}^*a_q^*a_{p}a_{q+r} \end{equation}  for the kinetic and the potential energy operators. Moreover, we set $\cH_N = \cK + \cV_N$. It will sometimes also be useful to switch to position space, introducing, for $x \in \Lambda$, operator-valued distributions $\check{a}_x, \check{a}^*_x, \check{b}_x, \check{b}^*_x$ and writing the interaction as 
\begin{equation}\label{eq:cV-x} \cV_N = \frac12\int dx dy \, N^{2-2\kappa}  V (N^{1-\kappa} (x-y)) \check{a}_x^* \check{a}_y^* \check{a}_y \check{a}_x \end{equation} 



\section{Renormalized Hamiltonians}
\label{sec:renorm_ham}

The vacuum expectation of $\cL_N$ still differs from its ground state (hence that of $H_N$) to the leading order $N^{1+\ka}$. This is because correlations among particles still carry an energy of order $N^{1+\kappa}$ in low-energy states. To extract the relevant contributions to the energy from the cubic and quartic terms, we are going to conjugate $\cL_N$ with a generalized Bogoliubov transformation. To choose the appropriate Bogoliubov transformation, we consider, for a fixed $\ell \in(0;1/2)$, the ground state $f_\ell$ of the Neumann problem
\be \label{eq:scattering_f}
\left[-\D+\fra{1}{2}V\right]f_\ell = \l_\ell f_\ell
\ee
on $B_{\ell N^{1-\kappa}}$ with boundary condition $f_\ell (x)=1$ for $\abs x = \ell N^{1-\kappa}$. By scaling, $f_\ell (N^{1-\ka}\cdot)$ satisfies the equation 
\[
\left[-\D+N^{2-2\ka}\fra{1}{2}V(N^{1-\ka}\cdot)\right]f_\ell (N^{1-\ka}\cdot) = N^{2-2\ka}\l_\ell f_\ell (N^{1-\ka}\cdot)\] on $B_\ell$. We now define $f_N$ to be the extension of $f_\ell (N^{1-\ka}\cdot)$ to $\Lambda$, obtained by setting $f_N (x) = 1$ if $x \in \Lambda \backslash B_\ell$.  With $\c_\ell$ denoting the characteristic function of $B_\ell$ we have
\be\label{eq:scattering_f_rescaled}
\left[-\D+N^{2-2\ka}\fra{1}{2}V(N^{1-\ka}\cdot)\right] f_N = N^{2-2\ka}\l_\ell \c_\ell f_N\,.
\ee
We further define $w_\ell = 1- f_\ell$ on $B_{\ell N^{1-\kappa}}$ and its rescaled version $w_N = w_\ell (N^{1-\ka}\cdot) = 1 - f_N$ on $\Lambda$, with Fourier coefficients 
\[
\wh{w}_N(p)=\int_\Lambda e^{-ip\cdot x} w_N(x) dx=\fra{1}{N^{3-3\ka}}\wh w_\ell (p/N^{1-\ka})=\d_{p,0}-\wh f_N(p)\,.
\] 
Some important properties of these functions are collected in the next lemma, whose proof is a straightforward adaptation of \cite[Lemma 3.1]{BBCS3}. 
\begin{lemma} \label{lem:properties_scattering_function}
Fix $\ell  \in (0;1/2)$ and denote by $f_\ell$ the solution of \eqref{eq:scattering_f}. For $N\in\bN$ large enough the following properties hold true: 
	\begin{enumerate}[i)]
		\item We have 
		\be\label{eq:scattering_lambda}
		\l_\ell=\fra{3\fa_0}{(\ell N^{1-\ka})^3}\left(1+\fra{9}{5}\fra{\fa_0}{\ell N^{1-\ka}}+\cO\left(\fra{\fa_0^2}{(\ell N^{1-\ka})^2}\right)\right)\,.
		\ee
		\item We have $0\le f_\ell$, $w_\ell \le 1$, and there exists a constant $C>0$ such that 
		\be \label{eq:approx_8pia_0}
		\abs{\int V(x)f_\ell (x)dx - 8\pi\fa_0\left(1+\fra{3}{2}\fra{\fa_0}{\ell N^{1-\ka}}\right)}\le\fra{C\fa_0^3}{(\ell N^{1-\ka})^2}\,.
		\ee
		\item There exists a constant $C>0$ such that 
		\be\label{eq:wl_bounds}
		0\le w_\ell \le\fra{C}{1+\abs{x}}\qquad\mbox{and}\qquad\abs{\nabla w_\ell}\le\fra{C}{1+\abs x^2}\,,\ee 
and precisely
\be\label{eq:wl_integral}
\abs{\fra{1}{(\ell N^{1-\ka})^2}\int_{\bR^3}w_\ell (x)dx-\fra{2}{5}\pi\fa_0}\le \fra{C\fa_0}{\ell N^{1-\ka}}\,.
\ee
\item There exists a constant $C>0$ such that
\be \label{eq:bound_wp}
\wh w_N(p)\le\fra{CN^\ka}{N\abs p^2}
\ee for all $p\in\bR^3$ and $N\in\bN$ large enough (such that $N^\ka/N\le 1/\ell$).
\end{enumerate}
\end{lemma}

We define now $\eta:\retp\arr\bR$ through
\[
\eta_p = - N \wh w_N(p) =-\fra{N^\ka}{N^{2-2\ka}}\wh w_\ell (p/N^{1-\ka})\,,
\]
In particular, $\wh f_N(p)=\d_{p,0}+\eta_p / N$. In position space, we have 
\[ 
\check\eta(x)=-Nw_\ell (N^{1-\ka}x)
\] for $x\in\L$. From Lemma \ref{lem:properties_scattering_function}, and in particular Eq. \eqref{eq:wl_bounds}, \eqref{eq:bound_wp}, we obtain  
\be\label{eq:eta_punctual_bound}
\eta_p\le\fra{CN^\ka}{p^2}\,,\qquad \norm{\eta}_2\le CN^\ka \, , \qquad \norm{\check\eta}_\io\le CN
\ee
Moreover, from \eqref{eq:wl_bounds} we deduce 
\[
\norm{\eta}_{H^1}^2 = N^{4-2\ka}\norm{\nabla w_\ell (N^{1-\ka}\cdot)}_2^2\le \int_{B_\ell}dx\fra{CN^{4-2\ka}}{(1+ N^{2-2\kappa}\abs x^2)^2}\le CN^{1+\ka}.
\]
The scattering equation \eqref{eq:scattering_f_rescaled} translates to Fourier space into 
\be \label{eq:scattering_fN_fourier}
p^2\wh f_N(p)+\fra{N^\ka}{2N}(\wh V(\cdot /N^{1-\ka})\star\wh f_N)_p=N^{2-2\ka}\l_\ell (\wh \c_\ell \star \wh f_N)_p\,,
\ee for $p\in\retp$ or, equivalently, 
\bes{\label{eq:scattering_eta}
	p^2\eta_p+\fra{1}{2}N^\ka\wh V(p/N^{1-\ka})+\fra{1}{2N}&\sum_{q\in\ret}N^\ka\wh V((p-q)/N^{1-\ka})\eta_q\\&=N^{3-2\ka}\l_\ell \wh \c_\ell (p)+N^{2-2\ka}\l_\ell \sum_{q\in\ret}\wh \c_\ell (p-q)\eta_q.}

In order to make the $\ell^2$-norm of $\eta$ small (which is important to control the action of the corresponding generalized Bogoliubov transformation), we introduce an infrared cutoff, restricting $\eta$ to high momenta. For an $\alpha > 0$ to be specified below, we define $P_H:=\{p\in\retp\st\abs{p}\ge N^\a\}$ and the coefficients 
\[
\eta_H(p)=\eta_p\c_{P_H}(p)\,.
\]
Then, we have 
\spl{\norm{\eta_H}_2&\le CN^{\ka-\alpha/2}, \hspace{0.5cm}\norm{\eta_H}_{H^1}\le CN^{(1+\ka)/2}, \hspace{0.5cm}\norm{\eta_H}_\io\le CN^{\ka-\a}\,,
} and moreover, from Lemma \ref{lem:properties_scattering_function} and \eqref{eq:eta_punctual_bound}, 
\[
\abs{\cke_H(x)}=\Big| \sum_{\abs p \ge N^\a}e^{ip\cdot x}\eta_p \Big| \le\abs{\cke(x)}+\sum_{\abs p < N^\a}\abs{\eta_p}\le C(N+N^{\ka+\alpha})\,.
\] We assume throughout the following analysis that $\alpha>2\ka,\,\ka+\alpha<1$ so that 
\gatn{\label{eq:etaH_norms}\norm{\eta_H}_2\arr 0\,, \norm{\eta_H}_\io\arr 0\, \quad\mbox{as }N\arr\io\,,\hspace{1cm}\norm{\cke_H}_\io\le CN\,.	
} 
With the coefficients $\eta_H$, we define the antisymmetric operator
\begin{equation}\label{eq:B-def} B = \frac{1}{2} \sum_{p \in \L^*_+} \eta_H (p) \big[ b_p^* b_{-p}^* - b_p b_{-p} \big] = \frac{1}{2} \sum_{p \in \Lambda^*_+ : |p| > N^\alpha} \eta_p \big[ b_p^* b_{-p}^* - b_p b_{-p} \big] \end{equation} 
and we consider the corresponding generalized Bogoliubov transformation $e^B$. With (\ref{eq:etaH_norms}), we can control the action of $e^B$ on powers of the number of particles operator. The proof of the following lemma can be found in \cite[Lemma 3.1]{BS} (see \cite{BBCS1} for the analogue in the translation invariant setting).
\begin{lemma}\label{lem:B_bound_N_easy}
Assume $\alpha > 2\kappa$. For every $j\in\bN$ there exists a constant $C>0$ such that 
\[
e^{-B}(\cN_++1)^je^{B}\le C (\cN_++1)^j\,.
\]
\end{lemma} 
On states with few excitations, the generalized Bogoliubov transformation 
$e^B$ acts approximately like a standard Bogoliubov transformation. To make this statement more precise, for $p \in \Lambda^*_+$ we set $\g_p=\cosh (\h_H (p))$, $\s_p= \sinh  (\h_H (p))$ and we define operators $d_p, d_p^*$ through the identities 
\bes{ \label{eq:dp_def}
		e^{-B}b_pe^{B}&=\g_pb_p+\s_pb_{-p}^*+d_p\\
		e^{-B}b_p^*e^{B}&=\g_pb_p^*+\s_pb_{-p}+d_p^*.}
 Observe that $\gamma_p = 1$ and $\s_p = 0$ for $p \in P_H^c$. For $p \in P_H$, on the other hand, we can use \eqref{eq:eta_punctual_bound} and \eqref{eq:etaH_norms} to bound 
\be\label{eq:sq_gq_bounds}
\abs{\s_p}\le\fra{CN^\ka}{\abs p^2},\,\qquad\abs{\s_p-\h_p}\le\fra{CN^{3\ka}}{\abs p^6}\,,\qquad\abs{\g_p}\le C\,,\qquad\abs{\g_p -1}\le\fra{CN^{2\ka}}{\abs p^4}\, .
\ee In position space, we have 
\be\label{eq:sx_gx_bounds}
\norm{\cks}_2=\norm{\sigma}_2\le CN^{\k-\a/2}\arr 0\,,\qquad\norm{\cks}_\io\le CN\,,\qquad\norm{\cks\star\ckg}_\io\le CN \,.
\ee 
In the next lemma, taken from \cite[Lemma 2.3]{BBCS4}, we establish bounds for the remainder operators $d_p, d_p^*$.
\begin{lemma}\label{lem:action_bogoliubov}
For $p \in \L^*_+$, let $d_p, d_p^*$ be defined as in (\ref{eq:dp_def}), $\a>2\ka$ and $N$ large enough. 
Then 
\bes{\label{eq:dp_bounds}
\| (\cN_++1)^{n/2}d_p\xi \| &\le \fra{C}{N}\left[ |\eta_H (p)| \| (\cN_++1)^{(n+3)/2}\xi \| + \| \eta_H \|_2 
\| b_p (\cN_++1)^{(n+2)/2}\xi \| \right]\\
\| (\cN_++1)^{n/2}d_p^*\xi \| &\le \fra{C}{N} \| (\cN_++1)^{(n+3)/2}\xi \| \,.} Moreover, defining $\check{d}_x, \check{d}_x^*$ similarly as in (\ref{eq:dp_def}) but in position space, we have 
\bes{\label{eq:dx_bounds}
		\| (\cN_++1)^{n/2}\check d_x\xi \| &\le\fra{C}{N}\left[\| (\cN_++1)^{(n+3)/2}\xi \| +\| \ckb_x (\cN_++1)^{(n+2)/2}\xi \| \right]\\
		\| (\cN_++1)^{n/2}\cka_y\ckd_x\xi \| &\le\fra{C}{N}\Big[\| \cka_x(\cN_++1)^{(n+1)/2}\xi \| +(1+|\cke(x-y)|) \| (\cN_++1)^{(n+2)/2}\xi \| \\
		&\qquad+\| \cka_y (\cN_++1)^{(n+3)/2}\xi \| +\| \cka_x\cka_y(\cN_++1)^{(n+2)/2}\xi \| \Big]\\
		\| (\cN_++1)^{n/2}\ckd_x\ckd_y\xi \| &\le\fra{C}{N^2}\Big[\| (\cN_++1)^{(n+6)/2}\xi \| +|\cke(x-y)| \| (\cN_++1)^{(n+4)/2}\xi \| \\
		&\qquad+\| \cka_x(\cN_++1)^{(n+5)/2}\xi \|+\| \cka_y(\cN_++1)^{(n+5)/2}\xi\| \\
		&\qquad+\| \cka_x\cka_y(\cN_++1)^{(n+4)/2}\xi \| \Big]\,.
	}
\end{lemma} 

We can now define the renormalized excitation Hamiltonian 
\be\label{eq:GN_def}
\cG_N=e^{-B}\cL_Ne^{B}. \ee
\begin{prop}\label{prop:GN}
	Assume that $2\ka<\a<1/2$. Then we have that
	\[
	\cG_N=C_{\cG_N}+Q_{\cG_N}+\cH_N+\cC_N+\cE_{\cG_N}\,,
	\] where
	\bes{\label{eq:CGN_def}
		C_{G_N}=\,&\fra{N-1}{2}N^\ka\wh V(0)+\sum_{p\in\retp}\left[p^2\s_p^2+N^\ka\wh V(p/N^{1-\ka})(\s_p\g_p+\s_p^2)\right]\\
		&+\fra{1}{N}\sum_{p\in P_H}\eta_p\Big[ p^2\eta_p +\frac{N^\kappa}{2N}\big(\wh V(\cdot/N^{1-\kappa})\star \eta\big)_p\Big] \\
		&+\fra{1}{2N}\sum_{p,q\in\retp}N^\ka\wh V((p-q)/N^{1-\ka})\s_p\g_p\s_q\g_q\\
		&-\fra{1}{N}\sum_{u\in\retp}\s_u^2\sum_{p\in \ret}N^\ka\wh V(p/N^{1-\ka})\h_p\,,
	}
Moreover, we have that	
	\bes{\label{eq:QGN_def}
		Q_{\cG_N}=\sum_{p\in\retp}\left[\Phi_pb_p^*b_p+\frac12\G_p\left(b_p^*b_{-p}^*+b_pb_{-p}\right)\right]\,
	} with (recall the convention that $\gamma_p = 1$ and $\sigma_p = 0$ for $p \in P_H^c$)
	\[
	\begin{split}			
	\Phi_p= &\; 2p^2\s_p^2+N^\ka\wh V(p/N^{1-\ka}) (\g_p+\s_p)^2+\fra{2N^\ka\g_p\s_p}{N}(\wh V(\cdot/N^{1-\ka})\star\h)_p\\
			&-\fra{N^\ka(\g_p^2+\s_p^2)}{N}(\wh V(\cdot/N^{1-\ka})\star\h)_0
			\end{split} \]
as well as
	\[
	\begin{split} 
	\G_p = &\; 2p^2\s_p\g_p+ N^\ka(\g_p+\s_p)^2 \wh V(p/N^{1-\ka})+(\gamma_p^2+\sigma_p^2)\fra{N^{\kappa}}{N}\convo{\h}_p\\
	&-2\g_p\s_p\fra{N^\ka}N\convo{\h}_0
\end{split} \]
and furthermore that
	\bes{\label{eq:CubicGN_def}
		\cC_N=\fra{N^\ka}{\sqrt N}\sum_{\substack{p,q\in\retp\\p\neq-q}}\wh V(p/N^{1-\ka})\left[b_{p+q}^*b_{-p}^*\left(\g_qb_q+\s_qb_{-q}^*\right)+\hc\right]\,.
	} Finally, the self-adjoint error term $\cE_{\cG_N}$ satisfies the operator inequality	
	\be\label{eq:GN_error_bounds}
	\pm\cE_{\cG_N}\le CN^{-1/2+3\ka/2+\a}(\cH_N+\cN_+^2+1)(\cN_++1)\,.
	\ee
\end{prop}
The proof of Proposition \ref{prop:GN} is similar to the proof of \cite[Prop. 3.2, part b)]{BBCS4}. For completeness, we sketch the proof in Appendix \ref{sec:quadratic}. The proposition describes the main contributions to $\cG_N$, up to an error that can be estimated as in \eqref{eq:GN_error_bounds}. The fact that this error is small on low-energy states is a consequence of the following theorem, which is the main result of \cite{ABS}. 
\begin{theorem}\label{thm:mainABS}
Let $\cG_N$ be as in \eqref{eq:GN_def} and assume that $6\ka<\a < 1/2-3\kappa/2$ and $\kappa\in [0;1/44)$. Let $E_N$ be the ground state energy of $H_N$, as defined in (\ref{eq:hamiltoniano}) and $\psi_N  \in L^2_s (\Lambda^N)$ so that $\psi_N = \chi(H_N-E_N\leq  \zeta)\psi_N \in L^2_s(\Lambda^N)$ for a $\zeta>0$. Moreover, let $\xi_N = e^{-B} U_N \psi_N \in \cF_+^{\leq N}$ be the excitation vector associated to $\psi_N$. Then, for every $j\in\mathbb{N}$ and every $\eps>0$, there exists $C>0$ such that
			\[ \langle\xi_N, (\cH_N+1)(\cN_++1)^j\xi_N\rangle\leq C \Big[ N^{20\kappa+\eps}\zeta^2 + N^{44\kappa+2\eps}\Big]^{j+1}. \]	
\end{theorem}

While the vacuum expectation of the renormalized excitation Hamiltonian (\ref{eq:GN_def}) captures the correct ground state energy, to leading order (the main contribution to (\ref{eq:CGN_def}) is exactly $4\pi \frak{a}_0 N^{1+\kappa}$), there are still non-negligible next-order corrections hidden in cubic and quartic parts of $\cG_N$. To extract them, we conjugate $\cG_N$ with another unitary operator of the form $e^{A}$, where $A$ is now an antisymmetric phase, cubic in generalized creation and annihilation operators. More precisely, we define the low-momentum set $P_L=\{p\in\retp\st\abs{p}< N^{\b}\}$, depending on the parameter $\b < \a$, and 
\be	\label{eq:A_def}
A=\fra{1}{\sqrt{N}}\sum_{\substack{r\in P_H,\\v\in P_L}}\eta_r\left[b^*_{r+v}b^*_{-r}b_v-\hc\right]\,.
\ee  
The next lemma will be used to control the action of the unitary operator $e^A$ on powers of the number of particles operator $\cN_+$ and on the product $\cH_N \cN_+$. It can be proven similarly as \cite[Prop. 4.2, Prop. 4.4]{BBCS4} (the second estimate requires bounds on the commutator 
$[\cH_N, A]$ that are shown below, in Lemma \ref{lem:comm_HN_A}). 
\begin{lemma} \label{lm:growNA} 
Let $A$ be defined as in (\ref{eq:A_def}), and let $\ka<\b<\a$ satisfy 	
\[ 2\ka<\b<\a<1/2-2\ka,	\qquad 4\ka <\alpha,	\qquad 3\b+4\ka-1<\alpha .	
\] For any $j \in \bN$, there exists $C > 0$ such that 
\[ e^{-A} (\cN_+ + 1)^j  e^A \leq C (\cN_+ + 1)^j \] 
and 
\[ e^{-A} (\cN_+ + 1) (\cH_N + 1) e^A \leq C(\cH_N+\cN_+^2+N^\ka\cN_+)(\cN_++1). \]
\end{lemma}

With $A$, we define a second renormalized excitation Hamiltonian 
\[
\cJ_N=\conA{\cG_N}\,.
\]
Some important properties of $\cJ_N$ are collected in the next proposition. 
\begin{prop}\label{prop:JN}
	Let $\ka<\b<\a$ satisfy 	
	\bes{\label{eq:conditions_parameters_JN}
	2\ka<\b<\a<1/2-2\ka,	\qquad 4\ka <\alpha,	\qquad 3\b+4\ka-1<\alpha .	
	}
	Then we have that
	\[\label{eq:JN_terms}
	\cJ_N=C_{\cJ_N}+\cQ_{\cJ_N}+\cV_N+\cE_{\cJ_N}\,,
	\] where 	
	\be\label{eq:CJN_QJN}
	C_{\cJ_N}=C_{\cG_N}\,,\qquad\cQ_{\cJ_N}=\sum_{p\in\retp}\Big[F_pb_p^*b_p+\fra{1}{2}G_p(b_p^*b_{-p}^*+b_pb_{-p})\Big]\,,
	\ee
	for (recall that $\gamma_p = 1$ and $\s_p = 0$ for $p \in P_H^c$) 
	\be\label{eq:F_p}
	F_p = (\gamma_p^2+\sigma_p^2)p^2+(\gamma_p+\sigma_p)^2N^\ka (\wh V(\cdot/N^{1-\ka})\star\hat f_N)_p 
		\ee
and
	\be\label{eq:G_p}
	G_p = 2p^2\s_p\g_p+(\g_p+\s_p)^2 N^\ka (\wh V(\cdot/N^{1-\ka})\star\wh f_N)_p 
	\ee The self-adjoint error term $\cE_{\cJ_N}$ satisfies the bound 
	\bes{\label{eq:EJN_bounds}
		\pm\cE_{\cJ_N}\le\,& CN^{-\fra{1}{2}+\fra{5}{2}\ka+\a}(\cH_N+\cN_+^2+1)(\cN_++1)\\
		&+CN^{\frac{\ka}2-\frac{\b}2}(\cK+1)(\cN_++1).
	}
Moreover, for any $\delta> 0$, $\cE_{\cJ_N}$ also satisfies the estimate
	\bes{\label{eq:EJN_bounds2} 
		\pm\cE_{\cJ_N}\le\,& C N^{-\fra{1}{2}+\fra{5}{2}\ka+\a}(\cH_N+\cN_+^2+1)(\cN_++1)\\
		&+ C (N^{\fra12\ka-\fra12\b}+N^{2\ka-\fra\a2})(\cN+1)+ CN^{\ka-2\b}(\cK+1) + \delta \cV_N + C \delta^{-1} N^{\kappa} (\cN_++1).
	}
\end{prop}

{\it Remark:} the second bound (\ref{eq:EJN_bounds2}) on the error term $\cE_{\cJ_N}$ is useful to establish upper bounds on the eigenvalues of $\cJ_N$ (because, on trial states, we will be able to show that the contribution of $\cV_N$ is small, as $N \to \infty$).

The proof of Prop. \ref{prop:JN} is similar to the proof of \cite[Prop. 3.3]{BBCS4}. For completeness, we describe it in Appendix \ref{sec:cubic} below.



\section{Diagonalization of the Quadratic Hamiltonian}
\label{sec:diagonalization}

In this section we apply a final (generalized) Bogoliubov transformation to the renormalized Hamiltonian $\cJ_{N}$ in order to diagonalize its quadratic part. First, we prove some bounds on the coefficients $F_p,\,G_p$ defined in (\ref{eq:F_p}), (\ref{eq:G_p}). 
\begin{lemma}\label{lem:GpFpBounds}
Assume (\ref{eq:conditions_parameters_JN}). There exist constants $C>c>0$ independent of $N$, such that for $N$ large enough, we have:
	\begin{enumerate}[i)]
		\item $c(p^2+N^\ka)\le F_p\le C(p^2+N^\ka)$ for all $p\in\retp$, 
		\item For all $p \in \Lambda^*_+$, 
		\begin{equation}\label{eq:itemii} c \leq 1 + \frac{G_p}{F_p} \leq 2 \, , \quad c \cdot \min \Big[ 1 , p^2 / N^\kappa \Big] \leq 1 - \frac{G_p}{F_p} \leq C. \end{equation} 
		In particular, this implies that $|G_p|/ F_p \leq 1 - c N^{-\kappa}$, for all $p \in \Lambda^*_+$. 
		\item  $\abs{G_p}\le CN^\ka$ for $p\in P_H^c$,
		\item $\abs{G_p}\le CN^{2\ka}\abs{p}^{-2}$ for $p\in P_H$.
	\end{enumerate}
\end{lemma}
\begin{proof}
	We start by proving $i)$. For $p\in P_H^c$, we write 
	\[
	F_p=p^2+N^\ka\convo{\wh f_N}_0+ N^\kappa \left(\convo{\wh f_N}_p-\convo{\wh f_N}_0\right)
	\] and we observe that, by \eqref{eq:approx_8pia_0}, $\convo{\wh f_N}_0 \geq C > 0$, uniformly in $N$. With 
	\[
	\abs{\convo{\wh f_N}_p-\convo{\wh f_N}_0}\le C\abs p/N^{1-\ka} \leq C N^{\alpha + \kappa -1},
	\] for all $p \in P_H^c$, we conclude that
	\[
	\wt c(p^2+N^\ka)- CN^{2\ka+\a-1}\le F_p\le \wt C(p^2+N^\ka)+CN^{2\ka+\a-1}. 
	\] This gives $i)$ for $N$ sufficiently large, thanks to $2\ka+\a-1<0$. For $p\in P_H$ we proceed similarly, using the fact that $ \big|\convo{\wh f_N}_p\big|\le C$, uniformly in $N$.  As for $ii)$, we compute $F_p  - G_p = (\gamma_p-\sigma_p)^2 p^2$ (which in particular implies that $G_p \leq F_p$). With $i)$, distinguishing the cases $|p| \leq N^{\kappa/2}$ and $|p| > N^{\kappa/2}$ (and recalling that $\gamma_p = 1$, $\s_p = 0$ for $p \in P_H^c$), we obtain (\ref{eq:itemii}). Part $iii)$ follows immediately from $\big|\convo{\wh f_N}_p\big|\le C$. To  prove $iv)$, we can use the scattering equation. Indeed, for $p\in P_H$, we have by \eqref{eq:sq_gq_bounds} that
\[ \begin{split} 2p^2\abs{\s_p\g_p-\h_p} &\le CN^{3\ka}\abs{p}^{-4}\le CN^{2\ka}\abs p^{-2}\,,\\ 
N^\ka\convo{\wh f_N}_p\abs{(\g_p+\s_p)^2-1} &\le CN^\ka\convo{\wh f_N}_p\abs{\s_p\g_p+\s_p^2} \\ &\le CN^{2\ka}\abs p^{-2}\, . \end{split} \] Thus, it is enough to prove that $p^2\h_p+N^\ka\convo{\wh f_N}_p$ satisfies the desired bound. By the scattering equation \eqref{eq:scattering_fN_fourier} and the bound \eqref{eq:chi_f_bound}, we find
\begin{equation} \label{eq:scat-G} p^2\h_p+N^\ka\convo{\wh f_N}_p=N^{3-2\kappa}\lambda_\ell (\wh \c_\ell \star \wh f_N)_p\le CN^{\ka}\abs p ^{-2}.\end{equation}  
To show the last inequality, we remark that 
\[
\wh \c_\ell(p)=\int_{\abs x\le \ell} e^{-ip\cdot x}=\fra{4\pi}{\abs{p}^2}\left(\fra{\sin(\ell\abs p)}{\abs p}-\ell\cos(\ell\abs p)\right)\,,
\] which in particular implies that $|\wh \c_\ell (p)| \le C / p^2$ (recall that $\ell < 1/2$). Moreover, recalling that $ \check \eta$ is supported in $B_\ell$, we find, with \eqref{eq:eta_punctual_bound}, $|(\wh\c_\ell \star\h)_q | =  | \widehat{(\chi_\ell  \check{\eta}) }(q) | = |\eta_q| \leq C N^\kappa / q^2$ for all $q\in\Lambda_+^*$. Since $\kappa < 1$, we conclude that 
\be\label{eq:chi_f_bound}
\abs{(\wh \c_\ell\star\wh f_N)_p}\le \frac{C}{\abs{p}^2}\,.
\ee
which implies \eqref{eq:scat-G}. This concludes the proof of the lemma.
\end{proof}
	
Part \textit{ii)} of Lemma \ref{lem:GpFpBounds} allows us to define $\tau\in\ell^1(\retp; \bR)\subset\ell^2(\retp;\bR)$ through \be\label{eq:tau_def}
\tanh(2\t_p)=-\fra{G_p}{F_p}\,,
\ee or, equivalently, through 
\be\label{eq:tau_def2}
\t_p=\fra{1}{4}\left[\log\left(1-\fra{G_p}{F_p}\right)-\log\left(1+\fra{G_p}{F_p}\right)\right]\,,
\ee
These coefficients will be used to define the generalized Bogoliubov transformation which is going to diagonalize the quadratic part of the Hamiltonian $\cJ_N$. In the next Lemma we collect some useful properties of $\t_p$. 

\begin{lemma}\label{lem:bounds_tau}
Assume \eqref{eq:conditions_parameters_JN} and let $\tau$ be defined through \eqref{eq:tau_def}. Then there exists a constant $C > 0$ such that  
	\be\label{eq:tau_norms}
	\begin{split}
	\norm{\t}_\io &\le C + \log N^{\kappa/4}  \,,\qquad\quad \hspace{.15cm}  \norm{\t}_2^2\le CN^{3\ka/2} \,,\\
	\norm{\t}_1 &\le CN^{\a+\ka}\,,\qquad\qquad\quad\norm{\t}_{H^1}^2\le CN^{\a+2\ka}\,,
	\end{split}
	\ee for every $N$ large enough.
\end{lemma}
\begin{proof} We observe that, for $|p| < N^{\kappa/2}$, part $ii)$ of Lemma \ref{lem:GpFpBounds} and  \eqref{eq:tau_def2} imply that 
\begin{equation}\label{eq:point-tau1} |\tau_p| \leq \frac{1}{4} \log (N^{\kappa} / |p|^2) + C \end{equation} 
For $N^{\kappa/2} \leq |p| < N^\alpha$, we use instead part $i)$ and part $iii)$ of Lemma \ref{lem:GpFpBounds} and we apply the mean value theorem to \eqref{eq:tau_def2} to show that 
\begin{equation}\label{eq:point-tau2} |\tau_p| \leq C |G_p| / F_p \leq C N^\kappa / p^2 \end{equation}
Similarly, for $|p| \geq N^\alpha$, part $i)$ and part $iv)$ of Lemma \ref{lem:GpFpBounds} lead us to 
\begin{equation}\label{eq:point-tau3}  |\tau_p| \leq C |G_p| / F_p \leq C N^{2\kappa} / p^4 \end{equation}
The three estimates (\ref{eq:point-tau1}), (\ref{eq:point-tau2}), (\ref{eq:point-tau3}) immediately imply the bound for $\| \tau \|_\infty$. To control the other norms, let us denote by $\tau^{(1)}$, $\tau^{(2)}$, $\tau^{(3)}$ the restriction of $\tau$ to the domains $|p| < N^{\kappa/2}$, $N^{\kappa/2} \leq |p| < N^\alpha$ and, respectively, $|p| \geq N^\alpha$. With (\ref{eq:point-tau2}), we easily find that $\| \tau^{(2)} \|^2_2  \leq C N^{3\kappa/2}$, $\| \tau^{(2)} \|_1 \leq C N^{\alpha + \kappa}$, $\| \tau^{(2)} \|_{H^1}^2 \leq C N^{\alpha + 2\kappa}$. Similarly, with (\ref{eq:point-tau3}), we obtain $\| \tau^{(3)} \|_2^2 \leq C N^{4\kappa-5\alpha}$, $\| \tau^{(3)} \|_1 \leq C N^{2\kappa-\alpha}$, $\| \tau^{(3)} \|_{H^1}^2 \leq C N^{4\kappa-3\alpha}$. To compute norms of $\tau^{(1)}$, we need to be a bit more careful. With (\ref{eq:point-tau1}) and with a dyadic decomposition, we find
\[ \begin{split} \| \tau^{(1)} \|_2^2 &\leq C N^{3\kappa/2} + C \sum_{|p| < N^{\kappa/2}} \log^2 (N^{\kappa/2} / |p|) \\ &\leq C N^{3\kappa/2} + C \sum_{j=0}^\infty  \sum_p \chi ( N^{\kappa/2} 2^{-(j+1)} \leq |p| < N^{\kappa/2} 2^{-j}) \, \log^2 (N^{\kappa/2} / |p|) \\ &\leq C N^{3\kappa/2} + C  N^{3\kappa/2} \sum_{j=0}^\infty 2^{-3j} (j+1)^2 \leq C N^{3\kappa/2} \end{split} \]
where we used the fact that the ball $|p| < N^{\kappa/2} 2^{-j}$ contains at most $C N^{3\kappa/2} 2^{-3j}$ points in $\Lambda^*_+$ (the sum could be truncated at $j = C \kappa \log N$, for sufficiently large $C$). Similarly, we find $\| \tau^{(1)} \|_1 \leq C N^{3\kappa}$ and $\| \tau^{(1)} \|_{H^1}^2 \leq C N^{5\kappa/2}$. With \eqref{eq:conditions_parameters_JN}, we obtain (\ref{eq:tau_norms}). 
\end{proof} 

With the coefficients $\tau_p$ introduced in (\ref{eq:tau_def}) we can define the antisymmetric operator \[ T = \frac{1}{2} \sum_{p \in \L^*_+} \tau_p (b_p^* b_{-p}^* - b_p b_{-p} ) \]
and we consider the corresponding generalized Bogoliubov transformation $e^T$. Using the bounds in Lemma \eqref{lem:bounds_tau}, we control the action of $e^T$ on powers of  $\cN_+$ and on the kinetic energy $\cK$, with the following extension of Lemma 3.1 of \cite{BS}. 
\begin{lemma}\label{lem:BtauNumberkinetic}
Assume \eqref{eq:conditions_parameters_JN} and let $\tau$ be defined as in \eqref{eq:tau_def}.
	\begin{enumerate}[i)]
		\item For every $j \in \NN$, there exists a constant 
	$C>0$ such that
		\bes{\label{eq:BtauN_positive}
			e^{-sT} \cN^j_+ e^{sT} &\le CN^{j(j+1)\ka/4} \big( \cN_+ + N^{3\ka/2} \big)^j\, }
			for all $s \in [0;1]$. 	
		\item For every $j \in \NN$, there exists a constant 
	$C>0$ such that
		\bes{\label{eq:BtauK}
			&e^{-sT} \cK \cN^j_+ e^{sT} \le CN^{(j+1)^2 \ka /2} \left(\cK+N^{\a+2\ka}\right)\big( \cN_+ + 
			N^{3\ka/2} \big)^{j}\,,} 
			for all $s \in [0;1]$. 	
		\item Let $\cE_{\cV_N} = e^{-T} \cV_N e^T - \cV_N$. Then, we have
		\be\label{eq:VBtau_bound}
		\pm\cE_{\cV_N}, \pm e^T \cE_{\cV_N} e^{-T} \le CN^{-\fra{1}{2}+\a+\fra{7\ka}{2}}(\cK+N^{\a+2\ka})(\cN_++N^{3\ka/2})\,.	
		\ee
	\end{enumerate}
	
\end{lemma}
To prove Lemma \ref{lem:BtauNumberkinetic} (in particular, part $iii)$), it is useful to bound the potential energy operator in terms of $\cK$ and $\cN_+$.
\begin{lemma}\label{lem:V<KN}
 There exists a constant $C>0$ such that on $\excsp$ we have 
	\be\label{eq:V<KN}
	\cV_N\le C\cK\cN_+\,.
	\ee
\end{lemma} 
\begin{proof}
	For $\x\in\excsp$, we bound
	\spl{\expec{\x}{\cV_N}\le\,&\fra{1}{N^{1-\ka}}\sum_{\substack{p,q,r\in\retp\\r\neq-p,-q}}\big| \wh V(r/N^{1-\ka})\big|\norm{a_{p+r}a_q\x}\norm{a_{q+r}a_p\x}\\
		\le\,&\fra{1}{N^{1-\ka}}\sum_{r\in\retp} \big|\wh V(r/N^{1-\ka})\big| \sum_{\substack{p,q\in\retp\\r\neq-p,-q}}\fra{(p+r)^2}{(q+r)^2} \norm{a_{p+r}a_q\x}^2 \\
		\le\,&\left[\sup_{q\in\retp}\fra{1}{N^{1-\ka}}\sum_{\substack{r\in\retp\\r\neq-q}} \fra{\big|\wh V(r/N^{1-\ka})|}{(q+r)^2}\right]\norm{\cK^{1/2}\cN_+^{1/2}\x}^2\le C\norm{\cK^{1/2}\cN_+^{1/2}\x}^2\,.
	}
\end{proof}
We can now proceed with the proof of Lemma \ref{lem:BtauNumberkinetic}.

\begin{proof}[Proof of Lemma \ref{lem:BtauNumberkinetic}]
	To prove $i)$ we will use induction over $j$. For $j=0$, the claim is trivial. So we assume \eqref{eq:BtauN_positive} to hold with $j$ replaced by $j-1$, for some $j \geq 1$. To prove it holds also for $j$, we compute the commutator
\bes{\label{eq:NBCommutatorNegative}
		[\cN_+^j,T]=\,&\fra{1}{2}\sum_{p\in\retp}\t_pb_p^*b_{-p}^*[(\cN_++2)^j- \cN_+^j] + \text{h.c.} 	} 
Since $(\cN_+ + 2)^j - \cN_+^j \leq 2j \cN_+^{j-1} + C (\cN_+ + 1)^{j-2}$, for a constant $C > 0$ depending on $j$, we obtain 
\[ \begin{split} 
\big| \bra \xi  [\cN_+^j,T] \ket\x \big| &\le \sum_{p\in\retp} |\t_p | \| b_p \big[ \cN_+^j - (\cN_+ -2)^j \big]^{1/2} \xi \| \| b_{-p}^* \big[ (\cN_+ + 2)^j - \cN_+^j \big]^{1/2} \x \| 
\\ &\le \sum_{p\in\retp}\abs{\t_p} \| b_p \big[ \cN_+^j - (\cN_+ -2)^j \big]^{1/2} \xi \| 
\\ &\hspace{2cm} \times\left(\| b_{-p} \big[ (\cN_+ + 2)^j - \cN_+^j \big]^{1/2} \x \|+ \| \big[ (\cN_+ + 2)^j - \cN_+^j \big]^{1/2}\x \| \right) 
\\ &\le ( 2j \norm{\t}_\io + C) \|\cN_+^{j/2}\x \|^2+ C \norm{\t}_2^2 \| (\cN_++1)^{(j-1)/2}\x \|^2\,.\end{split} \] 
For $s \in [0;1]$, we set $\x_s=e^{sT}\x$ and $\ph(s)= \bra{\x_s}(\cN_+)^j\ket{\x_s}$. With \eqref{eq:NBCommutatorNegative}, Lemma \ref{lem:bounds_tau} and the inductive assumption, we obtain, for $N$ large enough, 
	\begin{equation*}
	\begin{split}
	\abs{\partial_s\ph(s)}\le&  (\log(N^{j \ka/2}) + C) \ph(s)+CN^{3\ka/2} \expec{\x_s}{(\cN_++1)^{(j-1)}}
	\\ \le & (\log(N^{j \ka /2}) + C) \ph(s)+CN^{j(j-1)\ka/4} \expec{\x}{(\cN_+ + N^{3\ka/2})^{j'}}
	\end{split}
	\end{equation*}  for every $s \in [0;1]$. Gronwall's Lemma implies \eqref{eq:BtauN_positive}.
	
	We proceed analogously to prove $ii)$. For $j=0$ we have 
	\[
	[\cK,T]=\sum_pp^2\t_p(b_p^*b_{-p}^*+b_pb_{-p})\,,
	\] hence 
	\spl{
		\bra{\x}[\cK,T]\ket{\x}&\le 2\sum_{p\in\retp}p^2\abs{\t_p}\norm{b_p\x}\big(\norm{b_{-p}\x}+\norm{\x}\big)\\
		&\le (2 \norm{\t}_\io +C ) \| \cK^{1/2}\x \|^2+C\norm{\t}_{H^1}^2\norm{\x}^2,\,
	} and applying Gronwall's Lemma as above yields the claim. For $j\ge 1$ we have 
	\be\label{eq:BKNCommutator}
	[\cK \cN_+^j, T]=\cK [\cN_+^j,T] + [\cK,T] \cN_+^j.
	\ee To bound the first term in the right-hand side of \eqref{eq:BKNCommutator} we compute 
	\spl{
		\cK [ \cN_+^j,T ] 
		=\,&\sum_{p,q\in\retp}q^2\t_pa_q^*a_qb_p^*b_{-p}^* [(\cN_++2)^j- \cN_+^j]  \\ & +\sum_{p,q\in\retp}q^2\t_pa_q^*a_qb_pb_{-p} [\cN_+^j- (\cN_+-2)^j]  \\
		=\,&\sum_{p,q\in\retp}q^2\t_pb_q^*b_p^*a_{-p}^*a_q [(\cN_++2)^j- \cN_+^j] \\ & +2 \sum_{p\in\retp}p^2\t_pb_p^*b_{-p}^* [(\cN_++2)^j- \cN_+^j] \\
		&+\sum_{p,q\in\retp}q^2\t_pa_q^*a_pb_qb_{-p} [\cN_+^j- (\cN_+-2)^j]  
		=: \text{I}+\text{II}+\text{III}\,,
	} We have 
	\spl{
		\big| \bra{\x}\text{I}\ket{\x}\big|\le\,& \sum_{p,q\in\retp}q^2\abs{\t_p} \| b_q b_p [\cN_+^j- (\cN_+-2)^j]^{1/2} \xi \| \\
		&\times\Big(\| a_{-p}a_q [(\cN_++2)^j- \cN_+^j]^{1/2} \x \| +\| a_q [(\cN_++2)^j- \cN_+^j]^{1/2} \x \| \Big)\\
		\le\,& (2j \norm{\t}_\io +C) \| \cN_+^{j/2}\cK^{1/2}\x \|^2 + C\norm{\t}_2^2 \| (\cN_++1)^{(j-1)/2}\cK^{1/2}\x \|^2\,,
	} and similarly 
	\[
	\big|\bra{\x}\text{III}\ket{\x}\big|\le (2j \norm{\t}_\io + C) \| \cN_+^{j/2} \cK^{1/2}\x \|^2+C\norm{\t}_2^2 \|(\cN_+ +1)^{(j-1)/2}\cK^{1/2}\x \|^2\,.
	\] As for $\text{II}$, we find
	\spl{
		\big| \bra{\x}\text{II}\ket{\x}\big|
		\le\,& C\norm{\t}_\io \| \cN_+^{(j-1)/2}\cK^{1/2}\x \|^2+C\norm{\t}_{H^1}^2 \|(\cN_++1)^{(j-1)/2}\x \|^2\,.
	} The second term on the right-hand side of \eqref{eq:BKNCommutator} is given by
	\[
	[\cK,T] \cN_+^j=\sum_{p\in\retp}p^2\t_p(b_p^*b_{-p}^*+b_pb_{-p})\cN_+^j
	\] which can be bounded as above and satisfies the  estimate
	\[
	\expec{\x}{[\cK,T] \cN_+^j} \le (2 \norm{\t}_\io + C) \| \cN_+^{j/2}\cK^{1/2}\x \|^2 + C\norm{\t}_{H^1}^2 \|\cN_+^{j/2}\x \|^2.
	\]
	 Putting things together we find 
	 \begin{equation*}
	 \begin{split}
	 	 \expec{\x}{[ \cK \cN_+^j, T]}\le& \, [ (4j + 2) \norm{\t}_\io + C] \| \cN_+^{j/2}\cK^{1/2}\x
		 \|^2 + C\norm{\t}_{H^1}^2 \| \cN_+^{j/2}\x \|^2\\
	 &+C\norm{\t}_2^2 \| (\cN_++1)^{(j-1)/2} \cK^{1/2}\x \|^2.
	 \end{split}
	 \end{equation*}
	 Using Gronwall's Lemma again as above, the inductive assumption, Lemma \ref{lem:bounds_tau} and the bound \eqref{eq:BtauN_positive} we obtain (\ref{eq:BtauK}). 
	
	Finally, we turn to the proof of part $iii)$. We write 
	\be\label{eq:VBtau_expansion}
	\cE_{\cV_N}= - \int_0^1 e^{-sT} [\cV_N, T] e^{sT} ds, 
	\ee and using the expression (\ref{eq:cV-x}) in position space, we find
	\[
	\begin{split}
	[\cV_N, T ]=&\fra{1}{2}\int dxdy\,N^{2-2\ka}V(N^{1-\ka}(x-y))\check{\t}(x-y)(\ckb_x^*\ckb_y^*+\ckb_x\ckb_y)\\
	&+\int dxdy\,N^{2-2\ka}V(N^{1-\ka}(x-y))[\ckb_x^*\ckb_y^*a^*(\check\t_y)\cka_x+\hc].
	\end{split}
	\]
With $\norm{\check\t}_\io\le\norm{\t}_1$, $\norm{\check \t_y}_2=\norm{\check\t(\cdot-y)}_2=\norm{\check\t}_2=\norm\t_2$ and Lemma \ref{lem:V<KN}, we bound 
	\spl{
		\abs{\expec{\x}{[\cV_N, T ]}}\le\,& \int dxdy\,N^{2-2\ka}V(N^{1-\ka}(x-y))\abs{\check{\t}(x-y)} \| \ckb_x\ckb_y\x \| \norm\x\\
		&+\int dxdy\,N^{2-2\ka}V(N^{1-\ka}(x-y))\norm{\ckb_x\ckb_y\x}\norm{a^*(\check\t_y)\cka_x\x}\\
		\le\,&C N^{(\ka-1)/2} \norm{\t}_1 \| \cV_N^{1/2}\x \| \norm{\x}+CN^{(\ka-1)/2} \norm{\t}_2 \| \cV_N^{1/2}\x \| \norm{\cN\x}\\
		\le\,&CN^{(\ka-1)/2}\norm\t_1\expec{\x}{(\cK+1)(\cN_++1)}\,.
	} Using (\ref{eq:BtauK}) 
	and Lemma \ref{lem:bounds_tau} we get	
	\[
	\begin{split}
	\pm e^{-sT}{[\cV_N,T]}e^{sT}\le
 CN^{-\fra{1}{2}+\a+\fra{7\ka}{2}}(\cK+N^{\a+2\ka})(\cN_++N^{3\ka/2}),
	\end{split}
\] for every $s\in[0,1]$. Together with \eqref{eq:VBtau_expansion}, this concludes the proof of $iii)$. 
\end{proof}

Notice that by Lemma \ref{lem:bounds_tau}, $\norm{\t}_2$ is not uniformly bounded in $N$: this prevents us from applying Lemma \ref{lem:action_bogoliubov} to understand the action of $e^T$. Nevertheless, following \cite{BBCS4}, we obtain an expansion similar to \eqref{eq:dp_def}, using the fact that $\t\in\ell^1$ and that $\norm{\t}_\io$ only grows logarithmically in $N$. Using the commutation relations \eqref{eq:bpCCR} we write 
\spl{
e^{-T} b_p e^T =\,&b_p+\int_0^1\consB{s}{[b_p, T ]}ds\\
	=\,&b_p+\int_0^1 e^{-sT} \Big[\t_p b_{-p}^*-\Big(\fra{\cN_+}{N}\t_pb_{-p}^*+\frac1N\sum_{q\in\retp}\t_qb_q^*a_{-q}^*a_p\Big)\Big] e^{sT} \, ds\\
	=\,&b_p+\t_pb_{-p}^*-\int_0^1 e^{-sT} \Big[\fra{\cN_+}{N}\t_pb_{-p}^*+\fra{1}{N}\sum_{q\in\retp}\t_qb_q^*a_{-q}^*a_p\Big] e^{sT} \, ds\\
	&+\int_0^1\int_0^{s_1} e^{-s_2 T} [\t_pb_{-p}^*,B(\t)] e^{s_2 T} \, ds_2ds_1\\
} Iterating we find, for every $k \in \bN$, the truncated expansion
\begin{equation}\label{eq:Dp_exp_trunc}
\begin{split}
e^{-T} b_p e^T =\,&\sum_{n=0}^k\fra{\t_p^{2n}}{(2n)!}b_p+\sum_{n=1}^k\fra{\t_p^{2n-1}}{(2n-1)!}b_{-p}^*+D_p^{(k)},
\end{split}
\end{equation}
with
\begin{equation*}
\begin{split}
D_p^{(k)}=\,&\sum_{n=0}^k\fra{\t_p^{2n}}{N}\int_0^1 .. \int_0^{s_{2n}} e^{-s_{2n+1} T} \Big[ \t_p\cN_+b_{-p}^*+\sum_{q\in\retp}\t_qb_q^*a_{-q}^*a_p\Big] e^{s_{2n+1} T} \, ds_{2n+1} .. ds_1\\
&+\sum_{n=1}^k\fra{\t_p^{2n-2}}{N}\int_0^1 .. \int_0^{s_{2n-1}} e^{-s_{2n} T} \Big[ \t_p^2b_{p}\cN_++\t_p\sum_{q\in\retp}\t_qa_{-p}^*a_{-q}b_q\Big] e^{s_{2n} T}  ds_{2n} ..  ds_1\,,
\end{split}
\end{equation*} Taking the limit for $k\to\io$ in \eqref{eq:Dp_exp_trunc} and its Hermitian conjugate we get the expansions
\bes{\label{eq:Dp_def}
e^{-T} b_p e^T  &= \tg_p b_p+\ts_p b_{-p}^*+D_p,\\
e^{-T} b_p^* e^T &= \tg_p b_p^*+\ts_p b_{-p}+D_p^*,
} where we defined the coefficients $\tg_p = \cosh \tau_p$, $\ts_p = \sinh \tau_p$ and where the remainder operators satisfy, for any $n \in \bZ$, the estimate 
\be\label{eq:Dp_bounds}
\begin{split}
\| (\cN_+ &+1)^{n/2} D_p \xi \| \\ \le & \, \fra{Ce^{\norm{\t}_\io}}{N}\int_0^1\Bigg(\abs{\t_p}\norm{(\cN_++1)^{(n+3)/2}e^{sT} \xi}+\norm{\t}_1\norm{a_p(\cN_++1)^{(n+2)/2} e^{sT}\xi}\Bigg)ds\\
	\le &CN^{-1+\ka /4}\int_0^1\Bigg(\abs{\t_p}\norm{(\cN_++1)^{(n+3)/2}e^{sT}\xi}+N^{\a+\ka}\norm{a_p(\cN_++1)^{(n+2)/2} e^{sT}\xi}\Bigg)ds,
\end{split}
\ee  
Observe that, from Lemma \ref{lem:bounds_tau} and from the bounds (\ref{eq:point-tau1}), (\ref{eq:point-tau2}),  (\ref{eq:point-tau3}),  
\begin{equation}\label{eq:tstg_low}
\abs{\tg_p},\abs{\ts_p}\le  \frac{C N^{\ka /4}}{|p|^{1/2}}  \qquad\text{for }\abs p \le N^{\ka/2},
\end{equation} whereas $\abs{\tg_p}\le C$ if $|p| > N^{\kappa/2}$ and 
\begin{equation}\label{eq:tstg_high}
\abs{\ts_p} \leq C |G_p|/ F_p \leq C \left\{ \begin{array}{ll} N^\kappa / p^2  \qquad &\text{if $N^{\kappa/2} \leq |p| \leq N^\alpha$} \\ N^{2\kappa} / |p|^4\qquad &\text{if $|p| \geq N^\alpha$}  \end{array} \right. 
\end{equation}
 
 We can now study the action of $e^T$ on the quadratic operator $\cQ_{\cJ_N}$, defined in (\ref{eq:CJN_QJN}). 
\begin{lemma}\label{lem:diagonalization}
Assume \eqref{eq:conditions_parameters_JN}. Then we have 
	\be\label{eq:diagonalizationQJN}
	e^{-T} \cQ_{\cJ_N} e^T =-\fra{1}{2}\sum_{p\in\retp}\left[F_p-\sqrt{F_p^2-G_p^2}\right]+\sum_{p\in\retp}\sqrt{F_p^2-G_p^2}\,a_p^*a_p+ \cE_{\cQ_{\cJ_N}}\,,
	\ee with the bound 
	\bes{\label{eq:diagonalization_error_bound}
		\pm\cE_{\cQ_{\cJ_N}}\le CN^{-1+2\a + 7\kappa/2} (\cK + N^{\alpha+2\kappa}) (\cN_+ + N^{3\kappa/2}) 
	}
\end{lemma}

\begin{proof}
Applying \eqref{eq:Dp_def} to (\ref{eq:CJN_QJN}), we find
	\spl{
		e^{-T} \cQ_{\cJ_N} e^T=\,&\sum_{p\in\retp}F_p(\tg_pb_p^*+\ts_pb_{-p})(\tg_pb_p+\ts_pb_{-p}^*)\\
		&+\fra{1}{2}\sum_{p\in\retp}\Big[G_p(\tg_pb_p^*+\ts_pb_{-p})(\tg_pb_{-p}^*+\ts_pb_p)+\hc\Big]+ \cE_1 \,,
	} with
	\bes{\label{eq:diagon_tdelta}
	 \cE_1=\,&\sum_{p\in\retp}F_pD_p^* e^{-T} b_p e^T +\sum_{p\in\retp}F_p(\tg_pb_p^*+\ts_pb_{-p})D_p\\ &+\fra{1}{2}\sum_{p\in\retp}\Big[G_pD_p^* e^{-T} b_{-p}^* e^T + \hc\Big]+\fra{1}{2}\sum_{p\in\retp}\Big[G_p(\tg_pb_p^*+\ts_pb_{-p})D_{-p}^*+\hc\Big]\\
		=: \,&\text{I}+\text{II}+\text{III}+\text{IV}\,.
	} The coefficients $\tau_p$, defined in \eqref{eq:tau_def}, are exactly chosen to approximately diagonalize 
	\begin{equation*}
	\begin{split}
		\sum_{p\in\retp}&F_p(\tg_pb_p^*+\ts_pb_{-p})(\tg_pb_p+\ts_pb_{-p}^*)+\fra{1}{2}\sum_{p\in\retp}\Big[G_p(\tg_pb_p^*+\ts_pb_{-p})(\tg_pb_{-p}^*+\ts_pb_p)+\hc\Big]\\
	&=-\fra{1}{2}\sum_{p\in\retp}\left[F_p-\sqrt{F_p^2-G_p^2}\right]+\sum_{p\in\retp}\sqrt{F_p^2-G_p^2}\, a_p^* a_p+ \cE_2,
	\end{split}
	\end{equation*} where, from the commutation relations (\ref{eq:bpCCR}), 
\spl{
\cE_2 =&\sum_{p\in\retp}\sqrt{F_p^2-G_p^2}\, (b_p^* b_p - a_p^* a_p)  - \fra{1}{2}\sum_{p\in\retp}\left[F_p-\sqrt{F_p^2-G_p^2}\right]\left(\fra{\cN_+}{N}+\fra{a_{-p}^*a_{-p}}{N}\right).
} 
Observe that, by Lemma \ref{lem:GpFpBounds}, 
	\spl{
		\pm\sum_{p\in\retp} \sqrt{F_p^2-G_p^2}\,\left(b_p^*b_p-a_p^*a_p\right) \le \;& C\sum_{p\in\retp}F_pa_p^* \, \fra{\cN_+}{N} \, a_p \\ \le \; &CN^{-1}\sum_{p\in\retp}(p^2+N^\ka) a_p^* \cN_+ a_p \leq CN^{-1}(\cK+N^\ka) \cN_+ .
	} 
With 
\[ 0 \leq F_p - \sqrt{F_p^2 - G_p^2} \leq G_p^2 / F_p \]
and from the bound in Lemma \ref{lem:GpFpBounds}, we conclude that  
\[
\pm \cE_2 \leq C N^{-1} (\cK + N^\kappa ) \cN_+ + C N^{\alpha+2\kappa -1} \cN_+ .
\]
We still have to show that the four terms on the r.h.s. of \eqref{eq:diagon_tdelta} satisfy (\ref{eq:diagonalizationQJN}). Using Lemma \ref{lem:GpFpBounds}, Lemma \ref{lem:BtauNumberkinetic} and the bound \eqref{eq:Dp_bounds} with $n=-1$, we 
can bound  
 \[ \begin{split} 
|\bra{\xi} \text{I} \ket{\xi} | \le\,&C \sum_{p\in\retp} (p^2+N^\ka)  \| (\cN_+ + 1)^{1/2} e^{-T} b_p e^T \xi \| \|( \cN_+ +1)^{-1/2} D_p \xi \| \\ \leq \; &C N^{-1+\kappa /2} \sum_{p \in \Lambda^*_+} (p^2 + N^\kappa ) \| a_p (\cN_+ + N^{3\kappa/2} )^{1/2}  e^T \xi \| \\ &\hspace{1cm} \times  \int_0^1 \left[ |\tau_p| \| (\cN_+ + 1) e^{sT} \xi \| + N^{\alpha+\kappa} \| a_p (\cN_+ + 1)^{1/2} e^{sT} \xi \| \right] ds \\ \leq \; &C N^{-1+\alpha+7\kappa/2} \| (\cK + N^{\alpha+2\ka})^{1/2} (\cN_+ + N^{3\kappa/2})^{1/2} \xi \|^2 
\end{split} \]
To estimate the second term on the r.h.s. of \eqref{eq:diagon_tdelta} we split $\text{II} = \text{II}_1 + \text{II}_2$, where $\text{II}_1$ contains the sum over $|p| \leq N^{\kappa/2}$ and $\text{II}_2$ the sum over $|p| > N^{\kappa/2}$. With \eqref{eq:tstg_low} and \eqref{eq:tstg_high} we estimate
\[ \begin{split} 
|\bra{\xi} \text{II}_1 \ket{\xi}| \leq \; &C N^{-1+\alpha+9\kappa/2}  \| (\cN_+ + N^{3\kappa/2} ) \xi \|^2 \end{split} \]
and 
\[ \begin{split} 
|\bra{\xi} \text{II}_2 \ket{\xi}| \leq \; &C N^{-1+\kappa/4} \sum_{|p| > N^{\kappa/2}} p^2  \left[ \| a_p (\cN_+ +1)^{1/2} \xi \| + |\tau_p| \| (\cN_+ + 1) \xi \| \right]  \\ &\hspace{1cm} \times  \int_0^1 \left[ |\tau_p| \| (\cN_+ + 1) e^{sT} \xi \| + N^{\alpha+\kappa} \| a_p (\cN_+ + 1)^{1/2} e^{sT} \xi \| \right] ds \\
\leq \; &C N^{-1+3\alpha /2+ 13 \kappa /4}   \| (\cK + N^{\alpha+2\kappa})^{1/2} (\cN_+ + N^{3\kappa/2})^{1/2}  \xi \|^2
\end{split} \]
Noting that, by Lemma \ref{lem:GpFpBounds}, $\norm{G}_2 \le CN^{\ka+3\a/2}$, $\norm{G/\abs \cdot}_2\le CN^{\ka+\a/2}$, we further get 
	\spl{
		| \expec{\x}{\text{III}} | \le\,& C N^{-1+5\ka /4}\sum_{p\in\retp}\abs{G_p} \| (\cN_+ + N^{3\kappa/2} ) \x \| \\
		&\hspace{1cm}\times\int_0^1\Bigg(\abs{\t_p}\norm{(\cN_++1)e^{sT}\xi}+ N^{\alpha+\kappa} \norm{a_p (\cN_++1)^{1/2}e^{sT}\xi}\Bigg) ds\\
		 \le\, & CN^{-1+3\a/2+17\kappa/4}\expec{\x}{(\cK+N^{\a+2\ka}) (\cN+ N^{3\ka/2})}
	}
Finally, applying \eqref{eq:Dp_bounds} with $n=-2$, we find
\[ \begin{split}  | \bra{\x} &{\text{IV}} \ket{\xi} | \\ \le \; &C \| (\cN_+ + 1) \xi \| \sum_{p\in\retp} \abs{G_p} \| (\cN_+ + 1)^{-1} D_{-p} (\tg_p b_p + \ts_p b^*_{-p}) \xi \| \\
\leq \; &N^{-1+\alpha+5\kappa/4}  \| (\cN_+ + 1) \xi \| \sum_{p \in \L^*_+} |G_p|  \int_0^1 ds \, \| (\cN_+ + 1)^{1/2} e^{sT} (\tg_p b_p + \ts_p b^*_{-p}) \xi \|   \\
\leq \; &N^{-1+\alpha+3\kappa/2}  \| (\cN_+ + 1) \xi \| \sum_{p \in \L^*_+} |G_p| \| (\tg_p b_p + \ts_p b^*_{-p}) (\cN_+ + N^{3\kappa/2})^{1/2}  \xi \| 
\end{split} \]
Dividing the sum into the domains $|p| < N^{\kappa/2}$ and $|p| > N^{\kappa/2}$, and using the bounds (\ref{eq:tstg_low}), $\norm{G/\abs \cdot}_2\le CN^{\ka+\a/2}$ and Lemma \ref{lem:bounds_tau}, we find 
\[ | \expec{\x}{\text{IV}} | \le C N^{-1+2\alpha+7\kappa/2}   \| \cK^{1/2} (\cN_+ + N^{3\kappa/2})^{1/2}  \xi \|^2 \]
Combined with the previous bounds, this implies (\ref{eq:diagonalization_error_bound}). 
\end{proof}

We define now our final excitation Hamiltonian 
\[
\cM_N = e^{-T} \cJ_N e^T, \]and we introduce the notation 
\be\label{eq:EBog_notation}
E_{Bog}=\fra{1}{2}\sum_{p\in\retp}\left[\sqrt{|p|^4+16\pi\fa_0N^\ka p^2}-p^2-8\pi\fa_0N^\ka+\fra{(8\pi\fa_0N^\ka)^2}{2p^2}\right] 
\ee 
\begin{prop}\label{prop:MN_Def}
Assume $\ka,\a,\b$ satisfy the conditions \eqref{eq:conditions_parameters_JN}. Then \[
\begin{split} 
	\cM_N = \; &4\pi\fa_0N^\ka(N-1) + e_\Lambda(\fa_0N^{\ka})^2+ E_{Bog} \\ &+ \sum_{p\in\retp}\sqrt{|p|^4+16\pi\fa_0N^\ka |p|^2}\,a_p^*a_p + \cV_N +\cE_{\cM_N}, \end{split} \]
with $e_\Lambda$ as defined in (\ref{eq:eLambda}) and where the error term is such that
\begin{equation}\label{eq:EMN_bound_1}
\begin{split}  \pm e^T \cE_{\cM_N} e^{-T} \leq & \; C (N^{-1+3\ka+2\a}+N^{3\ka-\a} ) \\ &+  
C  (N^{-1/2 + \alpha + 7\kappa/2} + N^{-1+2\alpha + 11\kappa/2}) (\cK+ N^{\alpha+2\kappa}) ( \cN_+ + N^{3\kappa/2}) \\ &+ 
CN^{-\fra{1}{2}+\fra{5}{2}\ka+\a}(\cH_N+\cN_+^2+1)(\cN_++1)+CN^{\frac{\ka}2-\frac{\b}2}(\cK+1)(\cN_++1) \,. \end{split}
\end{equation} 
Moreover, for every $\delta > 0$, we find 
\begin{equation}\label{eq:EMN_bound_2}
\begin{split}
\pm \cE_{\cM_N}\leq \; & C (N^{-1+3\ka+2\a}+N^{3\ka-\a} ) + CN^{-\fra{1}{2}+ 7\ka+\a}\cKK\cNN^2 \\ &+ C( N^{\kappa-\fra\beta2}+N^{\fra{7\ka}2-\fra\a2}) \cNN + N^{\fra32\ka-2\b}(\cK+N^{\a+2\ka}) + \delta \cV_N \\ &+C \delta N^{-1/2+\alpha + 7\kappa/2} \cKK \cNN + C\delta^{-1} N^{\fra{3}{2}\kappa} \cNN.
\end{split}
\end{equation}
\end{prop}
{\it Remark:} Similarly as remarked after Prop. \ref{prop:JN}, the second estimate (\ref{eq:EMN_bound_2}) for the error term $\cE_{\cM_N}$ will be useful to prove upper bounds on the eigenvalues of $\cM_N$.

\begin{proof}
From Prop.  \ref{prop:JN} and Lemma \ref{lem:diagonalization}, we obtain that
\begin{equation*}
\begin{split} \cM_N = \; &C_{\cJ_N}  - \fra{1}{2}\sum_{p\in\retp}\left[F_p-\sqrt{F_p^2-G_p^2}\right] +\sum_{p\in\retp}\sqrt{F_p^2-G_p^2}\,a_p^*a_p + \cV_N  \\ &+ \cE_{\cV_N} + \cE_{\cQ_{\cJ_N}} + e^{-T} \cE_{\cJ_N} e^T \end{split} \end{equation*} 
where $\cE_{\cJ_N}$, $\cE_{\cQ_{\cJ_N}}$ satisfy the bounds (\ref{eq:EJN_bounds}), (\ref{eq:diagonalization_error_bound}) while $\cE_{\cV_N}$ and $e^T \cE_{\cV_N} e^{-T}$ satisfy  (\ref{eq:VBtau_bound}). Using Lemma \ref{lem:BtauNumberkinetic} to control the action of $e^T$ on $\cE_{\cQ_{\cJ_N}}$, we find that
\begin{equation}\label{eq:deltaeps}
\begin{split} 
\pm e^T &\left[ \cE_{\cV_N} +  \cE_{\cQ_{\cJ_N}} + e^{-T} \cE_{\cJ_N} \e^{T} \right] e^{-T} \\  \leq \; &
C (N^{-1/2 + \alpha + 7\kappa/2} + N^{-1+2\alpha + 11\kappa/2}) (\cK+ N^{\alpha+2\kappa}) ( \cN_+ + N^{3\kappa/2}) \\ &+ 
 CN^{-\fra{1}{2}+\fra{5}{2}\ka+\a}(\cH_N+\cN_+^2+1)(\cN_++1) +C N^{\frac{\ka}2-\frac{\b}2}(\cK+1)(\cN_++1).
\end{split} \end{equation} 
Using (\ref{eq:EJN_bounds2}), instead of (\ref{eq:EJN_bounds}), to control $\cE_{\cJ_N}$, and applying Lemma \ref{lem:BtauNumberkinetic} (together with the estimate (\ref{eq:V<KN})) 
to control the action of $e^T$ on $\cE_{\cJ_N}$, we also find
\begin{equation}\label{eq:deltaeps2}
\begin{split} 
\pm  &\left[ \cE_{\cV_N} +  \cE_{\cQ_{\cJ_N}} + e^{-T} \cE_{\cJ_N} \e^{T} \right] \\ \leq \; &CN^{-\fra{1}{2}+ 7\ka+\a}\cKK\cNN^2 \\
&+ C (N^{\kappa-\fra\beta2}+N^{\fra{7\ka}2-\fra\a2}) \cNN + N^{\fra32\ka-2\b}(\cK+N^{\a+2\ka})+ \delta \cV_N \\ &+C \delta N^{-1/2+\alpha + 7\kappa/2} \cKK \cNN + C\delta^{-1} N^{\fra{3}{2}\kappa} \cNN.
\end{split}
\end{equation}
From (\ref{eq:F_p}), (\ref{eq:G_p}), we obtain 
\begin{equation*}
		\sqrt{F_p^2-G_p^2} =
		\sqrt{|p|^4+2p^2N^\ka\convo{\wh f_N}_p},
	\end{equation*} 
and with the bound $| (\wh V(\cdot/N^{1-\ka})\star \wh f_N)_p-(\wh V(\cdot/N^{1-\ka})\star \wh f_N)_0 | \le C|p|/N^{1-\ka}$ and the approximation \eqref{eq:approx_8pia_0}, we find 	
	\be \label{eq:CGM_approx_convo}
	\abs{N^\ka(\wh V(\cdot/N^{1-\ka})\star\wh f_N)_p-8\pi\fa_0N^\ka}\le C N^{2\ka-1}\abs p.
	\ee 
Thus, 
\begin{equation}\label{eq:MN-Q}  \sum_{p \in \L^*_+} \sqrt{F_p^2-G_p^2}\,a_p^*a_p = \sum_{p \in \L^*_+} \sqrt{|p|^4 + 16 \pi \frak{a}_0 N^\kappa p^2}\,a_p^*a_p + \cE_1 \end{equation} 
with $\pm \cE_1 \leq C N^{-1+2\kappa} \cK$ and, from Lemma \ref{lem:BtauNumberkinetic}, 
\begin{equation*}
\pm e^T \cE_1 e^{-T} \leq C N^{-1+5\kappa/2}  (\cK + N^{\alpha+2\kappa}).
\end{equation*} 

Let us now turn our attention to the constant term \begin{equation*}
 C_{M_N} := C_{\cJ_N}  - \frac{1}{2} \sum_{p\in\retp} \big[ F_p-\sqrt{F_p^2-G_p^2} \big]. \end{equation*}   
With the notation introduced in (\ref{eq:eLambda}) and in (\ref{eq:EBog_notation}), we will show that 
\begin{equation}\label{eq:const-claim} C_{\cM_N} =  4\pi\fa_0N^\ka(N-1)+ E_{Bog} + e_\Lambda (\frak{a}_0 N^\kappa)^2 + \cO(N^{-1+3\ka+2\a}+N^{3\ka-\a}). \end{equation}
Combined with (\ref{eq:deltaeps}), with (\ref{eq:deltaeps2}) and with the bounds for the error $\delta_1$ introduced in (\ref{eq:MN-Q}), (\ref{eq:const-claim}) shows (\ref{eq:EMN_bound_1}) and (\ref{eq:EMN_bound_2}) and completes therefore the proof of Prop. \ref{prop:MN_Def}. 
To show (\ref{eq:const-claim}), we first observe that, from (\ref{eq:CGN_def}), (\ref{eq:CJN_QJN}), (\ref{eq:F_p}) and (\ref{eq:G_p}),  
\spl{
		C_{\cM_N}=\,&\fra{N-1}{2}N^\ka \wh V(0)+\sum_{p\in P_H}\left[p^2\s_p^2+N^\ka\wh V(p/N^{1-\ka})\left(\s_p\g_p+\s_p^2\right)\right]\\
		&+\fra{1}{N}\sum_{p\in P_H}\left[p^2\eta_p^2+\fra{N^\ka}{2N}(\wh V(\cdot/N^{1-\ka})\star\h)_p\h_p\right]\\
		&+\fra{N^\ka}{2N}\sum_{p,q\in P_H}\wh V((p-q)/N^{1-\ka})\g_q\s_q\g_p\s_p-\fra{1}{N}\sum_{u\in\retp}\s_u^2\sum_{p\in\ret}N^\ka\wh V(p/N^{1-\ka})\h_p\\
		&+\fra{1}{2}\sum_{p\in\retp}\Bigg[\sqrt{|p|^4+2p^2N^\ka(\wh V(\cdot/N^{1-\ka})\star\wh f_N)_p}- (\g_p^2+\s_p^2) p^2\\
		& \hspace{4cm} - (\g_p+\s_p)^2 N^\ka(\wh V(\cdot/N^{1-\ka})\star\wh f_N)_p\Bigg].
	} 
Recalling that $\g_p = 1$ and $\s_p = 0$ for $p \in P_H^c$, we find that $p^2 \s_p^2 - (1/2) (\g_p^2 + \s_p^2) p^2 = -p^2 /2$, for all $p \in \Lambda^*_+$. Using the bound 
	\[
	\sum_{p\in\retp}\wh | V((p-q)/N^{1-\ka})| |\eta_p| \le CN,
	\]  we find  	
\[ \fra{1}{N}\sum_{u\in\retp}\s_u^2\sum_{p\in\ret}N^\ka\abs{\wh V(p/N^{1-\ka})\h_p} \le C N^{3\ka-\a} \]
as well as
	\spl{
		\sum_{p\in P_H}N^\ka&\wh V(p/N^{1-\ka})\left(\s_p\g_p+\s_p^2\right) -\fra{1}{2}\sum_{p\in\retp} (\g_p+\s_p)^2 N^\ka(\wh V(\cdot/N^{1-\ka})\star\wh f_N)_p\\
		=\,&-\fra{1}{2}\sum_{p\in\retp}N^\ka(\wh V(\cdot/N^{1-\ka})\star\wh f_N)_p-\fra{N^\ka}{N}\sum_{p\in P_H}(\g_p\s_p+\s_p^2)(\wh V(\cdot/N^{1-\ka})\star\h)_p\\
		=\,&-\fra{1}{2}\sum_{p\in\retp}N^\ka(\wh V(\cdot/N^{1-\ka})\star\wh f_N)_p-\fra{N^\ka}{N}\sum_{p\in P_H}\g_p\s_p(\wh V(\cdot/N^{1-\ka})\star\h)_p+ \cO(N^{3\ka-\a})\,.
	} 
Thus, we have 	
	\bes{ \label{eq:CMN_step2}
		C_{\cM_N}=\,&\fra{N-1}{2}N^\ka \wh V(0)-\fra{N^\ka}{N}\sum_{p\in P_H}\g_p\s_p(\wh V(\cdot/N^{1-\ka})\star\h)_p\\
		&+\fra{1}{N}\sum_{p\in P_H}\left[p^2\eta_p^2+\fra{N^\ka}{2N}(\wh V(\cdot/N^{1-\ka})\star\h)_p\h_p\right]-\sum_{p\in\retp}\fra{N^{2\ka}(\wh V(\cdot/N^{1-\ka})\star\wh f_N)^2}{4p^2}\\
		&+\fra{N^\ka}{2N}\sum_{p,q\in P_H}\wh V((p-q)/N^{1-\ka})\g_q\s_q\g_p\s_p+E_{Bog,N}+\cO(N^{3\ka-\a})\,,
	} where we defined	
	\spl{
		E_{Bog,N}=\fra{1}{2}\sum_{p\in\retp}\Bigg[&\sqrt{|p|^4+2p^2N^\ka(\wh V(\cdot/N^{1-\ka})\star\wh f_N)_p}-p^2\\
		&-N^\ka(\wh V(\cdot/N^{1-\ka})\star\wh f_N)_p+\fra{N^{2\ka}(\wh V(\cdot/N^{1-\ka})\star\wh f_N)^2}{2p^2}\Bigg]=\fra{1}{2}\sum_{p\in\retp}e_{p, N}\,.
	} To compare $E_{Bog,N}$ with its limiting value $E_{Bog}$, we first compare the summands $e_{p,N}$ in $E_{Bog,N}$ with the corresponding summands 
\[
e_p=\sqrt{\abs p^4+16\pi\fa_0N^\ka\abs p^2}-\abs{p}^2-8\pi\fa_0N^\ka+\fra{(8\pi\fa_0N^\ka)^2}{2\abs p^2}
\] in $E_{Bog}$. On one hand, Taylor expanding the square root we see that	
	\[
	\abs{e_{p,N}},\,\abs{e_p}\le CN^{6\ka}\abs{p}^{-4}\,,
	\] which yields	
	\be\label{eq:EBog_compare_1}
	\fra{1}{2} \sum_{\abs p>N} |e_{p,N}-e_p| \le CN^{6\ka-1}\,,
	\ee for a constant $C$ independent of $p, N$. On the other hand, \eqref{eq:CGM_approx_convo} implies, expanding once again the square roots in $e_p, e_{p,N}$,	
	\be\label{eq:EBog_compare_2}\abs{e_p-e_{p,N}}\le \fra{CN^{6\ka}}{\abs p^4}\abs{(8\pi\fa_0)^3-(\wh V(\cdot/N^{1-\ka})\star\wh f_N)_p^3}\le\fra{CN^{8\ka-1}}{\abs p^3}\,.\ee	
	Combining \eqref{eq:EBog_compare_1} and \eqref{eq:EBog_compare_2} we get, for $N$ large enough,	
	\begin{equation}\label{eq:CMN_step2_EBog}
	\abs{E_{Bog}-E_{Bog, N}}\le CN^{8\ka-1}\log N \leq   CN^{-1+10\ka}.\end{equation}
	We now analyze the remaining terms on the right-hand side of \eqref{eq:CMN_step2}. Using the scattering equation \eqref{eq:scattering_eta}, the bound \eqref{eq:chi_f_bound} 
	and the approximation \eqref{eq:approx_8pia_0} we find
	\bes{\label{eq:CMN_step2_orderNka}
		-\fra{N^\ka}{2}\wh V(0)&+\fra{1}{N}\sum_{p\in P_H}\left[p^2\eta_p^2+\fra{N^\ka}{2N}(\wh V(\cdot/N^{1-\ka})\star\h)_p\h_p\right]\\
		=\,&	-\fra{N^\ka}{2}\wh V(0)-\fra{N^\ka}{2N}\sum_{p\in\ret}\wh V(p/N^{1-\ka})\h_p+\cO(N^{2\ka+\a-1})\\
		=\,&-\fra{1}{2}N^\ka(\wh V(\cdot/N^{1-\ka})\star \wh f_N)_0+\cO(N^{2\ka+\a-1})\\
		=\,&-4\pi N^\ka\fa_0+\cO(N^{2\ka+\alpha-1}),
	} with the remainder $\cO(N^{2\ka+\alpha-1})$ arising from the missing low momenta in the sum on the first line. Next, we combine the second and the fifth term  on the r.h.s. of  \eqref{eq:CMN_step2}. To this end, we write 
\begin{equation}\label{eq:trick_gpspgqsq_step1}
\begin{split}
-\fra{N^\ka}{N}\sum_{p\in P_H}\g_p\s_p&(\wh V(\cdot/N^{1-\ka})\star\h)_p+\fra{N^\ka}{2N}\sum_{p,q\in P_H}\wh V((p-q)/N^{1-\ka})\g_q\s_q\g_p\s_p\\
=&\fra{N^\ka}{N}\sum_{p\in P_H}\sum_{q\in P_H}\wh V((p-q)/N^{1-\ka}) \Big[ \fra{1}{2}\g_p\s_p\g_q\s_q-\g_p\s_p\h_q \Big] \\
&-\fra{N^\ka}{N}\sum_{p\in P_H}\sum_{q\in P_H^c}\wh V((p-q)/N^{1-\ka})\g_p\s_p\h_q.
\end{split}
\end{equation}
To deal with the first term, we write 
\[
\begin{split}
\fra{1}{2}\g_p\s_p&\g_q\s_q-\g_p\s_p\h_q\\
=&\fra{1}{2}(\g_p\s_p-\h_p+\h_p)(\g_q\s_q-\h_q+\h_q)-\Big[(\g_p\s_p-\h_p)\h_q+\h_p\h_q\Big]\\=&\fra{1}{2}(\g_p\s_p-\h_p)(\g_q\s_q-\h_q)+\fra{1}{2}\Big[ (\g_q\s_q-\h_q)\h_p- (\g_p\s_p-\h_p)\h_q\Big]-\fra{1}{2}\h_p\h_q.
\end{split}
\]
With the bound $\abs{\g_p\s_p-\h_p}\le C\abs{p}^{-5}$ (and noticing, exchanging $p$ and $q$, that the contribution of the terms in square brackets vanishes), we obtain 
\begin{equation*}
\begin{split}
\fra{N^\ka}{N}\sum_{p\in P_H}\sum_{q\in P_H}\wh V(&(p-q)/N^{1-\ka})\Big[ \fra{1}{2}\g_p\s_p\g_q\s_q-\g_p\s_p\h_q\Big] \\
=&-\fra{N^\ka}{2N}\sum_{p\in P_H}\sum_{q\in P_H}\wh V((p-q)/N^{1-\ka})\h_p\h_q +\cO(N^{-1+\ka}).
\end{split}
\end{equation*} Inserting this into \eqref{eq:trick_gpspgqsq_step1} we find (recall $f_N = 1 + N^{-1} \eta$) 	
\spl{\label{eq:trick_gpspgqsq_step2}
-\fra{N^\ka}{N}\sum_{p\in P_H}\g_p\s_p&(\wh V(\cdot/N^{1-\ka})\star\h)_p+\fra{N^\ka}{2N}\sum_{p,q\in P_H}\wh V((p-q)/N^{1-\ka})\g_q\s_q\g_p\s_p\\
		=&-\fra{N^\ka}{2}\sum_{p\in P_H}\left[(\wh V(\cdot/N^{1-\ka})\star\wh f_N)_p-\wh V(p/N^{1-\ka})\right]\h_p\\
		&-\fra{N^\ka}{N}\sum_{p\in P_H}\sum_{q\in P_H^c}\wh V((p-q)/N^{1-\ka})\left[(\g_p\s_p-\h_p)\h_q+\fra{1}{2}\h_p\h_q\right]+ \cO(N^{-1+\ka}),
	}
In the second term, the contribution proportional to $(\g_p \s_p - \eta_p) \eta_q$ is small, of order $N^{-1+\ka}$; the other contribution leads, adding terms with $p \in P_H^c$ (producing an error of 
order $N^{-1+3\ka +2\alpha}$) to: 
\bes{\label{eq:CMN_step2_convolutions}
-\fra{N^\ka}{N}\sum_{p\in P_H}\g_p\s_p&(\wh V(\cdot/N^{1-\ka})\star\h)_p+\fra{N^\ka}{2N}\sum_{p,q\in P_H}\wh V((p-q)/N^{1-\ka})\g_q\s_q\g_p\s_p\\
		=-\fra{N^\ka}{2}\sum_{p\in P_H}&\left[(\wh V(\cdot/N^{1-\ka})\star\wh f_N)_p-\wh V(p/N^{1-\ka})\right]\h_p\\
		&-\fra{N^\ka}{2N}\sum_{q\in P_H^c} \convo{\h}_q\h_q+ \cO(N^{-1+3\ka+2\a})\,.
	} The second term in the square bracket can be combined with the leading term in \eqref{eq:CMN_step2}. In fact, we find 
\[ \begin{split} 
		\fra{N^{1+\ka}}{2}\Big[ \wh V(0)&+\fra{1}{N}\sum_{p\in P_H}\wh V(p/N^{1-\ka})\h_p\Big]\\
		=\,&\fra{N^{1+\ka}}{2}\convo{\wh f_N}_0-\fra{N^\ka}{2}\sum_{p\in P_H^c}\hat V(p/N^{1-\ka})\h_p\\
		=\,&4\pi\fa_0N^{1+\ka}+\fra{N^{1+\ka}}{2}\left[\convo{\wh f_N}_0-8\pi\fa_0\right] -\fra{N^\ka}{2}\sum_{p\in P_H^c}\widehat V(p/N^{1-\ka})\h_p \, .
\end{split} \] 
The last term can be combined with the other terms on the r.h.s. of \eqref{eq:CMN_step2_convolutions}. We find 
\[ \begin{split}  -\frac{N^\kappa}{2} \sum_{p \in P_H} (\widehat{V} (\cdot / N^{1-\ka}) &* \widehat{f}_N)_p \h_p - \frac{N^\ka}{2N} \sum_{p \in P_H^c} (\widehat{V} (\cdot / N^{1-\ka}) * \eta)_p \eta_p - \frac{N^\kappa}{2} \sum_{p \in P_H^c} \widehat{V} (p/N^{1-\ka}) \eta_p \\ &= -\frac{N^\ka}{2} \sum_{p \in \Lambda^*_+} (\widehat{V} (\cdot / N^{1-\ka}) * \widehat{f}_N)_p \eta_p - \frac{N^\ka}{2} (\widehat{V} (\cdot / N^{1-\ka}) * \widehat{f}_N)_0 \eta_0. \end{split}\]

Combining the last equations with \eqref{eq:CMN_step2_EBog}, \eqref{eq:CMN_step2_orderNka}, \eqref{eq:CMN_step2_convolutions}, Eq. \eqref{eq:CMN_step2} implies that 
\spl {
		\cC_{\cM_N}=&4\pi\fa_0N^{\ka}(N-1)+E_{Bog}\\
		&-\fra{1}{2}\sum_{p\in\retp}N^\ka\convo{\wh f_N}_p\h_p-\sum_{p\in\retp}\fra{N^{2\ka}\convo{\wh f_N}_p^2}{4p^2}\\
		&-\fra{1}{2}N^\ka\convo{\wh f_N}_0\h_0+\fra{N^{1+\ka}}{2}\left[\convo{\wh f_N}_0-8\pi\fa_0\right]\\
		&+ \cO(N^{-1+3\ka+2\a}+N^{3\ka-\a}).
	} Using again the scattering equation \eqref{eq:scattering_fN_fourier} and the approximation \eqref{eq:approx_8pia_0}, we get 
	\spl{
		C_{\cM_N}=&-4\pi\fa_0N^\ka(N-1)+E_{Bog} -\sum_{p\in\retp}\fra{N^{3-\ka}\lambda_l\convo{\wh f_N}_p}{2p^2} \, (\wh \c_\ell \star\wh f_N)_p\\
		&+\fra{N^{1+\ka}}{2}\left[\convo{\wh f_N}_0-8\pi\fa_0\right]-4\pi\fa_0N^\ka\h_0 + \cO(N^{-1+3\ka+2\a}+N^{3\ka-\a}).
	} 
By Lemma \ref{lem:properties_scattering_function}, the definition of $\h_0$ and \eqref{eq:wl_integral} we have 
	\[\begin{split}
	\fra{N^{1+\ka}}{2}\left[\convo{\wh f_N}_0-8\pi\fa_0\right] &=\fra{6\pi\fa_0^2N^{2\ka}}{\ell}+\cO(N^{-1+3\ka}), \\
	-4\pi\fa_0N^\ka\h_0 &=\fra{8(\pi\fa_0N^\ka)^2}{5} \ell^2+\cO(N^{-1+3\ka})\,,
	\end{split}\] 
and using \eqref{eq:CGM_approx_convo} and the bound $|(\widehat{\chi}_\ell * \eta)_p| \leq C N^\kappa/ p^2$ (see argument before \eqref{eq:chi_f_bound}) we get 
	\spl{
		-\sum_{p\in\retp}&\fra{N^{3-\ka}\convo{\wh f_N}_p}{2p^2}\l_\ell (\wh \c_\ell \star\wh f_N)_p\\
		&=-\fra{3\fa_0}{\ell^3} \sum_{p\in\retp}\fra{N^{2\ka}\convo{\wh f_N}_p}{2p^2}(\wh \c_\ell \star\wh f_N)_p+ \cO(N^{-1+3\ka})\\
		&=-\fra{12\pi\fa_0^2N^{2\ka}}{\ell^3}	\sum_{p\in\retp}\fra{(\wh \c_\ell \star\wh f_N)_p}{p^2}+ \cO(N^{-1+3\ka})\\
		&=-\fra{12\pi\fa_0^2N^{2\ka}}{\ell^3}	\sum_{p\in\retp}\fra{\wh \c_\ell (p)}{p^2}+ \cO(N^{-1+3\ka} ).
	} In conclusion, we have 
	\spl{
		C_{\cM_N}=\,&4\pi\fa_0N^\ka(N-1)+E_{Bog} \\
		&+ 6\pi\fa_0^2N^{2\ka}\left[ \fra{1}{\ell}+\fra{4}{15}\pi \ell^2-\fra{2}{\ell^3}\sum_{p\in\retp}\fra{\wh\c_\ell (p)}{p^2}\right]+ \cO(N^{-1+3\ka+2\a}+N^{3\ka-\a})\\
		=\,&4\pi\fa_0N^\ka(N-1)+E_{Bog} + 6\pi\fa_0^2N^{2\ka}I_\ell+ \cO(N^{-1+3\ka+2\a}+N^{3\ka-\a})\,.
	}  It is possible to show that	
	\[
	I_\ell =\fra{4\pi \ell^2}{3}-\fra{8\pi}{3}\lim_{M\to\infty}\sum_{\substack{p\in\retp\\\abs{p_1},\abs{p_2},\abs{p_3}\le M}}\fra{\cos(\ell \abs p)}{\abs p^2}
	\] and, in particular, that the limit on the r.h.s. exists for every $\ell \in(0,1/2)$. Furthermore, this quantity is in fact independent of the particular choice of $\ell \in (0;1/2)$. This is proved in \cite[Lemma 5.4]{BBCS4} and, choosing for instance $\ell = 1/(2\pi)$, it implies (\ref{eq:const-claim}) (recall the definition (\ref{eq:eLambda}) of $e_\Lambda$).
\end{proof}

\section{Proof of Theorem \ref{thm:main}}\label{sec:main_thm}
In this section we prove our main result. We assume here that the parameters $\kappa, \alpha, \beta > 0$ satisfy \eqref{eq:conditions_parameters_JN} and also the conditions $6\kappa < \alpha < 1/2 -3\kappa/2$ and $\kappa \in [0;1/44)$, so that Theorem \ref{thm:mainABS} holds true. 
 
We set \[E_{\cM_N}=4\pi\fa_0N^\ka(N-1)+E_{Bog}+e_\L(\fa_0N^\ka)^2\] and we consider the diagonal operator 
\[\cD=\sum_{p\in\retp}\e_p\,a_p^*a_p,\]
with $\e_p=\sqrt{|p|^4+16\pi\fa_0N^\ka |p|^2}$. With this notation, Proposition \ref{prop:MN_Def} reads 
\begin{equation}\label{eq:final_MN}
\begin{split}
	\cM_N-E_{\cM_N}= \cD+ \cV_N + \cE_{\cM_N} 
\end{split}
\end{equation} 
where the error term $\cE_{\cM_N}$ satisfies the bounds \eqref{eq:EMN_bound_1} and  \eqref{eq:EMN_bound_2}.

To prove a lower bound on the ground state energy of $\cM_N$ (and later on its excited eigenvalues), we need a-priori estimates on the number and energy of excitations in low-energy states. Suppose 
that $\xi_N \in \cF_+^{\leq N}$, with $\| \xi_N \| = 1$ and 
\begin{equation}\label{eq:lowen} \xi_N = \chi ( \cM_N - E_N \leq  N^{\kappa/2 + \mu} ) \xi_N \end{equation} 
for some $\mu > 0$. Recalling that $\cM_N = e^{-T} e^{-A} e^{-B} U H_N U^* e^B e^A e^T$ and 
defining $\psi_N = U_N^* e^B e^A e^T \xi_N \in L^2_s (\L^N)$, we find $\| \psi_N \| =1$ and 
\[ \psi_N = \chi (H_N - E_N \leq N^{\kappa/2+\mu} ) \psi_N \]
From Theorem \ref{thm:mainABS} we conclude therefore that 
\[ \begin{split}  \langle e^{-B} U_N \psi_N , (\cH_N + 1) (\cN_+ +1) e^{-B} U_N \psi_N \rangle &\leq C\Big[ N^{21\kappa+\eps}\zeta^2 + N^{44\kappa+2\eps}\Big]^{2}
\\ 
\langle e^{-B} U_N \psi_N , (\cN_+ +1)^3 e^{-B} U_N \psi_N \rangle &\leq C \Big[ N^{21\kappa + 2\mu +\eps} + N^{44\kappa+2\eps}\Big]^{3} 
\end{split} \]
With $e^{-B} U_N \psi_N = e^A e^T \xi_N$ and with Lemma  \ref{lm:growNA} we arrive at  
\[ \begin{split}  \langle e^T \xi_N ,  (\cH_N + 1) (\cN_+ +1) e^T \xi_N \rangle &\leq C\Big[ N^{21\kappa + 2\mu +\eps} + N^{44\kappa+2\eps}\Big]^{3} \\
\langle e^T \xi_N ,  (\cN_+ +1)^3 e^T \xi_N \rangle &\leq C\Big[ N^{21\kappa+2\mu +\eps} + N^{44\kappa+2\eps}\Big]^{3} \end{split}
\]
If $\kappa, \mu > 0$ are small enough, we can use this and \eqref{eq:EMN_bound_1} with an appropriate choice of $\alpha, \beta >0$ (making sure, in particular, that \eqref{eq:conditions_parameters_JN} holds true) to show that there exists $\eps > 0$ with 
\begin{equation}\label{eq:lower-eps} \langle \xi_N, \cE_{\cM_N} \xi_N \rangle \leq C N^{-\eps} 
\end{equation} 
for all $\xi_N \in \cF_+^{\leq N}$ satisfying (\ref{eq:lowen}). 

If $\xi_N$ denotes now the ground state of the operator $\cM_N$, we have $\cM_N \xi_N = E_N \xi_N$ and (\ref{eq:lowen}) is certainly satisfied. With (\ref{eq:final_MN}) and (\ref{eq:lower-eps}), we obtain 
\begin{equation}\label{eq:low-gs} E_N \geq E_{\cM_N} -C N^{-\eps} \end{equation} 
if $\kappa , \eps > 0$ are small enough.  Testing (\ref{eq:final_MN}) on the vacuum and using the bound (\ref{eq:EMN_bound_2}) to bound the vacuum expectation of $\cE_{\cM_N}$ (choosing again $\alpha, \beta > 0$ appropriately), we also find the upper bound
\begin{equation}\label{eq:up-gs} E_N \leq E_{\cM_N} + C N^{-\eps} \end{equation} 
if $\kappa, \eps > 0$ are small enough. We conclude that $|E_N - E_{\cM_N}| \leq C N^{-\eps}$. 

Let us now study excitations. We denote by $\lambda_j$ the eigenvalues of $\cM_N-E_{N}$ and by $\n_j$ those of $\cD$, indexed in increasing order. Assume $\lambda_j \leq N^{\kappa/2+ \mu}$. Then, with the notation $P_\zeta = \chi (\cM_N - E_N \leq N^{\kappa/2+\mu}) \cF^{\leq N}_+$ for the spectral subspace of $\cM_N$, we have (applying the min-max principle for the eigenvalues of $\cM_N - E_N$)
\[ \begin{split} 
	\l_j=\; & \inf_{\substack{Y\subset \cF^{\leq N}_+ \\ \dim (Y)=j}} \sup_{\substack{\xi_N \in Y\\\norm{\xi_N}=1}}\expec{\xi_N}{(\cM_N-E_{N})} \\ =\; &\inf_{\substack{Y\subset P_\zeta \\ \dim (Y)=j}} \sup_{\substack{\xi_N\in Y\\\norm{\xi_N}=1}}\expec{\xi_N}{(\cM_N-E_{N})} 
	\\ \geq \; &\inf_{\substack{Y\subset P_\zeta \\ \dim (Y)=j}} \sup_{\substack{\xi_N\in Y\\\norm{\xi_N}=1}}\expec{\xi_N}{(\cM_N-E_{\cM_N})} - C N^{-\eps} 
\end{split} \]
if $\kappa , \eps > 0$ are small enough. Here we used the upper bound (\ref{eq:up-gs}) for $E_N$. With (\ref{eq:final_MN}) and using the positivity of $\cV_N$, we obtain
\[ \begin{split} 
	\l_j \geq \; &\inf_{\substack{Y\subset P_\zeta \\ \dim (Y)=j}} \sup_{\substack{\xi_N\in Y\\\norm{\xi_N}=1}}\expec{\xi_N}{(\cD + \cE_{\cM_N})} - C N^{-\eps} 
\end{split} \]
From (\ref{eq:lower-eps}), we conclude that
\begin{equation}\label{eq:low-fin}  \begin{split} 
	\l_j \geq \; &\inf_{\substack{Y\subset P_\zeta \\ \dim (Y)=j}} \sup_{\substack{\xi_N\in Y\\\norm{\xi_N}=1}}\expec{\xi_N}{\cD} - C N^{-\eps}   
	\\ \geq \; &\inf_{\substack{Y\subset \cF_+^{\leq N} \\ \dim (Y)=j}} \sup_{\substack{\xi_N\in Y\\\norm{\xi_N}=1}}\expec{\xi_N}{\cD} - C N^{-\eps} = \nu_j - C N^{-\eps} \end{split} \end{equation} 
if $\kappa , \eps > 0$ are small enough. 

Finally, we need to establish upper bounds for the eigenvalues $\lambda_j$. To this end, we are going to use eigenvectors of $\cD$ as trial states. Notice that the eigenvalues of $\cD$ have the form 
\begin{equation}\label{eq:nuj} 
\n_j=\sum_{p\in\retp}n_p^{(j)}\e_p\,,
\end{equation} 
for a sequence $\{ n_p^{(j)} \}_{p \in \Lambda^*_+}$ with $n_p^{(j)} \in\bN$. An eigenvector of $\cD$ associated with the eigenvalue (\ref{eq:nuj}) has the form  
\be\label{eq:defthj}
\th_j=C_j\prod_{p\in\retp}(a_p^*)^{n_p^{(j)}}\O\,
\ee for appropriate normalization constants $C_j$ ($\O$ denotes as usual the vacuum in Fock space). If $\n_j$ is degenerate, the choice of $\th_j$ is not unique, but in the following we work exclusively with eigenvectors of the form (\ref{eq:defthj}). We are going to need the following lemma, which controls the expectation of $\cV_N$ on spaces spanned by the eigenvectors $\theta_j$. 
\begin{lemma}\label{lem:VNBoundD}
	Let $\th_1,...,\th_m$ be normalized eigenvectors for $\cD$ as in \eqref{eq:defthj}, corresponding to its first $m$ eigenvalues $\n_1< ...\le\n_m<  N^{\ka/2 + \mu}$. Then there exists a constant $C>0$ such that
	\[
	\expec{\xi_N}{\cV_N}\le CN^{-1+9\kappa/4 + 5\mu/2}
	\] for all $\xi_N \in\mathrm{Span}(\th_1,...,\th_m)$ with $\| \xi_N \|=1$.
\end{lemma}

\begin{proof}
The proof is very similar to the proof of \cite[Lemma 6.1]{BBCS4}. Since $\e_p> p^2$, we have $a_p\th_j=0$ if $\abs{p}\ge N^{\ka/4+\mu/2}$ and $j\in \{1,...,m\}$, hence also for all $\xi_N \in\mathrm{Span}(\th_1,...,\th_m)$. This implies with $\cN_+\le C\cD$ and $[\cN_+,\cD]=0$ that
	\spl{
		\expec{\xi_N}{\cV_N}=\,&\fra{N^\ka}{2N}\sum_{\substack{p,q,r\in\retp,r\neq-p,-q\\\abs{p},\abs{q},\abs{r}<2 N^{\ka/4 + \mu/2}}} \wh V(r/N^{1-\ka})\norm{a_{p+r}a_q\xi_N}\norm{a_{q+r}a_p\xi_N}\\
		\le\,&\fra{CN^\ka}{N} \sum_{\substack{r\in\retp\\ \abs r \le N^{\ka/4+\mu/2}}}\norm{(\cN_++1)\xi_N}^2\\ \le \, &CN^{-1+ 7\kappa/4+ 3\mu/2} \norm{(\cD+1)\xi_N}^2\leq C N^{-1+ 9\kappa/4+ 5\mu/2} } \end{proof}
	
Suppose now that $\lambda_j \leq N^{\kappa/2+\mu}$. Notice that, from the lower bound (\ref{eq:low-fin}), this also implies that $\nu_j \leq C N^{\kappa/2+\mu}$. By the min-max principle, we have 
\[ \begin{split}  \l_j = \,& \inf_{ \substack{Y \subset \cF_+^{\leq N}  \\ \dim (Y)= j} } \sup_{\substack{\xi_N\in Y\\\norm{\xi_N}=1}}\expec{\xi_N}{(\cM_N-E_N)} \le\,\sup_{\substack{\xi_N\in\mathrm{Span}(\th_1,...,\th_j)\\ \norm{\xi_N}=1}}\expec{\xi_N}{(\cM_N-E_{N})} \\ \le\, &\sup_{\substack{\xi_N\in\mathrm{Span}(\th_1,...,\th_j)\\ \norm{\xi_N}=1}}\expec{\xi_N}{(\cM_N-E_{\cM_N})} + N^{-\eps} 
\end{split} \]
by (\ref{eq:low-gs}). With (\ref{eq:final_MN}) and Lemma \ref{lem:VNBoundD}, we obtain 
\begin{equation*}
\begin{split} 
 \l_j  \leq \, & \n_j + \sup_{\substack{\xi_N\in\mathrm{Span}(\th_1,...,\th_j)\\\norm{\xi_N}=1}}\Bigg[\expec{\xi_N}{\cV_N}+\expec{\xi_N}{\cE_{\cM_N}}\Bigg]\\
	\le\,&\n_j+C N^{-\eps} +\sup_{\substack{\xi_N\in\mathrm{Span}(\th_1,...,\th_j)\\\norm{\xi_N}=1}} \expec{\xi_N}{\cE_{\cM_N}}. \end{split} \end{equation*} 
if $\kappa,\mu, \eps > 0$ are small enough. Using now (\ref{eq:EMN_bound_2}) (with an appropriate choice of $\alpha,\beta > 0$), estimating, for $n=0,1,2$, 
\[ \cN_+^{n+1} \leq C \cK \cN^n_+ \leq C \cD^{n+1}  \]  
and observing that $\cD^{n+1} \leq \nu_j^{n+1}  \leq C N^{(\kappa/2 + \mu) (n+1)}$ on $\mathrm{Span}(\th_1,...,\th_j)$, since $\nu_j \leq C N^{\kappa/2 + \mu}$, we obtain that $\langle \xi_N, \cE_{\cM_N} \xi_N \rangle \leq C N^{-\eps}$ for all normalized $\xi_N \in \mathrm{Span}(\th_1,...,\th_j)$. Thus 
\[   \l_j  \leq \nu_j + CN^{-\eps} \] 
if $\kappa,\eps > 0$ are small enough (in (\ref{eq:EMN_bound_2}) we can use, for example, $\delta = N^{4\kappa}$; we can then estimate then expectation of $\delta \cV_N$, on $\mathrm{Span}(\th_1,...,\th_j)$, with Lemma \ref{lem:VNBoundD}). \qed


\appendix



\section{Analysis of $\cG_N$}
\label{sec:quadratic}

In this section, we sketch the proof of Proposition \ref{prop:GN}. The proof goes along the same lines as the one of \cite[Prop. 3.2]{BBCS4}, taking into account the different scaling. We write
\[\cG_N=\cG_N^{(0)}+\cG_N^{(2)}+\cG_N^{(3)}+\cG_N^{(4)}\,,\] with $\cG_N^{(i)}=e^{-B}\cL_N^{(i)}e^{B}$ for $i=0,2,3,4$ and with $B$ as defined in \eqref{eq:B-def}; we analyze these terms individually in the next subsections. First, we need some rough bound to control the growth of operators of the form $(\cH_N+1)(\cN_++1)^j$ under the action of $B$. In the next lemma, as well as in the rest of the section, we assume $\a>2\ka$. 
\begin{lemma}\label{lem:B_bounds_HN}
For every $j\in\bN$ there exists a constant $C$ such that: 
	\bes{\label{eq:B_bounds_HN}
		e^{-B}\cK(\cN_++1)^je^{B}&\le C\cK(\cN_++1)^j+CN^{1+\ka}(\cN_++1)^{j+1}\\
		e^{-B}\cV_N(\cN_++1)^je^{B}&\le C\cV_N(\cN_++1)^j+CN^{1+\ka}(\cN_++1)^{j}\,.
	}
\end{lemma}

The proof is analogous to that of \cite[Lemma 7.1]{BBCS4} and we omit it. We have to take into account the different scaling, producing the growth $\norm{\eta}_{H^1}^2\leq C N^{1+\ka}$ instead of $N$.

\subsection{Analysis of $\cG_N^{(0)}$}

From \eqref{eq:excitation_hamiltonian} and Lemma \ref{lem:B_bound_N_easy} we immediately obtain 
\be\label{eq:GN0}
\cG_N^{(0)}=\fra{N-1}{2}N^\ka\wh V(0)+\cE_{\cG_N}^{(0)}\,
\ee with 
\be\label{eq:GNO_error}
\pm\cE_{\cG_N}^{(0)}\le CN^{-1+\ka}(\cN_++1)^2\,.
\ee

\subsection{Analysis of $\cG_N^{(2)}$}

First we write $\cL_N^{(2)}=\cK+\cL_N^{(2,V)}$, with 
\[\cL_N^{(2,V)}=\sum_{p\in\retp}N^\ka\wh V(p/N^{1-\ka})\left[b_p^*b_p-\fra{1}{N}a_p^*a_p\right]+\fra{N^\ka}{2}\sum_{p\in\retp}\wh V(p/N^{1-\ka})[b_p^*b_{-p}^*+b_pb_{-p}]\,,\] and we proceed to analyze the terms individually. We start with the kinetic energy.
\begin{lemma}\label{lem:B_kinetic}
	We have
	\bes{\label{eq:B_kinetic}
		e^{-B}\cK e^{B}=\,&\cK+\sum_{p\in\retp}\Bigg[p^2\s_p^2\left(1-\fra{\cN_+}{N}+\fra{1}{N}\right)+2p^2\s_p^2b_p^*b_p+p^2\g_p\s_p(b_p^*b_{-p}^*+\hc)\Bigg]\\
		&+\sum_{p\in\retp}\fra{1}{N}p^2\s_p^2\sum_{q\in\retp}\Bigg[(\g_q^2+\s_q^2)b_q^*b_q+\g_q\s_q(b_q^*b_{-q}^*+\hc)+\s_q^2\Bigg]\\
		&+\sum_{p\in\retp}[p^2\s_pb_{-p}d_p+\hc]+\cE_{\cG_N}^{(\cK)}\,,
	} with
	\be \label{eq:B_kinetic_error}
	\pm\cE_{\cG_N}^{(\cK)}\le N^{-\fra{1}{2}+\ka}(\cK+\cN_+^2+1)(\cN_++1)\,.
	\ee
\end{lemma}
\begin{proof}
	We first write $\cK$ in terms of the operators $b_p,b_p^*$ in order to apply the expansion \eqref{eq:dp_def}. We find, with $b_p^* b_p = a_p^* (1- \cN_+ / N) a_p$, 
	\bes{\label{eq:B_conj_kinetik}
		e^{-B}\cK e^{B}=\sum_{p\in\retp} p^2 \, e^{-B}b_p^*b_p e^B+\fra{1}{N}\sum_{p,q\in\retp} p^2 \, e^{-B} b_p^*b_q^*b_qb_p e^B+\wt \cE_1\,,
	} where, with Lemma \ref{lem:B_bounds_HN}
	\[
	\pm\wt\cE_1\le C N^{-2}  \, e^{-B} \cK(\cN_++1)^2 e^B \le CN^{-1+\ka}\left[\cK(\cN_++1)+(\cN_++1)^3\right]\,.	
	\]
	Using \eqref{eq:dp_def} we compute 
	\spl{
		\sum_{p\in\retp}p^2\, e^{-B} b_p^*b_p e^B=\,&\sum_{p\in\retp}p^2\left[(\g_p^2+\s_p^2)b_p^*b_p+\g_p\s_p(b_p^*b_{-p}^*+b_pb_{-p})+\s_p^2\left(1-\fra{\cN_+}{N}\right)\right]\\
		&+\sum_{p\in\retp}\Big(p^2\s_pb_{-p}^*d_p+\hc\Big)+\wt\cE_2\,,
	} with
	\[
	\wt\cE_2=\sum_{p\in\retp}p^2\left[\g_p(b_p^*d_p+\hc)+d_p^*d_p-\fra{1}{N}\s_p^2a_p^*a_p\right]\,.
	\] With \eqref{eq:sq_gq_bounds}, \eqref{eq:dp_bounds}, we obtain 
	\[
	\pm\wt\cE_2\le N^{(-1+\ka)/2}\left[(\cK+1)(\cN_++1)+(\cN_++1)^3\right]\,.
	\] 
	Next we consider the second term on the r.h.s. of \eqref{eq:B_conj_kinetik}. We first consider the operator 
	\spl{
		D&=\sum_{q\in\retp} e^{-B} b_q^*b_q e^B =\sum_{q\in\retp}\left[(\g_q^2+\s_q^2)b_q^*b_q+\g_q\s_q(b_q^*b_{-q}^*+b_qb_{-q})+\s_q^2\right]+\wt\cE_3\,,	
	} with $\pm\wt\cE_3\le CN^{-1}(\cN_++1)^2$. From Lemma \ref{lem:B_bound_N_easy}, we also have  $\pm D\le C(\cN_++1)$. From these bounds we find 
	\spl{
		\fra{1}{N}\sum_{p,q\in\retp}p^2\, e^{-B} b_p^*b_q^*b_qb_p e^B =\,&\frac{1}{N}\sum_{p\in\retp}p^2(\g_pb_p^*+\s_pb_{-p}+d_p^*)D(\g_pb_p+\s_pb_{-p}^*+d_p)\\
		=\,&\fra{1}{N}\sum_{p\in\retp}p^2\s_p^2b_{p}Db_p^*+\wt\cE_4\\
		=\,&\fra{1}{N}\sum_{p,q\in\retp}p^2\s_p^2(\g_q^2+\s_q^2)b_{p}b_q^*b_qb_p^*+\fra{1}{N}\sum_{p,q\in\retp}p^2\s_p^2\s_q^2b_pb_p^*\\
		&+\fra{1}{N}\sum_{p,q\in\retp}p^2\s_p^2\g_q\s_qb_{p}(b_q^*b_{-q}^*+b_qb_{-q})b_p^*+\wt\cE_5,
	} where 
	\[
	\pm\wt\cE_5\le CN^{-\fra{1}{2}+\ka}(\cK+\cN_+^2+1)(\cN_++1)\,.
	\] Bringing the expression to normal order and bounding some additional errors, we find
\[ \begin{split} 
		\fra{1}{N}\sum_{p,q\in\retp}p^2\, e^{-B} b_p^*b_q^*b_qb_p e^B =\,&\fra{1}{N}\sum_{p,q\in\retp}p^2\s_p^2(\g_q^2+\s_q^2)b_q^*b_q+\fra{1}{N}\sum_{p,q\in\retp}p^2\s_p^2\s_q^2\\
		&+\fra{1}{N}\sum_{p,q\in\retp}p^2\s_p^2\g_q\s_q(b_q^*b_{-q^*}+\hc) +\fra{1}{N}\sum_{p\in\retp}p^2\s_p^2+\wt\cE_6\,, \end{split} \] 
with 	 \[\pm{}\wt\cE_6\le CN^{-\fra{1}{2}+\ka}(\cK+\cN_+^2+1)(\cN_++1)\,.
	\] This concludes the proof of the lemma. 
\end{proof}

We proceed with the potential term $\cG_N^{(2,V)}= e^{-B} \cL_N^{(2,V)} e^B$.
\begin{lemma}\label{lem:B_2V}
	We have
	\bes{\label{eq:B_2V}
		\cG_N^{(2,V)}=\,&N^\ka\sum_{p\in\retp}\wh V(p/N^{1-\ka})\s_p^2+N^\ka\sum_{p\in\retp}\wh V(p/N^{1-\ka})\g_p\s_p(1-\cN_+/N)\\
		&+\fra{N^\ka}{2}\sum_{p\in\retp}\wh V(p/N^{1-\ka})(\g_p+\s_p)^2(2b_p^*b_p+b_p^*b_{-p}^*+b_pb_{-p})\\
		&+\fra{N^\ka}{2}\sum_{p\in\retp}\wh V(p/N^{1-\ka})\Big[(\g_pb_p^*+\s_pb_{-p})d_{-p}^*+d_p^*(\g_pb_{-p}^*+\s_pb_p)\Big]+\hc\\
		&+\cE_{\cG_N}^{(V)}\,,
	} with 
	\be\label{eq:B_2V_error}
	\pm\cE_{\cG_N}^{(V)}\le CN^{-1/2}(\cN_++1)^3\,.
	\ee
\end{lemma}
\begin{proof}
	From the definition of $\cL_N^{(2,V)}$, using Lemma \ref{lem:B_bound_N_easy} and the bounds \eqref{eq:dp_bounds}, we have 
	\[
	\cG_N^{(2,V)}=\fra{N^\ka}{2}\sum_{p\in\retp}\wh V(p/N^{1-\ka})\, e^{-B} [2b_p^*b_p+b_p^*b_{-p}^*+b_pb_{-p}] e^B + \wt \cE_7\,,
	\] with 
	\[
	\pm\wt\cE_7\le N^{-1}(\cN_++1)\,.
	\] Using again the bounds \eqref{eq:dp_bounds} we find 
	\bes{ \label{eq:G2V_Step1}
		\fra{N^\ka}{2}&\sum_{p\in\retp}\wh V(p/N^{1-\ka})\, e^{-B} [2b_p^*b_p+b_p^*b_{-p}^*+b_pb_{-p}] e^B \\
		=N^\ka&\sum_{p\in\retp}\wh V(p/N^{1-\ka})[\g_pb_p^*+\s_pb_{-p}][\g_pb_p+\s_pb_{-p}^*]\\
		+\fra{N^\ka}{2}&\sum_{p\in\retp}\wh V(p/N^{1-\ka})\Big([\g_pb_p^*+\s_pb_{-p}][\g_pb_{-p}^*+\s_pb_p]+\hc\Big)\\
		+\fra{N^\ka}{2}&\sum_{p\in\retp}\wh V(p/N^{1-\ka})\Big([(\g_pb_p^*+\s_pb_{-p})d_{-p}^*+d_p^*(\g_pb_{-p}^*+\s_pb_p)]+\hc\Big)\\
		+\wt \cE_8\,,
	} with
	$
	\pm \wt\cE_8\le N^{-1/2}(\cN_++1)^3\,
	$  (for $\ka<1/2$). Bringing \eqref{eq:G2V_Step1} into normal order, we get
\[ \begin{split} 
		\cG_N^{(2,V)} =\,&N^\ka\sum_{p\in\retp}\wh V(p/N^{1-\ka})\s_p^2+N^\ka\sum_{p\in\retp}\wh V(p/N^{1-\ka})\g_p\s_p(1-\cN_+/N)\\
		&+\fra{N^\ka}{2}\sum_{p\in\retp}\wh V(p/N^{1-\ka})(\g_p+\s_p)^2(2b_p^*b_p+b_p^*b_{-p}^*+b_pb_{-p})\\
		&+N^\ka\sum_{p\in\retp}\wh V(p/N^{1-\ka})\Big([(\g_pb_p^*+\s_pb_{-p})d_{-p}^*+d_p^*(\g_pb_{-p}^*+\s_pb_p)]+\hc\Big)\\
		&+ \wt\cE_7 + \wt\cE_8 + \wt\cE_9\,,
\end{split} \]  with the remainder term 
\[ \wt\cE_9=-\fra{N^\ka}{2N}\sum_{p\in\retp}\wh V(p/N^{-1+\ka})\left[\s_p\g_p( 2a_p^*a_p+a_p^*a_{-p}^*+a_pa_{-p})+2\s_p^2\left(\cN_++a_p^*a_p\right)\right] \] that satisfies $
	\pm\wt\cE_9\le CN^{-1+\ka}(\cN_++1)$. This concludes the proof.	
\end{proof}

\subsection{Analysis of $\cG_N^{(3)}$}
We have 
\[\cG_N^{(3)}=\, e^{-B} \cL_N^{(3)} e^B =\fra{N^\ka}{\sqrt N}\sum_{\substack{p,q\in\retp\\p+q\neq 0}}\wh V(p/N^{1-\ka})\, e^{-B} b_{p+q}^*a_{-p}^*a_q e^B +\hc\,.\]
\begin{lemma}\label{lem:GN3}
	We have  
	\be\label{eq:GN3_Result}
	\cG_N^{(3)}=\fra{N^\ka}{\sqrt N}\sum_{\substack{p,q\in\retp\\p+q\neq 0}}\wh V(p/N^{1-\ka})b_{p+q}^*b_{-p}^*\Big[(\g_qb_q+\s_qb_{-q}^*)+\hc\Big]+\cE_{\cG_N}^{(3)}\,,
	\ee with the estimate 
	
	\be\label{eq:GN3_error_bound}
	\pm\cE_{\cG_N}^{(3)}\le CN^{-1/2+2\ka}(\cV_N+\cN_++1)(\cN_++1) \,.
	\ee
\end{lemma}
\begin{proof}
	First we write everything in terms of generalized creation and annihilation operators. Using again  $a_{-p}^*a_q=b_{-p}^*b_q+a_{-p}^* (\cN_+ / N) a_q$, we find 
	\be\label{eq:GN3_change_b}
	\cG_N^{(3)}=\fra{N^\ka}{\sqrt N}\sum_{\substack{p,q\in\retp\\p+q\neq 0}}\wh V(p/N^{1-\ka})\Big[ e^{-B} b_{p+q}^*b_{-p}^*b_q e^B +\hc\Big] + \wt \cE_1,
	\ee and switching to position space and using \ref{eq:B_bounds_HN} we find 
	\[
	\pm\wt\cE_1\le N^{-1/2}(\cV_N+\cN_++1)(\cN_++1)\,.
	\]
	We can now write the first term of \eqref{eq:GN3_change_b} as 
	\spl{
		\fra{N^\ka}{\sqrt N}&\sum_{\substack{p,q\in\retp\\p+q\neq 0}}\wh V(p/N^{1-\ka})\, e^{-B} b_{p+q}^*b_{-p}^*b_q e^B \\
		&=\fra{N^\ka}{\sqrt N}\sum_{\substack{p,q\in\retp\\p+q\neq 0}}\wh V(p/N^{1-\ka})\, e^{-B} b_{p+q}^* e^B  \, e^{-B} b_{-p}^* e^B \, e^{-B} b_q e^B\, . } Using the expansion \eqref{eq:dp_def}, 
representing the potential in position space and applying the bounds \eqref{eq:etaH_norms}, \eqref{eq:sq_gq_bounds}, \eqref{eq:sx_gx_bounds} \eqref{eq:dp_bounds}, \eqref{eq:dx_bounds} one finds 
	\be\label{eq:GN3_last_step}
	\cG_N^{(3)}=\fra{N^\ka}{\sqrt N}\sum_{\substack{p,q\in\retp\\p+q\neq 0}}\wh V(p/N^{1-\ka})\g_{p+q}\g_pb_{p+q}^*b_{-p}^*\Big[(\g_qb_q+\s_qb_{-q}^*)+\hc\Big]+\wt\cE_2\,,
	\ee with 
	\[
	\pm\wt\cE_2\le C N^{(-1+\ka)/2}(\cV_N+\cN_++1)(\cN_++1)\,.
	\] We conclude by observing that, thanks to \eqref{eq:sq_gq_bounds}, we can replace the factors $\g_{p+q}$ and $\gamma_p$ in  \eqref{eq:GN3_last_step} by one, producing an error that satisfies the desired bound \eqref{eq:GN3_error_bound}.
\end{proof}

\subsection{Analysis of $\cG_N^{(4)}$}
We now analyse the term
\[\cG_N^{(4)}=e^{-B} \cL_N^{(4)} e^B = \fra{N^\ka}{ N}\sum_{\substack{p,q,r\in\retp\\r\neq -p,-q}}\wh V(p/N^{1-\ka})\, e^{-B} a_{p+r}^*a_q^*a_{p}a_{q+r} e^B \,.\]
\begin{lemma}\label{lem:GN4}
	We have
	\bes{\label{eq:GN4_Result}
		\cG_N^{(4)}=\,&\cV_N+\fra{N^\ka}{2N}\sum_{p,q\in\retp}\wh V((p-q)/N^{1-\ka})\s_q\g_q\s_p\g_p(1+1/N-2\cN_+/N)\\
		&+\fra{N^\ka}{2N}\sum_{p,q\in\retp}\wh V((p-q)/N^{1-\ka})(\h_H)_q\Big[\g_p^2b_p^*b_{-p}^*+2\g_p\s_pb_p^*b_p+\s_p^2b_pb_{-p}\\
		&\qquad+d_p(\g_pb_{-p}+\s_pb_p^*)+(\g_pb_p+\s_pb_{-p}^*)d_{-p}+\hc\Big]\\
		&+\fra{N^\ka}{N^2}\sum_{p,q,u\in\retp}\wh V((p-q)/N^{1-\ka})(\h_H)_q(\h_H)_p\Big[(\g_u^2+\s_u^2)b_u^*b_u+\g_u\s_u b_u^* b_{-u}^*\\
		&\qquad+\g_u\s_ub_ub_{-u}+\s_u^2\Big]+\cE_{\cG_N}^{(4)}\,,
	} with the estimate 
	\be\label{eq:GN4_error_bound}
	\pm\cE_{\cG_N}^{(4)}\le CN^{-1/2+3\ka}(\cV_N+\cN_++1)(\cN_++1) \,.
	\ee
\end{lemma}
\begin{proof}
	We write 
	\[
	a_{p+r}^*a_q^*a_{p}a_{q+r}=b_{p+r}^*b_q^*b_{p}b_{q+r}\left(1-\fra{3}{N}+\fra{2\cN_+}{N}\right)+a_{p+r}^*a_q^*a_{p}a_{q+r}\Theta_{\cN_+}\,,
	\] with 
	\spl{
		\Th_{\cN_+}=\,&\left[\fra{(N-\cN_++2)(\cN_+-1)}{N^2}+\fra{(\cN_+-2)}{N}\right]^2\\
		&-\left[\fra{\cN_+^2-3\cN_+-2}{N^2}\right]\left[\fra{(N-\cN_++2)(N-\cN_++1)}{N^2}\right]
	} satisfying $\pm\Th_{\cN_+}\le C(\cN_++1)^2/N^2$. From this and Lemma \ref{lem:B_bounds_HN}, it follows that
	\bes{\label{eq:GN4_step1}
		\cG_N^{(4)}
		=\,&\fra{N^\ka(N+1)}{N^2}\sum_{\substack{p,q,r\in\retp\\r\neq -p,-q}}\wh V(p/N^{1-\ka})\, e^{-B} b_{p+r}^*b_q^*b_{p}b_{q+r} e^B \\
		&+\fra{N^\ka}{N^2}\sum_{\substack{p,q,u,r\in\retp\\r\neq -p,-q}}\wh V(p/N^{1-\ka})\, e^{-B} b_{p+r}^*b_q^*b_u^*b_ub_{p}b_{q+r} e^B +\wt \cE_1
	} with an error that satisfies 
	\[
	\pm\wt\cE_1\le CN^{-1+\ka}(\cV_N+\cN_++1)(\cN_++1)\,.
	\]
	The first term on the right hand side of \eqref{eq:GN4_step1} can be expanded using \eqref{eq:dp_def}. Bringing the $b$ and $b^*$ operators into normal order and using again Lemma \eqref{lem:action_bogoliubov} to bound all remainder terms, it is possible to show that 
	\spl{
		\fra{N^\ka(N+1)}{N^2}&\sum_{\substack{p,q,r\in\retp\\r\neq -p,-q}}\wh V(p/N^{1-\ka})\, e^{-B} b_{p+r}^*b_q^*b_{p}b_{q+r} e^B \\
		=\,&\cV_N+\fra{N^\ka}{2N}\sum_{p,q\in\retp}\wh V((p-q)/N^{1-\ka})\s_q\g_q\s_p\g_p(1+1/N-2\cN_+/N)\\
		&+\fra{N^\ka}{2N}\sum_{p,q\in\retp}\wh V((p-q)/N^{1-\ka})\s_q\g_q\Big[\g_p^2b_p^*b_{-p}^*+2\g_p\s_pb_p^*b_p+\s_p^2b_pb_{-p}\\
		&\quad+\big(d_p(\g_pb_{-p}+\s_pb_p^*)+(\g_pb_p+\s_pb_{-p}^*)d_{-p}+\hc\big)\Big]+\wt\cE_2
	} with error term bounded by 
	\[
	\pm\wt\cE_2\le CN^{-(1-\ka)/2}(\cV_N+\cN_++1)(\cN_++1)\,.
	\]
	As for the remaining term in \eqref{eq:GN4_step1}, using the fact that $\sum_{u\in\retp}e^{-B} b_u^*b_u e^B = e^{-B} \cN_+(1-\cN_+/N) e^B$, applying Lemma \ref{lem:B_bound_N_easy} and expanding the remaining $b$ and $b^*$ operators, one can see that
	\spl{
		\fra{N^\ka}{N^2}&\sum_{\substack{p,q,r\in\retp\\r\neq -p,-q}}\wh V(r/N^{1-\ka})\, e^{-B} b_{p+r}^*b_q^*b_u^*b_ub_{p}b_{q+r} e^B \\
		=\,&\fra{N^\ka}{N^2}\sum_{p,q,u\in\retp}\wh V((p-q)/N^{1-\ka})\s_q\g_q\s_p\g_p\Big[(\g_u^2+\s_u^2)b_u^*b_u+\g_u\s_u b_u^* b_{-u}^*\\
		&\quad+\g_u\s_ub_ub_{-u}+\s_u^2\Big]+\wt\cE_3. 
	} Here, the error $\wt\cE_3$ satisfies the same bound as $\wt\cE_2$. Bringing everything together, we find 
	\bes{\label{eq:GN4_final_step}
		\cG_N^{(4)}=\,&\cV_N+\fra{N^\ka}{2N}\sum_{p,q\in\retp}\wh V((p-q)/N^{1-\ka})\s_q\g_q\s_p\g_p(1+1/N-2\cN_+/N)\\
		&+\fra{N^\ka}{2N}\sum_{p,q\in\retp}\wh V((p-q)/N^{1-\ka})\s_q\g_q\Big[\g_p^2b_p^*b_{-p}^*+2\g_p\s_pb_p^*b_p+\s_p^2b_pb_{-p}\\
		&\qquad+d_p(\g_pb_{-p}+\s_pb_p^*)+(\g_pb_p+\s_pb_{-p}^*)d_{-p}+\hc)\Big]\\
		&+\fra{N^\ka}{N^2}\sum_{p,q,u\in\retp}\wh V((p-q)/N^{1-\ka})\s_q\g_q\s_p\g_p\Big[(\g_u^2+\s_u^2)b_u^*b_u+\g_u\s_u b_u^* b_{-u}^*\\
		&\qquad+\g_u\s_ub_ub_{-u}+\s_u^2\Big]+\cE_{\cG_N}^{(4)}\,.
	} We conclude the proof by noticing that, with \eqref{eq:sq_gq_bounds}, we can replace $\g_q\s_q$ in the second line and $\s_q\g_q\s_p\g_p$ in the last line of \eqref{eq:GN4_final_step} by resp. $(\h_H)_p$, $(\h_H)_p(\h_H)_q$, producing an error that satisfies the bound \eqref{eq:GN4_error_bound}.
\end{proof}

\subsection{Proof of Proposition \ref{prop:GN}}
We now collect the results from the previous sections to prove Proposition \ref{prop:GN}. By \eqref{eq:GN0}, \eqref{eq:B_kinetic}, \eqref{eq:B_2V}, \eqref{eq:GN3_Result} and \eqref{eq:GN4_Result}, we find 
\begin{equation*}
	\cG_N=\wt C_{\cG_N}+\wt C_{\cG_N}+\cC_N+\cH_N+\cD_N+\wt\cE,
\end{equation*} where the error term satisfies 
\begin{equation*}
	\pm\wt\cE\le N^{-1/2+3\ka}(\cH_N+\cN_+^2+1)(\cN_++1)
\end{equation*} by \eqref{eq:GNO_error}, \eqref{eq:B_kinetic_error},
\eqref{eq:B_2V_error}, \eqref{eq:GN3_error_bound}, \eqref{eq:GN4_error_bound}, and where we defined 
\bes{\label{eq:C2N_first}
	\wt C_{G_N}=\,&\fra{N-1}{2}N^\ka\wh V(0)+\sum_{p\in\retp}p^2\s_p^2\bigg[ 1+\fra{1}{N} +\fra{1}{N}\sum_{u\in\retp}\s_u^2\bigg]\\
	&+\sum_{p\in\retp}N^\ka\wh V(p/N^{1-\ka})\s_p^2+\sum_{p\in\retp}N^\ka\wh V(p/N^{1-\ka})\g_p\s_p\\
	&+\fra{N^\ka}{2N}\sum_{p,q\in\retp}\wh V((p-q)/N^{1-\ka})\s_q\g_q\s_p\g_p(1+1/N)\\
	&+\fra{N^\ka}{N^2}\sum_{p\in P_H, u\in\retp}\convo{\h_H}\h_p\s_u^2,
} 
\begin{equation}\label{eq:QGN_tilda}
\begin{split}
	\wt Q_{\cG_N} =& \sum_{p\in\retp}b_p^*b_p\left[2\s_p^2p^2+N^\ka\wh V(p/N^{1-\ka})(\g_p+\s_p)^2+\fra{2N^\ka}{N}\g_p\s_p\convo{\h_H}_p\right]\\
	&+\sum_{p\in\retp}(b_p^*b_{-p}^*+b_pb_{-p})\\
	&\qquad\times\left[p^2\g_p\s_p+\fra{N^\ka}2\wh V(p/N^{1-\ka})(\g_p+\s_p)^2+\fra{N^\ka}{2N}\convo{\h_H}_p(\g_p^2+\s_p^2)\right]\\
	&-\fra{\cN_+}{N}\sum_{p\in\retp}\left[p^2\s_p^2+N^\ka\wh V(p/N^{1-\ka})\g_p\s_p+\fra{N^\ka}N\sum_{q\in\retp}\wh V((p-q)/N^{1-\ka})\g_p\s_p\g_q\s_q\right]\\
	&+\fra1N\sum_{u\in\retp}\left[(\g_u^2+\s_u^2)b_u^*b_u+\g_u\s_u(b_u^*b_{-u}^*+b_ub_{-u})\right]\\
	&\qquad \times\sum_{p\in P_H}\left[p^2\s_p^2+\fra{N^\ka}N\convo{\h_H}_p\h_p\right],
\end{split}
\end{equation}
and 
\bes{\label{eq:GN_remainder_terms}
	\cD_N=\,&\sum_{p\in\retp}p^2\s_pd_p^*b^*_{-p}+\hc\\
	&+\fra{1}{2}\sum_{p\in\retp}N^\ka\left[\wh V(p/N^{1-\ka})+\fra{1}{N}\sum_{q\in\retp} V((p-q)/N^{1-\ka})(\h_H)_q\right]\\
	&\qquad\times\Big(\left[(\g_pb_p^*+\s_pb_{-p})d_{-p}^*+d_p^*(\g_pb_{-p}^*+\s_pb_p)\right]+\hc\Big)\\
	=\,&\sum_{p\in\retp}\left[p^2\s_p+\fra{1}{2}N^\ka\wh V(p/N^{1-\ka})+\fra{1}{2N}N^\ka\convo{\h_H}\right]d_p^*b_{-p}^*+\hc\\
	&+\sum_{p\in\retp}\fra{1}{2}N^\ka\left[\wh V(p/N^{1-\ka})+\fra{1}{N}\convo{\h_H}\right]\left[(\g_p-1)d_p^*b^*_{-p}+\s_pd_p^*b_p+\hc\right]\\
	&+\sum_{p\in\retp}\fra{1}{2}N^\ka\left[\wh V(p/N^{1-\ka})+\fra{1}{N}\convo{\h_H}\right]\left[\s_pd_pb_p^*+\g_pd_p^*b_{-p}^*+\hc\right].
} 

$\cD_N$ collects all terms arising from the conjugation of $\cL_N$ containing $d, d^*$ operators that are not already included in $\wt\cE$. We want to extract relevant contributions from these terms. First observe that $\h_H$ can be replaced by $\h$ in the convolutions appearing on the r.h.s. of \eqref{eq:GN_remainder_terms}, producing an error smaller than $N^{-1/2+3\ka/2+\a/2}(\cN_++1)^2$ in the operator sense. Taking for example the term on the first line, using \eqref{eq:dp_bounds} and \eqref{eq:eta_punctual_bound} we can bound
\[
\begin{split}
\pm\langle\x,\fra{N^{\ka}}{2N}&\sum_{p\in\retp}\left[\convo{\h}-\convo{\h_H}\right]d_p^*b_{-p}^*\x\rangle\\
\le &CN^{-2+\ka}\sum_{\substack{p\in\retp,\abs q \le N^\a}}|\wh V((p-q)/N^{1-\ka})\h_q|\\
&\qquad \times\left[\h_H(p)\norm{(\cN_++1)\x}+\norm{\h_H}\norm{b_p(\cN_++1)^{1/2}\x}\right]\norm{\cN_+\x}\\
 \le& CN^{-2+\ka}\norm{\wh V(\cdot/N^{1-\ka})}_2\norm{\h_H}_2\norm{(\cN_++1)\x}^2\sum_{|q|\le N^\a}|\h_q|\\
\le & CN^{-\fra{1}{2}+\fra{3\ka+\a}2}\norm{(\cN_++1)\x}^2.
\end{split}
\]
The other errors can be bounded in the same way. 
We now denote by $D_i$ the term on the $i$-th line of the right hand side of  \eqref{eq:GN_remainder_terms} after this replacement. By \eqref{eq:scattering_eta}, \eqref{eq:dp_bounds} and \eqref{eq:chi_f_bound} we conclude that $\pm D_1, \pm D_2\le CN^{-1+\ka}(\cN_++1)^2$. As for $D_3$, switching to position space, one can bound the term proportional to $\g_pd_p^*b_{-p}^*$ by $CN^{-1+\ka/2}(\cV_N+\cN_++1)(\cN_++1)$ and easily replace $\s_p$ by $\h_H (p)$ in the other term. This produces an additional error of order $N^{-1+\ka}(\cN_++1)^2$. We are left with 
\[
\wt D_3=\fra12\sum_{p\in P_H}N^\ka\convo{\wh f_N}\h_pd_pb_p^*+\hc.
\] Here we don't have enough decay in $p$ to conclude by \eqref{eq:dp_bounds}: this term contains contributions that are not negligible in the limit of large $N$, at our desired level of precision. To isolate them, we write, for $p\in P_H$,
\spl{d_p=& (\g_p-1)b_p-\s_pb_{-p}^*+ e^{-B} b_p e^B -b_p\\
	=& - \int_0^1 \left(\s_p^{(s)}b_p+\g_p^{(s)}b_{-p}^*\right)ds- \int_0^1 e^{-sB} [B,b_p] e^{sB} ds\\
	=&\h_p\int_0^1d_{-p}^{(s)*}ds-\fra{\h_p}{N}\int_0^1e^{-sB}\cN_+ b_{-p}^* e^{sB}ds-\fra{1}{N}\int_0^1\sum_{q\in P_H} \h_q \, e^{-sB} b_q^*a_{-q}^*a_p e^{sB}ds,}
 where for $s\in[0,1]$ we used the notation $\s_p^{(s)}, \g_p^{(s)}, d_p^{(s)}$ to denote coefficients and operators built from $s\h$ instead of $\h$. The additional $\h_p$ factor in the first two terms lets us bound them as above. We can thus write 
 \[\wt D_3=-\fra{N^\ka}{2N}\int_0^1\sum_{p,q\in P_H}\convo{\wh f_N}\h_p\h_q\left[e^{-sB}b_q^*a_{-q}^*a_pe^{sB}b_p^*+\hc\right]ds+\wt\cE_1,\]
 with
\[\pm\wt\cE_1\le CN^{-1+\ka}(\cN_++1)^2\,.\]
We now further expand
\[
e^{-sB}a_{-q}^*a_pe^{sB} = a_{-q}^*a_{p}+\int_0^s e^{-tB} (\h_p b_{-q}^*b_{-p}^* + \h_q b_{p}b_{q})e^{tB}dt 
\] and plug this into $\wt\cD_3$. The term with the additional factor $\h_p$ can also be bounded as above. In the term proportional to $e^{-sB}b_q^*e^{sB}a_{-q}^*a_pb_p^*$ we commute $b_p^*$ to the left, while in the last term that is left we expand $e^{-tB}b_qe^{tB}=\g_q^{(t)}b_q+\s_q^{(t)}b_{-q}^*+d_q^{(t)}$ and we commute the $b_p^*$ to the left of $\g_q^{(t)}b_q$. One can see that the quartic terms (in creation and annihilation operators) can now be bounded as above by Cauchy-Schwarz. We are left with the quadratic terms arising from the commutators: 
\[ \begin{split}
\wt D_3=& -\fra{N^{\ka}}{2N}\int_0^1\sum_{p,q\in P_H}\convo{\wh f_N}\h_p\h_qe^{-sB}b_q^*e^{sB}b_{-q}^*ds\\
&-\fra{N^{\ka}}{2N}\int_0^1\int_0^s\sum_{p,q\in P_H}\convo{\wh f_N}\h_p\h_q^2e^{-sB}b_q^*e^{sB}e^{-tB}b_{q}e^{tB}dtds \\
&+\hc +\wt\cE_2
\end{split}
\]  with
\[\pm\wt\cE_2\le CN^{-1+\ka}(\cN_++1)^2\,.\] Expanding once more with \eqref{eq:dp_def} and collecting the terms with $d_q^{(s)*}, d_q^{(t)}$ in the error, and finally integrating in $t$ and $s$ we arrive at
\bes{\label{eq:D_relevant_contributions}
	\cD_N=\,&-\fra{N^\ka}{2N}\sum_{q\in\retp, p\in P_H}\convo{\wh f_N}_p\h_p\left[\g_q\s_q(b_q^*b_{-q}^*+b_qb_{-q})+(\s_q^2+\g_q^2)b_q^*b_q+\s_q^2\right]\\
	&+\fra{N^\ka}{2N}\sum_{q\in\retp, p\in P_H}\convo{\wh f_N}_p\h_pb_q^*b_q+\wt\cE_3.
} We denote the constant term on the r.h.s. of \eqref{eq:D_relevant_contributions} by $\cD_N^{(0)}$ and collect the quadratic ones in $\cD_N^{(2)}$, so that $\cD_N=\cD_N^{(0)}+\cD_N^{(2)}+\wt\cE_3$ with the error bound 
\[
\pm\wt\cE_3\le CN^{-\fra{1}{2} +\frac\ka2}(\cV_N+\cN_++1)(\cN_++1).
\]

Let us now go back to the term $\wt{C}_{\cG_N}$ defined in \eqref{eq:C2N_first}. Using \eqref{eq:eta_punctual_bound} and \eqref{eq:sq_gq_bounds} we see that we can replace the $\h_H$ in the convolution on the last line by $\h$, producing an error smaller than $N^{-1}$. Adding the constant contribution coming from \eqref{eq:D_relevant_contributions} and rearranging, we get 
\spl{
	\wt C_{G_N}+\cD_N^{(0)}=\,&\fra{N-1}{2}N^\ka\wh V(0)+\sum_{p\in\retp}\left[p^2\s_p^2+N^\ka\wh V(p/N^{1-\ka})(\s_p\g_p+\s_p^2)\right]\\
	&+\fra{1}{2N}\sum_{p,q\in\retp}N^\ka\wh V((p-q)/N^{1-\ka})\s_p\g_p\s_q\g_q\\
	&+\fra{1}{N}\sum_{p\in P_H}\h_p\left[p^2\h_p+\fra{1}{2N}N^\ka\convo{\h}_p\right]\\
	&+\fra{1}{N}\sum_{u\in\retp}\s_u^2\sum_{p\in P_H}\h_p\left[p^2\h_p-\fra{1}{2}N^\ka\wh V(p/N^{1-\ka})+\fra{1}{2N}N^\ka\convo{\h}_p\right]\\
	&+\cO(N^{-1+2\ka+\a})\,,
} where the error $N^{-1+2\ka+\a}$ arises from substituting the factors $\s_q$, $\s_q\g_q\s_p\g_p$ with $(\h_H)_q$, $(\h_H)_q(\h_H)_p$, and then $\h_H$ with $\h$ in the resulting convolution on the third line. This is, by the assumptions on $\a$, smaller than $N^{-1/2+2\ka}$. Using now equation \eqref{eq:scattering_eta} and the bound \eqref{eq:chi_f_bound} we find
\spl{
	\wt C_{G_N}+\cD_N^{(0)}=\,&\fra{N-1}{2}N^\ka\wh V(0)+\sum_{p\in\retp}\left[p^2\s_p^2+N^\ka\wh V(p/N^{1-\ka})(\s_p\g_p+\s_p^2)\right]\\
	&+\fra{1}{2N}\sum_{p,q\in\retp}N^\ka\wh V((p-q)/N^{1-\ka})\s_p\g_p\s_q\g_q\\
	&+\fra{1}{N}\sum_{p\in P_H}\h_p\left[p^2\h_p+\fra{1}{2N}N^\ka\convo{\h}_p\right]\\
	&-\fra{1}{N}\sum_{u\in\retp}\s_u^2\sum_{p\in P_H}N^\ka\wh V(p/N^{1-\ka})\h_p+\cO(N^{-1/2+2\ka})\,,
} 
and finally in the last line we replace the sum over $P_H$ with one over the whole $\L_+^*$ to get \eqref{eq:CGN_def}. This produces a negligible error, of order at most $N^{-1+2\ka+\a/2}\le N^{-1/2+2\ka}$. Indeed, using $\wh V\in L^\infty(\ret)$ and \eqref{eq:eta_punctual_bound} we see that 
\[
\pm\fra{1}{N}\sum_{u\in\retp}\s_u^2\sum_{p\in P_H^c}N^\ka\wh V(p/N^{1-\ka})\h_p\le CN^{-1+2\ka+\a/2}.
\]

Similarly, we combine $\wt\cQ_{\cG_N}$ defined in \eqref{eq:QGN_tilda} with the quadratic terms in \eqref{eq:D_relevant_contributions}. Using again the scattering equation \eqref{eq:scattering_eta} and the bound \eqref{eq:chi_f_bound}, we get \[\wt\cQ_{\cG_N}+\cD_N^{(2)}=\cQ_{\cG_N}+\wt\cE_4\] with the bound $\pm\wt\cE_4\le N^{-1/2+3\ka/2+\a}(\cK+\cN+1)(\cN+1)$. We omit the details, but we remark that also in these terms, replacing $\h_H$ with $\h$ produces terms of the leading order in the error. This concludes the proof of the proposition. 
%



\section{Analysis of $\cJ_n$}
\label{sec:cubic}

In this section, we sketch the proof of Proposition \ref{prop:JN}. The proof is similar to the proof of \cite[Prop. 3.3]{BBCS4}, taking into account the different scaling.
We recall the antisymmetric operator $A$, defined in \eqref{eq:A_def}. Throughout this section, we assume that the parameters $\kappa, \alpha, \beta$ satisfy \eqref{eq:conditions_parameters_JN}.

\subsection{Analysis of $\conA{\cQ_{\cG_N}}$}

To control the action of $A$ on the quadratic operator $\cQ_{\cG_N}$ introduced in (\ref{eq:QGN_def}), we will make use of the following lemma.  
\begin{lemma}\label{lem:G_PHI_A_commutator_bound}
Let $\G_p,\Phi_p$ be sequences satisfying \[
\abs{\G_p}\le CN^\ka\left(\chi_{\{\abs p \le N^\a\}}+1/\abs{p}^2\right),\qquad \abs{\Phi_p}\le CN^\ka.
\] Then we have 
\begin{gather}\label{eq:comm_Phi_A_bound}
	\pm\sum_{p\in\retp}\Phi_p[b_pb_p^*,A]\le CN^{-\fra{1}{2}}(\cN_++1)^2\\
	\label{eq:comm_Gamma_A_bound}\pm\sum_{p\in\retp}\G_p[(b_p^*b_{-p}^*+b_pb_{-p}),A]\le CN^{-\fra{1}{2}}(\cK+1)(\cN_++1).
	\end{gather}
\end{lemma}

Thanks to the scattering equation \eqref{eq:scattering_f} and to the bounds \eqref{eq:chi_f_bound}, the coefficients $\G_p$, $\Phi_p$ appearing in \eqref{eq:QGN_def} satisfy the assumptions of Lemma \ref{lem:G_PHI_A_commutator_bound}. 
\begin{lemma}\label{lem:QGN_A_conj}
We have	
	\[
	\conA{\cQ_{\cG_N}}=\cQ_{\cG_N}+\cE_{\cQ_N}\,,
	\] with 
	\[
	\pm\cE_{\cQ_N}\le CN^{-\fra{1}{2}+\ka}(\cH+\cN_+^2+1)(\cN_++1)\,.
	\]
\end{lemma}
\begin{proof}[Proof of Lemma \ref{lem:QGN_A_conj}]
	We expand 
	\[
	\conA{\cQ_{\cG_N}}=\cQ_{\cG_N}+\int_0^1\consA{[\cQ_{\cG_N},A]}ds
	\] and the claim follows from Lemmas \ref{lem:G_PHI_A_commutator_bound} and \ref{lm:growNA}.
\end{proof}

\begin{proof}[Proof of Lemma \ref{lem:G_PHI_A_commutator_bound}]
	The proof is similar to that of \cite[Lemma 8.2]{BBCS4}. Using the commutation rules (\ref{eq:bpCCR}), we find 
	\bes{\label{eq:comm_b_A}
		[b_p^*,A]=\,& N^{-\fra{1}{2}}\sum_{\substack{r\in P_H\\ v\in\ P_L}}\h_r\Bigg[ b_v^*b_{-r}\left(1-\fra{\cN_+}{N}\right)\d_{p,r+v}+b_v^*b_{r+v}\left(1-\fra{\cN_+-1}{N}\right)\d_{p,-r}\\
		&\qquad\qquad\qquad-b_{r+v}^*b_{-r}^*\left(1-\fra{\cN_+}{N}\right)	\d_{p,v}\Bigg]\\
		&-N^{\fra{3}{2}}\sum_{\substack{r\in P_H\\ v\in\ P_L}}\h_r\Big[b_v^*(b_{-r}a_p^*a_{r+v}+a_p^*a_{-r}b_{r+v})-b_{r+v}^*b_{-r}^*a_p^*a_v\Big]\, .
	} Since $A^*=-A$, we have $[b_p^*,A]^*=[A^*,b_p]=-[A,b_p]=[b_p,A]$ and we can also compute 
	\spl{
		[b_p^*b_p,A]=\sum_{j=1}^6\D_j+\hc\,,
	} with 
	\spl{
		\D_1=\,&N^{-1/2}\sum_{r\in P_H,v\in P_L}\Phi_{r+v}\h_r\left(1-\fra{\cN_+-1}{N}\right)b_{r+v}^*b_{-r}^*b_v\\
		\D_2=\,&N^{-1/2}\sum_{r\in P_H,v\in P_L}\Phi_{r}\h_r\left(1-\fra{\cN_+-2}{N}\right)b_{r+v}^*b_{-r}^*b_v\\
		\D_3=\,&N^{-1/2}\sum_{r\in P_H,v\in P_L}\Phi_{v}\h_r\left(1-\fra{\cN_+-2}{N}\right)b_{r+v}^*b_{-r}^*b_v\\
		\D_4=\,&N^{-3/2}\sum_{p\in\retp}\sum_{r\in P_H,v\in P_L}\Phi_{p}\h_rb_p^*a_{r+v}^*a_pb_{-r}^*b_v\\
		\D_5=\,&N^{-3/2}\sum_{p\in\retp}\sum_{r\in P_H,v\in P_L}\Phi_{p}\h_rb_p^*b_{r+v}^*a_{-r}^*a_pb_v\\
		\D_6=\,&N^{-3/2}\sum_{p\in\retp}\sum_{r\in P_H,v\in P_L}\Phi_{p}\h_rb_p^*a_{v}^*a_pb_{-r}b_{r+v}\,.
	}
We point out here that we get less terms than the corresponding ones in \cite[Lemma 8.2]{BBCS4}. The infrared cutoff that we introduced in sequence $\h$ leads to a slightly different form of the operator $A$ than the one used in \cite{BBCS4}, which in turn causes several terms appearing there to vanish in our setting. This remark also applies to other similar expansions in the rest of this section.
We now proceed to bound the individual terms $\D_i$. With $\abs{\Phi_p}\le CN^\ka$, $\norm{\h_H}_2\le CN^{\ka-\a/2}$ and $\alpha > 4\ka$, Cauchy-Schwarz implies that
\[\pm\D_{1,2,3}\le CN^{-1/2+2\ka-\a/2}(\cN_++1)^2\le C N^{-1/2}(\cN_++1)^2\,,\]
Rearranging $\D_4$ in normal order, we obtain a commutator term, cubic in creation and annihilation operators, that can be estimated like $\D_{1}, \D_2, \D_3$. The remaining terms are easily seen to satisfy 
\[ \pm\D_{4}, \pm \D_5, \pm \D_6\le CN^{-1+2\ka-\a/2}(\cN_++1)^2\,,\] concluding the proof of \eqref{eq:comm_Phi_A_bound}. 

As for the off-diagonal terms, we have that 	
	\[
	\sum_{p\in\retp}\G_p[(b_p^*b_{-p}^*+b_pb_{-p}),A]=\sum_{j=1}^7\Y_j+\hc\,,
	\] where 
	
	\spl{
		\Y_1=\,&N^{-1/2}\sum_{r\in P_H,v\in P_L}\G_{r+v}\h_r\left(1-\fra{\cN_++1}{N}\right)(b_{-r}^*b_{-r-v}-\fra{1}{N}a_{-r}^*a_{-r-v})b_v\\
		\Y_2=\,&N^{-1/2}\sum_{r\in P_H,v\in P_L}\G_{r+v}\h_r\left(1-\fra{\cN_+}{N}\right)b_{-r}^*b_vb_{-r-v}\\
		\Y_3=\,&N^{-1/2}\sum_{r\in P_H,v\in P_L}\G_{r}\h_r\left(1-\fra{\cN_++1}{N}\right)(b_{r+v}^*b_{r}-\fra{1}{N}a_{r+v}^*a_{r})b_v\\
		\Y_4=\,&N^{-1/2}\sum_{r\in P_H,v\in P_L}\G_{r}\h_r\left(1-\fra{\cN_+-1}{N}\right)b_{r+v}^*b_vb_{r}\\
		\Y_5=\,&-N^{-1/2}\sum_{r\in P_H,v\in P_L}\G_{v}\h_r\left[\left(1-\fra{\cN_+}{N}\right)+\left(1-\fra{\cN_++1}{N}\right)\right]b_{r+v}b_{-r}b_{-v}\\
		\Y_6=\,&-N^{-3/2}\sum_{p\in\retp}\sum_{r\in P_H,v\in P_L}\G_{p}\h_r\Big[b_p(a_{r+v}^*a_{-p}b_{-r}^*+b_{r+v}^*a_{-r}^*a_{-p})b_v-b_pa_v^*a_{-p}b_{-r}b_{r+v}\Big]\\
		\Y_7=\,&-N^{-3/2}\sum_{p\in\retp}\sum_{r\in P_H,v\in P_L}\G_{p}\h_r\Big[(a_{r+v}^*a_{-p}b_{-r}^*+b_{r+v}^*a_{-r}^*a_{-p})b_vb_{-p}-a_v^*a_{-p}b_{-r}b_{r+v}b_{-p}\Big]\,.
	} We start by showing how to bound $\Y_2$, the two terms in $\Y_1$ are bounded in the same way. Using Cauchy-Schwarz we get \[
\begin{split}
|\expec{\x}{\Y_2}|\le& N^{-\fra12+\ka-2\a}\left(\sum_{r\in P_H, v\in P_L}|\G_{r+v}|^2\|b_{-r}\x\|^2\right)^{1/2}\left(\sum_{r\in P_H, v\in P_L}\|b_{-r-v}b_v\x\|^2\right)^{1/2}\\
\le & CN^{-\fra12+2\ka+\fra{3}{2}\b-2\a}\|\cN_+^{1/2}\x\|\|(\cN_++1)\x\|,
\end{split}
\] where we used the fact that $|\h_r|\le CN^{\ka-2\a}$ and $\sum_{v\in P_L}|\G_{r+v}|^2\le CN^\ka\abs{P_L}^{1/2}\le CN^{\ka+3\b/2}$ uniformly in $r\in P_H$. $\Y_3$ and $\Y_4$ are bounded similarly, and thanks to the assumptions on $\a,\b,\k$ we conclude
\[\pm\Y_{1,2,3,4}\le N^{-\fra12}(\cN_++1)^2.\]
As for $\Y_5$, we use again Cauchy-Schwarz to bound 
\[\begin{split}
|\expec{\x}{\Y_5}|\le& N^{-\fra12}\left(\sum_{r\in P_H, v\in P_L}|\G_{v}|^2|\h_r|^2|r|^{-2}\|\cN_+^{1/2}\x\|^2\right)^{1/2}\left(\sum_{r\in P_H, v\in P_L}|r|^2\|b_{-r}b_{-v}\x\|^2\right)^{1/2}\\
\le & CN^{-\fra12+2\ka+\fra{3}{2}\b-\fra32\a}\|\cN_+^{1/2}\x\|\|(\cK+1)^{1/2}(\cN_++1)^{1/2}\x\|,
\end{split}
\] which together with the assumptions on $\a,\b,\k$ shows $\pm\Y_5\le N^{-1/2}(\cK+1)(\cN_++1)$. We now turn our attention to $\Y_{6}, \Y_7$. Here we commute each of the terms in square bracket to (partial) normal order and again bound it using Cauchy-Schwarz. We show how this is done for the first term in $\Y_6$. We have
	
	\spl{
		N^{-3/2}&\sum_{p\in\retp}\sum_{r\in P_H,v\in P_L}\G_{p}\h_rb_pa_{r+v}^*a_{-p}b_{-r}^*b_v\\
		=\,&N^{-3/2}\sum_{p\in\retp}\sum_{r\in P_H,v\in P_L}\G_{p}\h_rb_pb_{-r}^*a_{r+v}^*a_{-p}b_v 
		+N^{-3/2}\sum_{r\in P_H,v\in P_L}\G_{r}\h_rb_rb_{r+v}^*b_v\,,
	} and using $\|\G\|_2\le CN^{\ka+3\a/2}$ and $\|\chi_{P_H}\h\|_2\le CN^{\ka-\a/2}$ we bound 
	
	\spl{
		\pm&\expec{\x}{N^{-3/2}\sum_{p\in\retp}\sum_{r\in P_H,v\in P_L}\G_{p}\h_rb_pb_{-r}^*a_{r+v}^*a_{-p}b_v}\\
		\le\,&N^{-3/2}\left(\sum_{p\in\retp}\sum_{r\in P_H,v\in P_L}\abs{\G_p}^2\norm{a_{r+v}b_{-r}b_p^*\x}^2\right)^{1/2}\left(\sum_{p\in\retp}\sum_{r\in P_H,v\in P_L}\abs{\h_r}^2\norm{a_{-p}b_{v}\x}^2\right)^{1/2}\\
		\le & CN^{-1+2\ka+\a}\norm{(\cN_++1)\x}^2,
	} and analogously for the cubic term. The other terms in $\Y_6$ and $\Y_7$ are bounded in the same way, and using the assumptions \eqref{eq:conditions_parameters_JN}we get 
	\[
	\pm\Y_{6}, \pm \Y_7\le CN^{-1/2}(\cN_++1)^2\,,
	\] which concludes the proof of \eqref{eq:comm_Gamma_A_bound} and the lemma.
\end{proof}

\subsection{Analysis of $\conA{\cC_N}$}

To control the action of $A$ on the cubic term $\cC_N$, defined in \eqref{eq:CubicGN_def}, we will make use of the following lemma. 
\begin{lemma}\label{lem:comm_CN_A}
We have 
\be\label{eq:comm_CN_A}
[\cC_N,A]= \X_0 + \delta_{\cC_N} 
\ee
where
\begin{equation}\label{eq:Xi0}  \X_0 =2  N^{-1}\sum_{r\in P_H,v\in P_L}N^\ka \left(\wh V(r/N^{1-\ka})+\wh V((r+v)/N^{1-\ka})\right)\h_rb_v^*b_v \end{equation} 
and 
\be\label{eq:comm_X_j_A}
	\pm \delta_{\cC_N}  \le CN^{-1/2}\left[(\cN_++1)(\cK+1)+(\cN_++1)^3\right]\,,
	\ee 
Moreover, 
\be\label{eq:comm_X_0_A}
	\pm[\X_0,A]\le CN^{-1/2}(\cN_++1)^2\,.
	\ee
\end{lemma}

\begin{proof} 
Proceeding as in \cite[Lemma 8.4]{BBCS4}, we find 
\begin{equation*}
[\cC_N,A] = \Xi_0 + \sum_{j=1}^{8}\Big[\X_j+\hc\Big], 
\end{equation*} 
with $\X_0$ as in (\ref{eq:Xi0}) and
\begin{align*}
	\X_1=\,&N^{-1}\sum_{r\in P_H,v\in P_L} N^\ka\wh V((r+v)/N^{1-\ka})\h_rb_v^*\left[\left(1-\fra{\cN_+}{N}\right)^2-1\right]b_v\\
	&+N^{-1}\sum_{r\in P_H,v\in P_L}N^\ka\wh V(r/N^{1-\ka})\h_rb_v^*\left[\left(1-\fra{\cN_++1}{N}\right)\left(1-\fra{\cN_+}{N}\right)-1\right]b_v\\
	\X_2=\,&N^{-1}\sum_{\substack{r\in P_H\\v\in P_L}}\sum_{\substack{p\in\retp\\p\neq r+v}} N^\ka\wh V(p/N^{1-\ka})\h_rb_v^*\left(1-\fra{\cN_++1}{N}\right)(b_{-p}^*b_{-r}-\fra{1}{N}a_{-p}^*a_{-r})\\
	&\qquad\qquad\qquad\qquad\qquad\qquad\qquad\qquad\times(\g_{r+v-p}b_{r+v-p}+\s_{r+v-p}b_{p-r-v}^* )\\
	\X_3=\,&N^{-1}\sum_{\substack{r\in P_H\\v\in P_L}}\sum_{\substack{p\in\retp\\p\neq -r}} N^\ka\wh V(p/N^{1-\ka})\h_rb_v^*\left(1-\fra{\cN_+}{N}\right)(b_{-p}^*b_{r+v}-\fra{1}{N}a_{-p}^*a_{r+v})\\
	&\qquad\qquad\qquad\qquad\qquad\qquad\qquad\qquad\times(\g_{r+p}b_{r+p}+\s_{r+p}b_{-p-r}^* )\\
	\X_4=\,&-N^{-1}\sum_{\substack{r\in P_H\\v\in P_L}}\sum_{\substack{p\in\retp\\p\neq v}} N^\ka\wh V(p/N^{1-\ka})\h_r\left(1-\fra{\cN_+-2}{N}\right) b_{r+v}^*b_{-r}^*b_{-p}^*\\
	&\qquad\qquad\qquad\qquad\qquad\qquad\qquad\qquad\times(\g_{p-v}b_{v-p}+\s_{p-v}b_{p-v}^* )\\
	\X_5=\,&N^{-2}\sum_{\substack{r\in P_H\\v\in P_L}}\sum_{\substack{p,q\in\retp\\p\neq -q}} N^\ka\wh V(p/N^{1-\ka})\h_rb_v^*(a_{p+q}^*a_{r+v}b_{-r}+b_{r+v}a_{p+q}^*a_{-r})b_{-p}^*(\g_qb_q+\s_qb_{-q}^*) \\
	\X_6=\,&N^{-2}\sum_{\substack{r\in P_H\\v\in P_L}}\sum_{\substack{p,q\in\retp\\p\neq -q}} N^\ka\wh V(p/N^{1-\ka})\h_rb_{r+v}^*b_{-r}^*a_{p+q}^*a_vb_{-p}^*(\g_qb_q+\s_qb_{-q}^*)\, \\
	\X_7=\,&N^{-1/2}\sum_{\substack{p,q\in\retp\\p\neq -q}} N^\ka\wh V(p/N^{1-\ka})b_{p+q}^*[b_{-p}^*,A](\g_qb_q+\s_qb_{-q}^*),\\
	\X_8=\,&N^{-1/2}\sum_{\substack{p,q\in\retp\\p\neq -q}} N^\ka\wh V(p/N^{1-\ka})b_{p+q}^*b_{-p}^*[(\g_qb_q+\s_qb_{-q}^*),A]\,.
\end{align*} 
To get \eqref{eq:comm_CN_A} we set $\d_{\cC_N}= \sum_{j=1}^8[\X_j+\hc]$ and we proceed to bound the individual terms $\X_i$.
With $\sum_r|\wh V(r/N^{1-\ka})\h_r|\le CN$ we find $\pm\X_1\le CN^{-1}(\cN_++1)^2$. As for $\X_2$, we use $\norm{\h_H}_2,\norm{\s}_2\le CN^{\ka-\a/2}$ and $\abs{P_L}\le CN^{3\b}$ to conclude that 
\[\pm\X_2\le N^{-1+2\ka+3\b/2-\a/2}(\cN_++1)^2\le CN^{-1/2}(\cN_++1)^2\,.\]
Note that here we need the assumption $\a-\b>2\b+4\ka-1$ and the bound
\[\sum_{\substack{r\in P_H\\v\in P_L}}\sum_{\substack{p\in\retp\\p\neq r+v}}\abs{\s_{r+v-p}}^2\norm{b_{-r}b_{p-r-v}^*\x}\le C \abs{P_L}\norm{\s}_2\norm{(\cN_++1)^2\x}\,, \]
The terms $\X_{3}, \X_4$ are bounded similarly to $\X_2$. As for $\X_5$, we keep the factor $(\g_qb_q+\s_q b_{-q}^*)$ as it is, rearranging the other operators in normal order. Then, we apply Cauchy-Schwarz's inequality and the same bounds as above. We proceed analogously for $\X_6$ (here we only have to commute $a_v$ and $b_{-p}^*$) and we find 
\[\pm\X_{5}, \pm \X_6\le CN^{-1/2}(\cN_++1)^3\,.\]
Expanding $\X_7$ with \eqref{eq:comm_b_A} we find terms analogous to $\X_{1,\dots,6}$ that can be bounded as above. Expanding $\X_8$, however, we find some terms that need to be estimated with the kinetic energy $\cK$. To explain this point, let us write $\X_8=\X_8^{(a)}+\X_8^{(b)}$ respectively for the terms proportional to $\g_q, \s_q$ in the commutator with $A$. We compute $\X_8^{(a)}=\sum_{j=1}^{5}\X_8^{(a,j)}$, with
\[
\begin{split}
\X_8^{(a,1)}=&N^{-1+\ka}\sum_{\substack{r\in P_H\\v\in P_L}}\sum_{\substack{p\in\retp\\p\neq-r-v}}\wh V(p/N^{1-\ka})\h_r\g_{r+v}b_{p+r+v}^*b_{-p}^*\left(1-\fra{\cN_+}N\right)b_{-r}^*b_v\\
\X_8^{(a,2)}=&N^{-1+\ka}\sum_{\substack{r\in P_H\\v\in P_L}}\sum_{\substack{p\in\retp\\p\neq r}}\wh V(p/N^{1-\ka})\h_r\g_{r}b_{p-r}^*b_{-p}^*\left(1-\fra{\cN_+-1}N\right)b_{r+v}^*b_v\\
\X_8^{(a,3)}=&-N^{-1+\ka}\sum_{\substack{r\in P_H\\v\in P_L}}\sum_{\substack{p\in\retp\\p\neq-v}}\wh V(p/N^{1-\ka})\h_rb_{p+v}^*b_{-p}^*\left(1-\fra{\cN_+}N\right)b_{-r}b_{r+v}\\
\X_8^{(a,4)}=&-N^{-2+\ka}\sum_{\substack{r\in P_H\\v\in P_L}}\sum_{\substack{p,q\in\retp\\p\neq-q}}\wh V(p/N^{1-\ka})\h_r\g_{q}b_{p+q}^*b_{-p}^*\left(a_{r+v}^*a_qb_{-r}^*+b_{r+v}^*a_{-r}^*a_q\right)b_v\\
\X_8^{(a,5)}=&-N^{-2+\ka}\sum_{\substack{r\in P_H\\v\in P_L}}\sum_{\substack{p,q\in\retp\\p\neq-q}}\wh V(p/N^{1-\ka})\h_r\g_{q}b_{p+q}^*b_{-p}^*a_{v}^*a_qb_{-r}b_{r+v}.
\end{split}
\] By Cauchy-Schwarz we find 
\[
\begin{split}
\expec{\x}{\X_8^{(a,1)}}\le& N^{1-\ka}\Bigg(\sum_{\substack{r\in P_H\\v\in P_L}}\sum_{\substack{p\in\retp\\p\neq-r-v}}|p|^2\|b_{-r}b_{-p}b_{p+r+v}(\cN_++1)^{-1/2}\x\|^2\Bigg)^{1/2}\\
& \times \Bigg(\sum_{\substack{r\in P_H\\v\in P_L}}\sum_{\substack{p\in\retp\\p\neq-r-v}}\fra{\wh V(p/N^{1-\ka})}{|p|^2}\h_r^2\|b_{v}(\cN_++1)^{1/2}\x\|^2\Bigg)^{1/2}\\
\le & CN^{-1/2}\left[\|(\cK+1)^{1/2}(\cN_++1)^{1/2}\x\|^2+\|(\cN_++1)^{1/2}\x\|^2\right].
\end{split}
\] The other quartic terms in $\X_8^{(a)}$ are bounded in the same way, while the sestic terms are handled like $\X_5$. The terms arising from $\X_8^{(b)}$ are analogous and can be bounded similarly (we omit the details), and in the end we obtain 
\[ \pm\X_{8}\le C N^{-1/2}(\cK+\cN_+^2+1)(\cN_++1)\,,\] 
which concludes the proof of \eqref{eq:comm_X_j_A}. Finally, the bound \eqref{eq:comm_X_0_A} is a consequence of Lemma \ref{lem:G_PHI_A_commutator_bound} and the fact that for $v\in P_L$ we have the uniform bound
\[
\sum_{r\in P_H}N^{-1+\ka}(\wh V(r/N^{1-\ka})+\wh V((r+v)/N^{1-\ka}))\h_r\le CN^{\ka}.
\]
	
\end{proof}

\begin{lemma}\label{lem:conA_CN}
We have
\[
\conA{\cC_N}=\cC_N+2N^{-1}\sum_{\substack{r\in P_H\\v\in P_L}}N^\ka\left(\wh V(r/N^{1-\ka})+\wh V((v+r)/N^{1-\ka})\right)\h_rb_v^*b_v+ \cE_{\cC_N},\] 
with 
\[
\pm\cE_{\cC_N}\le CN^{-1/2+\ka}(\cH_N+\cN_+^2+1)(\cN_++1)\,.
\]
\end{lemma}
\begin{proof}
With Lemma \ref{lem:comm_CN_A}, we find
\spl{
\conA{\cC_N}=\,&\cC_N+\int_0^1\consA{[\cC_N,A]}ds\\
	=\,&\cC_N+\int_0^1\consA{\X_0}ds+\int_0^1\consA{\delta_{\cC_N}} ds\\
	=\,&\cC_N+\X_0+\int_0^1\int_0^{s_1}e^{-s_2A}[\X_0,A]e^{s_2A}ds_2ds_1+\int_0^1\consA{\delta_{\cC_N}}ds,
	} 
and the lemma follows from (\ref{eq:comm_X_j_A}) and (\ref{eq:comm_X_0_A}) and from Lemma \ref{lm:growNA}.
\end{proof}

\subsection{Analysis of $\conA{\cH_N}$}

Finally, we have to control the action of $A$ on the Hamilton operator $\cH_N = \cK + \cV_N$, introduced in (\ref{eq:KVN}). To this end, we will make use of the following two lemmas. 
\begin{lemma}\label{lem:comm_HN_A}
We have 
\be\label{eq:Th_j_worse_bound}\begin{split}
	\abs{\langle\x_1, [ \cH_N , A] \x_2\rangle}&\le C \expec{\x_1}{\cH_N+(\cN_++1)^2}+ C\expec{\x_2}{\cH_N+(\cN_++1)^2} \\	&\hspace{0.5cm} + CN^{\kappa} \expec{\x_1}{ (\cN_++1)}+ CN^{\kappa}\expec{\x_2}{ (\cN_++1)}
	\end{split}\ee 
Moreover, we can decompose 
\be  \label{eq:comm_A_HN} [ \cH_N , A] = \Th_0  + \delta_{\cH_N} \ee
where 
\begin{equation}\label{eq:Th0} 
\Th_0 = -N^{-1/2}\sum_{r\in P_H,v\in P_L}N^\ka\wh V(r/N^{1-\ka})b_{r+v}^*b_{-r}^*b_v +\hc
\end{equation}
and 
\begin{equation}\label{eq:deltaHN} 
\begin{split} 
\pm \delta_{\cH_N} \le & CN^{-\fra12+\ka+\fra\b2}(\cK+\cN_+^2+1)(\cN_++1)\,. \end{split}
\end{equation}
\end{lemma}
 
\begin{proof} 
Proceeding as in \cite[Lemma 4.3]{BBCS4}, we find
\[ [\cH_N,A]=\sum_{j=0}^7\Th_j+\hc\,, \]
with $\Th_0$ as defined in (\ref{eq:Th0}) and 
\spl{
		\Th_1=\,&2N^{-1/2}\sum_{r\in P_H,v\in P_L}N^\ka\h_r(r\cdot v )b_{r+v}^*b_{-r}^*b_v,\\
		\Th_2=\,&N^{-3/2}\sum_{\substack{r\in P_H\\v\in P_L}}\sum_{\substack{q\in\retp,u\in\ret\\u\neq-q,-r-v}}N^\ka\wh V(u/N^{1-\ka})\h_rb_{r+v+u}^*b_{-r}^*a_q^*a_{q+u}b_v,\\
		\Th_3=\,&N^{-3/2}\sum_{\substack{r\in P_H\\v\in P_L}}\sum_{\substack{q\in\retp,u\in\ret\\u\neq-q,r}}N^\ka\wh V(u/N^{1-\ka})\h_rb_{r+v}^*b_{-r+u}^*a_q^*a_{q+u}b_v,\\
		\Th_4=\,&-N^{-3/2}\sum_{\substack{r\in P_H\\v\in P_L}}\sum_{\substack{q\in\retp,u\in\ret\\u\neq-q,r}}N^\ka\wh V(u/N^{1-\ka})\h_rb_{r+v}^*b_{-r}^*a_q^*a_{q+u}b_{-u+v},\\
		\Th_5=\,&-N^{-3/2}\sum_{\substack{p\in P_H\\v\in P_L}}\sum_{r\in P_L\cup \{0\}}N^\ka\wh V((p-r)/N^{1-\ka})\h_rb_{p+v}^*b_{-p}^*b_v,\\
		\Th_6=\,&N^{-3/2}\sum_{\substack{r\in P_H\\v\in P_L}}\sum_{\substack{p\in P_L\\p\neq-v}}N^\ka\wh V((p-r)/N^{1-\ka})\h_rb_{p+v}^*b_{-p}^*b_v,\\
		\Th_7=\,&2N^2\sqrt N\l_l\sum_{r\in P_H,v\in P_L}(\wh \chi_l\star\wh f_N)_r\h_rb_{r+v}^*b_{-r}^*b_v\,.
	} 
This is \eqref{eq:comm_A_HN}, with $\d_{\cH_N} = \sum_{j=1}^7\Th_j$. Note that to isolate $\Th_0$ we need to use the scattering equation \eqref{eq:scattering_eta}, which also produces the last error term $\Th_7$. 
Switching to position space, it is simple to show that 
\[\begin{split}\pm \big(\Th_0+\hc) &= \pm \bigg( \int dxdy\, N^{5/2-2\kappa}V(N^{1-\ka}(x-y))\check{b}^*_x \check{b}^*_y \Big( \sum_{v\in P_L} e^{ivy}b_v \Big)+\hc\bigg)\\ & \leq  C \cV_N+ CN^{\kappa}(\cN_++1). \end{split}\]
For all other terms $\Th_1, \dots , \Th_7$, we show (\ref{eq:Th_j_worse_bound}) and (\ref{eq:deltaHN}). 
The contribution $\Th_1$ can be bounded by Cauchy-Schwarz, yielding
\[ \abs{\langle\x_1|\Th_1|\x_2\rangle}\le CN^{-1/2} \| (\cK_+1)^{1/2}(\cN_++1)^{1/2}\x_1\| \| (\cN_++1)^{1/2}\x_2 \| \,,\] which implies (\ref{eq:Th_j_worse_bound}) and (\ref{eq:deltaHN}).
To bound the quintic terms $\Th_{2}, \Th_3, \Th_4$, we switch to position space. We find that
\[\Th_2=N^{-1/2}\int_{\L^2} N^{2-2\ka}V(N^{1-\ka}(x-y))\ckb_x^*\ckb^*(\cke_{H,x})\cka_y^*\cka_y\ckb(\check \chi_{L,x})dxdy\,,
	\] where $\cke,\check\chi_{L,x}$ are defined by $\cke_{H,x}(y)=\cke_H(y-x)$, $\check\chi_{L,x}(y)=\check\chi_{P_L}(y-x)$. Using the bounds $\norm{\cke_{H,x}}_2=\norm{\h_H}_2\le CN^{\ka-\a/2}$, $\norm{\check\chi_{L,x}}_2=\norm{\chi_{P_L}}_2\le C N^{3\b/2}$ we obtain 	
\spl{
		\langle\x_1|\Th_2|\x_2\rangle\le\,& C N^{-1/2+\ka+(3\b-\a)/2}\int_{\L^2}N^{2-2\ka}V(N^{1-\ka}(x-y))\norm{\cka_x\cka_y\x_1}\norm{\cka_y(\cN_++1)\x_2}\\
		\le\,& CN^{-1+(3\ka+3\b-\a)/2}\norm{\cV_N\x_1} \|(\cN_++1)^{3/2}\x_2 \| \\
		\le\,& CN^{-1/2}\norm{\cV_N\x_1} \| (\cN_++1)^{3/2}\x_2 \| \,,
	}where we used the assumption $\a-\b>2\b+4\ka-1$ in the last step. This bound implies both \eqref{eq:Th_j_worse_bound} and \eqref{eq:deltaHN} for $\Th_2$.  The terms $\Th_{3}, \Th_4$ can be handled analogously. As for $\Th_{5}, \Th_6$, we use again the estimate 
	\[
	\sum_{p\in P_H}\fra{N^\ka\wh V((p-r)/N^{1-\ka})^2}{\abs p^2}\le CN
	\] which holds uniformly in $r\in P_L$. With Cauchy-Schwarz, we find that
	
	\[\abs{\langle\x_1|\Th_5|\x_2\rangle}\le CN^{-1/2} \|(\cK+1)^{1/2}(\cN_++1)^{1/2}\x_1 \| 
	\| (\cN_++1)^{1/2}\x_2 \| \,,	
	\] and, similarly, 
	\spl{
		\abs{\langle\x_1|\Th_6|\x_2\rangle}\le\,&N^{-3/2}\bigg(\sum_{\substack{r\in P_H\\v\in P_L}}\sum_{\substack{p\in P_L\\p\neq-v}}\fra{N^\ka\wh V((p-r)/N^{1-\ka})\abs{\h_r}}{\abs p^2}\norm{b_v\x_2}^2\bigg)^{1/2}\\
		&\times\bigg(\sum_{\substack{r\in P_H\\v\in P_L}}\sum_{\substack{p\in P_L\\p\neq-v}}N^\ka\wh V((p-r)/N^{1-\ka})\abs{\h_r}\abs p^2\norm{b_{p+v}b_{-p}\x_1}^2\bigg)^{1/2}\\
		\le\,& CN^{(-1+\b+2\ka)/2} \|(\cK+1)^{1/2}(\cN_++1)^{1/2}\x_1\|  \|(\cN_++1)^{1/2}\x_2 \|.
	} Finally, it follows from \eqref{eq:scattering_lambda}, \eqref{eq:chi_f_bound}, and $\a>4\ka$ that
	\[
	\abs{\langle\x_1|\Th_7|\x_2\rangle}\le N^{-1/2}\norm{(\cN_++1)\x_1} \| (\cN_++1)^{1/2}\x_2\| \,.
	\] 
This concludes the proof of the lemma.
\end{proof}

\begin{lemma}\label{lem:comm_Th0_A}
Let $\Th_0$ be defined as in (\ref{eq:Th0}). Then, we have 
\be\label{eq:comm_Th0_A}
[\Th_0,A]= \Pi_0 + \delta_{\Th_0} 
\ee
where
\begin{equation}\label{eq:Pi0}  
\Pi_0= -2N^{-1}\sum_{\substack{r\in P_H\\v\in P_L}}N^\ka \left(\wh V(r/N^{1-\ka})+\wh V((r+v)/N^{1-\ka})\right)\h_rb_v^*b_v, 
\end{equation} 
and 
\begin{equation}\label{eq:deltaXi} 
\pm \delta_{\Th_0} \leq C N^{-1/2}(\cK+\cN_+^2+1)(\cN_++1) \, .
\ee 	
Moreover, 
\be\label{eq:comm_Pi0_A_bound}
	\pm[\Pi_0,A]\le CN^{-1/2}(\cN_++1)^2\,.
\end{equation}
\end{lemma}

\begin{proof}
Proceeding as in \cite[Lemma 8.6]{BBCS4}, we obtain
\[ [ \Th_0, A] = \Pi_0+ \sum_{j=1}^8 \Big[\Pi_j+\hc\Big] + \delta \]
with $\Pi_0$ defined as in (\ref{eq:Pi0}), 
\[
\delta=-2N^{-1+\ka}\sum_{\substack{r\in P_H\\v\in P_L}} \chi_{P_H^c}(r+v)\wh V((r+v)/N^{1-\ka})\h_rb_v^*b_v
\]
 and 
\begin{align*}
\Pi_1=\,&-N^{-1}\sum_{\substack{r\in P_H,v\in P_L\\ r+v\in P_H}}N^\ka\wh V((r+v)/N^{1-\ka})\h_rb_v^*\left[\left(1-\fra{\cN_+}{N}\right)^2-1\right]b_v,\\
		&-N^{-1}\sum_{\substack{r\in P_H\\v\in P_L}}N^\ka\wh V(r/N^{1-\ka})\h_rb_v^*\left[\left(1-\fra{\cN_++1}{N}\right)\left(1-\fra{\cN_+}{N}\right)-1\right]b_v,\\
		\Pi_2=\,&-N^{-1}\sum_{\substack{r\in P_H\\v\in P_L}}\sum_{\substack{w\in P_L\\w-r-v\in P_H}}N^\ka\wh V((r+v-w)/N^{1-\ka})\h_rb_v^*\left(1-\fra{\cN_++1}{N}\right),\\
		&\qquad\qquad\qquad\qquad\times(b_{w-r-v}^*b_{-r}-\fra{1}{N}a^*_{w-r-v}a_{-r})b_w,\\
		\Pi_3=\,&-N^{-1}\sum_{\substack{r\in P_H\\v\in P_L}}\sum_{\substack{w\in P_L\\w+r\in P_H}}N^\ka\wh V((r+w)/N^{1-\ka})\h_rb_v^*\left(1-\fra{\cN_+}{N}\right),\\
		&\qquad\qquad\qquad\qquad\times(b_{w+r}^*b_{r+v}-\fra{1}{N}a^*_{w+r}a_{r+v})b_w,\\
		\Pi_4=\,&N^{-1}\sum_{\substack{r\in P_H\\v\in P_L}}\sum_{\substack{w\in P_L\\w-v\in P_H}}N^\ka\wh V((v-w)/N^{1-\ka})\h_rb_{r+v}^*b_{-r}^*\left(1-\fra{\cN_+}{N}\right)b_{w-v}^*b_w,\displaybreak \\
		\Pi_5=\,&N^{-2}\sum_{\substack{r\in P_H\\v\in P_L}}\sum_{\substack{s\in P_H\\w\in P_L}}N^\ka\wh V(s/N^{1-\ka})\h_rb_{v}^*(b_{-r}a_{s+w}^*a_{r+v}+a_{s+w}^*a_{-r}b_{r+v})b_{-s}^*b_w, \\
		\Pi_6=\,&-N^{-2}\sum_{\substack{r\in P_H\\v\in P_L}}\sum_{\substack{s\in P_H\\w\in P_L}}N^\ka\wh V(s/N^{1-\ka})\h_rb_{r+v}^*b_{-r}^*a_{s+w}^*a_{v}b_{-s}^*b_w,  \\
		\Pi_7=\,&-N^{-1/2}\sum_{\substack{r\in P_H\\v\in P_L}}N^\ka\wh V(r/N^{1-\ka})b_{r+v}^*[b_{-r}^*,A]b_v,\\
		\Pi_8=\,&-N^{-1/2}\sum_{\substack{r\in P_H\\v\in P_L}}N^\ka\wh V(r/N^{1-\ka})b_{r+v}^*b_{-r}^*[b_v,A]\,.
\end{align*}
This is \eqref{eq:comm_Th0_A}, with $\d_{\Th_0} = \d+\sum_{j=1}^8[\P_j+\hc]$. We proceed to bound the terms in the error to show (\ref{eq:deltaXi}).  The characteristic function in $\d$ vanishes for any $v$, if $\abs r>N^{\a}+N^\b$: in particular, 
\[
\expec{\x}{\d}\le CN^{-1+\ka}\expec{\x}{\cN}\Big(\sum_{|r|\le 2N^\a}|\h_r|\Big)\le CN^{-1+2\ka+\a}\expec\x\cN.
\]
The term $\Pi_1$ is easily bounded by Cauchy-Schwarz's inequality and $\sum_{r\in\retp} |\wh V(r/N^{1-\ka})\h_r| \le C N$: we find $\pm\Pi_1\le CN^{-1+\ka}(\cN_++1)^2\le CN^{-1/2}(\cN_++1)^2$. As for $\Pi_2$ we estimate 
	\spl{
		|\expec{\x}{\Pi_2}|\le\,& CN^{-1+\ka}\bigg(\sum_{\substack{r\in P_H\\v\in P_L}}\sum_{\substack{w\in P_L\\w-r-v\in P_H}}\abs {\h_r}^2\norm{b_{w-r-v}b_v\x}^2\bigg)^{1/2}\\
		&\qquad\times\bigg(\sum_{\substack{r\in P_H\\v\in P_L}}\sum_{\substack{w\in P_L\\w-r-v\in P_H}}\norm{b_{-r}b_ {-w}\x}^2\bigg)^{1/2}\\
		\le\,& CN^{-1+\ka}\norm{\h_H}_2\abs{P_L}^{1/2}\norm{(\cN_++1)^2\x}^2\\
		\le\,&  CN^{-1/2}\norm{(\cN_++1)^2\x}^2\,,} 
where we used $\a-\b>2\b+4\k-1$ in the last step. The terms $\Pi_{3}, \Pi_4$ can be bounded similarly. The terms $\Pi_{5}, \Pi_6$ are first rearranged in normal order, producing additional terms from commutators, and then they are bounded as we did for $\X_{5}, \X_6$ in the proof of Lemma \ref{lem:comm_CN_A}. We find $\Pi_{5}, \Pi_6\le CN^{-1/2}(\cN_++1)^3$. Expanding the commutator in  $\Pi_{7}$ we find terms similar to $\Pi_1,...,\Pi_6$, which are bounded in the same way. Finally, the term $\Pi_8$ is analogous to $\X_8$ in Lemma \ref{lem:comm_CN_A} and can be bounded similarly, using the kinetic energy operator. This yields $\Pi_8\le CN^{-1/2}(\cK+\cN_+^2+1)(\cN_++1)$. Collecting all bounds, we get \eqref{eq:deltaXi}. Finally, the bound \eqref{eq:comm_Pi0_A_bound} follows from the fact that $\Pi_0=-\Xi_0$ and from \eqref{eq:comm_X_0_A}.
\end{proof}
Applying the last two lemmas, we can now control the action of $A$ on $\cH_N$. 
\begin{lemma}\label{prop:conA_HN}
We have
\spl{
		\conA{\cH_N}=\,&\cH_N-N^{-1/2}\sum_{\substack{r\in P_H\\v\in P_L}}N^\ka\wh V(r/N^{1-\ka})\Big[b_{r+v}^*b_{-r}^*b_v+\hc\Big]\\
		&-N^{-1}\sum_{\substack{r\in P_H\\v\in P_L}}N^\ka\left(\wh V(r/N^{1-\ka})+\wh V((r+v)/N^{1-\ka})\right)\h_rb_v^*b_v+\cE_{\cH_N} \,,
	}with 
	\[
	\pm\cE_{\cH_N} \le CN^{-1/2+\b/2+2\ka}(\cK+\cN_+^2+1)(\cN_++1)\,.
	\]
\end{lemma}
\begin{proof}
	With Lemma \ref{lem:comm_HN_A} and Lemma \ref{lem:comm_Th0_A}, we find
	\spl{
		\conA{\cH_N} =\,&\cH_N+\int_0^1\consA{\Th_0}ds+\int_0^1\consA{\delta_{\cH_N}}ds\\	=\,&\cH_N+\Th_0+\int_0^1\int_0^{s_1}e^{-s_2A} \Pi_0e^{s_2A}ds_2ds_1\\
		&+\int_0^1\int_0^{s_1}e^{-s_2A} \delta_{\Th_0} e^{s_2A}ds_2ds_1+\int_0^1\consA{\delta_{\cH_N}}ds\\
		=\,&\cH_N+\Th_0+\fra 1 2 \Pi_0+\int_0^1\int_0^{s_1}\int_0^{s_2}e^{-s_3A} [\Pi_0,A]e^{s_3A}ds_3ds_2ds_1\\
		&+\int_0^1\int_0^{s_1}e^{-s_2A} \delta_{\Th_0} e^{s_2A}ds_2ds_1+\int_0^1\consA{\delta_{\cH_N}}ds\\
		=\,&\cH_N+\Th_0+\fra 12\Pi_0+\cE_{\cH_N},
	} and the claim follows by \eqref{eq:Th0} \eqref{eq:Pi0}, the bounds (\ref{eq:deltaHN}), (\ref{eq:deltaXi}) and Lemma \ref{lm:growNA}.
\end{proof}

\subsection{Proof of Proposition \ref{prop:JN}}

Bringing together the results of Lemma \ref{lem:QGN_A_conj}, Lemma \ref{lem:conA_CN} and Lemma 
\ref{prop:conA_HN}, we find that
\bes{\label{eq:conj_GN_step1}
	\conA{\cG_N}=\,&C_{\cG_N}+\cQ_{\cG_N}+\cC_N+\cH_N\\
	&+N^{-1}\sum_{\substack{r\in P_H\\v\in P_L}}N^\ka\left(\wh V(r/N^{1-\ka})+\wh V((r+v)/N^{1-\ka})\right)\h_rb_v^*b_v\\
	&-N^{-1/2}\sum_{\substack{r\in P_H\\v\in P_L}}N^\ka\wh V(r/N^{1-\ka})\Big[b_{r+v}^*b_{-r}^*b_v+\hc\Big]+\wt\cE_{\cJ_N}+\conA{\cE_{\cG_N}}\,,
} with 
\[
\pm\wt\cE_{\cJ_N}\le CN^{-\fra12+\fra\b2+2\ka}(\cK+\cN_+^2+1)(\cN_++1)\,.
\] 
We also notice that, by Equation \eqref{eq:GN_error_bounds} and Lemma \ref{lm:growNA}, we have
\[
\pm\conA{\cE_{\cG_N}}\le CN^{-\fra{1}{2}+5\ka/2+\a}(\cH_N+\cN_+^2+1)(\cN_++1)\,.
\]

We consider
\spl{ 
	\cC_N-N^{-1/2}&\sum_{\substack{r\in P_H\\v\in P_L}}N^\ka\wh V(r/N^{1-\ka})\Big[b_{r+v}^*b_{-r}^*b_v+\hc\Big]\\
	=\,&N^{-1/2}\sum_{\substack{r\in\retp,v\in P_L^c\\v\neq-r}}N^\ka\wh V(r/N^{1-\ka})b_{r+v}^*b_{-r}^*(\g_vb_v+\s_vb_{-v}^*)\\
	&+N^{-1/2}\sum_{\substack{r\in P_H^c,v\in P_L\\v\neq-r}}N^\ka\wh V(r/N^{1-\ka})b_{r+v}^*b_{-r}^*b_v+\hc\\
	=\,&Z_1+Z_2+\hc
} 
It is simple to check that
\[
\pm (Z_2+\hc)\le CN^{-\fra12+\fra\a2}(\cK+1)(\cN_++1).
\]
The term $Z_1$, on the other hand, can be bounded in two ways, leading to the estimates  \eqref{eq:EJN_bounds} and, respectively, \eqref{eq:EJN_bounds2}. In the first case, we have 
	\[ \begin{split} 
	|\expec{\x}{Z_1}|\le\,& N^{-1/2} \sum_{\substack{r\in\retp,v\in P_L^c\\v\neq-r}}N^\ka\fra{\wh V(r/N^{1-\ka})}{r}r\norm{b_{r+v}b_{-r}\x}\left(\norm{b_v\x}+\abs{\s_v} \|(\cN_++1)^{1/2}\x \| \right)\\
	\le\,& N^{-1/2-\b+\ka}\bigg(\sum_{\substack{r\in\retp,v\in P_L^c\\v\neq-r}}r^2\norm{b_{r+v}b_r\x}^2\bigg)^{1/2}\\
	&\hspace{2cm}\times\bigg(\sum_{\substack{r\in\retp,v\in P_L^c\\v\neq-r}}\fra{\wh V(r/N^{1-\ka})^2}{r^2}v^2\norm{b_{v}\x}^2\bigg)^{1/2}\\
	&+ N^{-1/2+\ka}\bigg(\sum_{\substack{r\in\retp,v\in P_L^c\\v\neq-r}}r^2\norm{b_{r+v}b_r\x}^2\bigg)^{1/2}\\
	&\hspace{2cm}\times\bigg(\sum_{\substack{r\in\retp,v\in P_L^c\\v\neq-r}}\fra{\wh V(r/N^{1-\ka})^2}{r^2}{\s_v^2}\| (\cN_++1)^{1/2}\x \|^2\bigg)^{1/2}\\
	\le\, & CN^{\fra\ka2-\b} \|(\cK+1)^{1/2}(\cN_++1)^{1/2}\x \|^2\\
	& +CN^{\fra\ka2-\fra\b2} \|(\cK+1)^{1/2}(\cN_++1)^{1/2}\x \| \| (\cN_++1)^{1/2}\x \|,\end{split}\]
In the second case, we switch to position space and find for any $\delta>0$ that
\[\begin{split}
\big| \langle \xi| Z_1|\xi\rangle\big|&\leq \int dxdy\, N^{5/2-2\kappa}V(N^{1-\ka}(x-y))|\Big| \langle\xi, \check{b}^*_x\check{b}^*_y\sum_{v\in P_L^c} e^{ivy} (\g_vb_v+\s_vb_{-v}^*)|\xi\rangle\Big|\\
&\leq \delta \langle\xi|\cV_N|\xi\rangle + CN^{\kappa}\delta^{-1}\langle\xi|(\cN_++1)|\xi\rangle.
\end{split}\]

Finally, we combine $\cQ_{\cG_N}$, as defined in (\ref{eq:QGN_def}), with the quadratic terms on the r.h.s. of \eqref{eq:conj_GN_step1}. Comparing with the definition of $\cQ_{\cJ_N}$ in (\ref{eq:CJN_QJN}), we obtain 
\[ \cQ_{\cG_N} + \cK + N^{-1}\sum_{\substack{r\in P_H\\v\in P_L}}N^\ka\left(\wh V(r/N^{1-\ka})+\wh V((r+v)/N^{1-\ka})\right)\h_rb_v^*b_v = \cQ_{\cJ_N} + \wt{\cE}_1 + \wt{\cE}_2 +\wt\cE_3\]
where  
\[
\wt\cE_1=\sum_{p\in\retp}p^2\left(a_p^*a_p-b_p^*b_p\right)=N^{-1}\sum_{p\in\retp}p^2a_p^*\cN_+a_p
\] is easily bounded by $\pm\wt\cE_1\le N^{-1}(\cK+1)(\cN_++1)$, and 
\[ \begin{split} \wt{\cE}_2 = &-N^{-1}\sum_{\substack{r\in P_H,v\in P_L^c}}N^\ka\left(\wh V(r/N^{1-\ka})+\wh V((r+v)/N^{1-\ka})\right)\h_rb_v^*b_v \\
&-N^{-1}\sum_{\substack{r\in P_H^C,v\in \retp}}N^\ka\left(\wh V(r/N^{1-\ka})+\wh V((r+v)/N^{1-\ka})\right)\h_rb_v^*b_v \\ &- 2N^{-1}\sum_{p\in P_H}\s_p^2N^\ka\left(\convo{\h}_p+\convo{\h}_0\right)b_p^*b_p \end{split} \]  can be bounded by $\pm \wt{\cE}_2 \leq C N^{\ka-2\beta} (\cK + 1)$ (note that the most dangerous term is the one on the first line).  Finally, the off-diagonal part of the error
\[ \wt{\cE}_3 = \sum_{p \in \L^*_+} (G_p - \Gamma_p) [ b_p^*b_{-p}^*+b_pb_{-p} ] \] can be estimated by 
\[ \begin{split}
&|\expec{\x}{\wt\cE_3}|=|\expec{\x}{\sum_{p\in\retp}(G_p-\G_p)[b_p^*b_{-p}^*+b_pb_{-p}]}|\\
&\le CN^{-1+\ka}\!\!\sum_{p\in P_H}\bigg(\sum_{q\in P_H^c} \abs{\wh V((p-q)/N^{1-\ka})}\abs{\h_q}+\!\!\!\sum_{q\in\retp}\abs{\s_p}\abs{\wh V(q/N^{1-\ka})}\abs{\h_q}\bigg)\norm{b_p^*\x}\norm{b_{-p}\x}\\
&\le  C(N^{-\fra12+\fra{3\ka}2+\a}+N^{2\ka-3\a/2})\| (\cK+1)^{1/2}\x \| \| (\cN_++1)^{1/2}\x \|,\end{split}\]
to get \eqref{eq:EJN_bounds}, or, without using the kinetic energy for the second term, by 
\[
|\langle\x,\wt\cE_3\x\rangle|\le CN^{-\fra12+\fra{3\ka}2+\a}\| (\cK+1)^{1/2}\x \| \| (\cN_++1)^{1/2}\x \|+CN^{2\ka-\a/2}\| (\cN_++1)^{1/2}\x \|^2
\] to get \eqref{eq:EJN_bounds2}.
Collecting the estimates from above and taking into account the assumptions on $\kappa,\alpha,\beta$, we conclude the proof of Proposition \ref{prop:JN}.

\end{document}